\documentclass[10pt,a4paper]{article}
\pdfoutput=1
\usepackage[dvipsnames]{xcolor}
\usepackage{amsmath,fourier,amssymb,amsthm,graphicx,epstopdf,tikz,hyperref,
	color,cases,float,mathdots,mathtools,ytableau}
\usepackage[all]{xy}
\usepackage[normalem]{ulem}

\usepackage[normalem]{ulem} 
\usepackage{soul}

\usetikzlibrary{shapes,snakes}
\usepackage{chngcntr}
\counterwithout{figure}{section}
\usepackage[T1]{fontenc}
\allowdisplaybreaks
\numberwithin{equation}{section}

\theoremstyle{plain}
\newtheorem {hypo}{\bf\hspace{-\parindent}Hypothesis}[section]

\newtheorem {prop}[hypo]{Proposition}
\newtheorem {lemma}[hypo]{Lemma}
\newtheorem {theo}[hypo]{Theorem}
\newtheorem {defin}[hypo]{Definition}
\newtheorem {cor}[hypo]{Corollary}
\newtheorem {conj}[hypo]{Conjecture} 

\theoremstyle{remark}
\newtheorem {rmk}[hypo]{Remark}
\newtheorem {eg}[hypo]{Example} 
\newcommand{\pf}{\begin{bpf}}

	\newcommand{\pfms}{\begin{bpfms}}
		\newcommand{\epf}{\end{bpf}\hfill$\square$\vspace{0.1cm}}
	\newcommand{\epfms}{\end{bpfms}\hfill$\square$\\ }
\newcommand\ben{\begin{equation*}}
	\newcommand\ebn{\end{equation*}}
\newcommand\beq{\begin{equation}}
	\newcommand\eeq{\end{equation}}
\newcommand\ds{\displaystyle}
\newcommand\lb{\left(}
\newcommand\rb{\right)}

\newcommand\PIeq{P$_\text{I}$ }
\newcommand\PVIeq{P$_\text{VI}$ }

\newcommand{\Res}{\mathop{\rm res}}

\begin{document}
\LARGE
	\noindent
	\textbf{Many-faced Painlev\'e I \,:\,   irregular conformal blocks, \\[+.3em] topological recursion, and holomorphic anomaly approaches
    }
	\normalsize
	\vspace{0.5cm}\\
	\textit{ 
	    N. Iorgov$\,^{a,b}$\footnote{iorgov@bitp.kyiv.ua},
	    K. Iwaki$\,^{c}$\footnote{k-iwaki@g.ecc.u-tokyo.ac.jp},
		O. Lisovyy$\,^{d}$\footnote{oleg.lisovyi@univ-tours.fr},
		Yu. Zhuravlov$\,^{a,b}$\footnote{ujpake@gmail.com}
		}
	\vspace{0.2cm}\\
	$^a$  Bogolyubov Institute for Theoretical Physics, 14B Metrolohichna st., 03143 Kyiv, Ukraine \\
        $^b$  Kyiv Academic University, 36 Vernadsky Ave., 03142 Kyiv, Ukraine          \\
	$^c$  The Graduate School of Mathematical Sciences, The University of Tokyo, 
	3-8-1 Komaba Meguro-ku Tokyo 153-8914, Japan  \\
	$^d$  Institut Denis-Poisson, Universit\'e de Tours, CNRS, Parc de Grandmont,
	37200 Tours, France


    \begin{abstract}
    In recent years, the Fourier series (Zak transform) structure of the Painlev\'e I tau function has emerged in multiple contexts.
		Its main building block admits several conjectural interpretations, 
		such as the partition function of an Argyres-Douglas gauge theory, 
		the topological recursion partition function for the Weierstrass elliptic curve, 
		and a 1-point conformal block on the Riemann sphere with an irregular insertion of rank $\frac52$. We review and further develop a mathematical framework for these constructions, 
		and formulate conjectures on their equivalence. In particular, we give a simple explanation of the Fourier series representation of the tau function based on the Jimbo-Miwa-Ueno differential extended to the space of Stokes data. 
        We provide an algebraic construction of the rank $\frac52$ Whittaker state for the Virasoro algebra  embedded into a rank $2$ Whittaker module, prove its existence and uniqueness, and fix its descendant structure. We also prove the conifold gap property of the relevant topological recursion partition function, which, on one hand, enables its efficient computation within the holomorphic anomaly approach and, on the other, establishes the existence of solution for the latter.
 	\end{abstract}

    \tableofcontents

	\section{Introduction}
	
	\subsection{Motivation}
	
	In the last decades, a number of relations has been discovered between integrable systems, 
	2D conformal field theory (CFT), 4D supersymmetric gauge theory, and topological strings. 
	An important class of these correspondences deals with isomonodromic systems, the simplest instances of which 
	are given by Painlev\'e equations \cite{FIKN}. The basic example is the sixth Painlev\'e equation \PVIeq 
	that describes monodromy preserving deformations of $\mathfrak{sl}_2$ Fuchsian systems 
	with four simple poles on the Riemann sphere. It was observed in \cite{GIL12} that the \PVIeq tau function 
	admits a Fourier series representation
	\beq\label{Kyiv}
	\tau\bigl( t\,\bigl|\,\nu,\rho\bigr)=\sum_{n\in \mathbb Z}\mathcal Z\bigl( t\,\bigl|\,\nu+n\bigr)\, e^{2\pi i n\rho},
	\eeq
	where $t$ denotes the isomonodromic time and $\nu,\rho$ are Fenchel-Nielsen type coordinates on 
	the two-dimensional space of \PVIeq monodromy data for fixed generic local monodromy. A series of the form \eqref{Kyiv} is called \textit{Zak transform}. Its coefficients are expressed in terms of a  function 
	$\mathcal Z=\mathcal Z\lb t\,|\,\nu\rb$ that has been identified with a $c=1$ Virasoro conformal block 
	and partition function of $\mathcal N=2$, $N_f=4$ $\mathrm{SU}\lb 2\rb$ gauge theory 
	in the self-dual $\Omega$-background by the Alday-Gaiotto-Tachikawa (AGT) correspondence \cite{AGT09}. 
	Importantly, it has an explicit convergent $t$-series representation in terms of Nekrasov functions.
	By now, there exist several derivations 
	 of the formula \eqref{Kyiv} from both CFT and gauge theory perspective, as well as rigorous proofs 
	\cite{ILT14,BS14,GL16,CGL17,Nek20,JN20}.
	A description of $\mathcal Z$ as (the 4D limit of) the topological string partition function 
	was also proposed in \cite{CPT23, CLT20}.
	These results connect several topics of modern mathematical physics  via Painlevé equations.  
	They have been extended in different directions, 
	including higher number of poles \cite{ILT14}, higher rank \cite{GIL18}, genus 1 \cite{BDGT19a,BDGT19b}
	and $q$-difference case \cite{BS16,JNS17}.
	
	%
	
    More recently, several representations similar to \eqref{Kyiv} have been proposed for the tau function of the first Painlevé equation,
	\begin{equation} \label{eq:PI}
    {\rm P}_{\rm I} ~:~ \frac{d^2q}{dt^2} = 6q^2 + t.
    \end{equation}
    In the \PIeq case, the building blocks $\mathcal Z$ are given by formal 
    (or asymptotic) power series in an appropriate scaling parameter. 
    They are known to arise in at least three different contexts, as shown in the Figure \ref{fig:duality-diagram}:
	\begin{itemize}
	\item  It was conjectured in \cite{BLMST}  from
		that $\mathcal Z$ coincides with the partition function of 
		a  particular 4D $\mathcal N=2$ SUSY gauge theory, called Argyres-Douglas (AD) theory $H_0$, in a self-dual background. 
		It is related to topological string theory, and hence, the AD partition functions 
		can be efficiently studied using the method of holomorphic anomaly equation (HAE).
		
		\item In \cite{Iwaki19}, a Fourier series representation of the \PIeq tau function was discovered in the framework of the topological recursion (TR)
		run on the simplest genus 1 curve --- the Weierstrass elliptic curve.
		The counterpart of $\mathcal Z$ is constructed as the perturbative TR partition function.
		
		\item Most recently, Poghosyan and Poghossian \cite{PP23} have introduced 
		a new class of  conformal blocks of the Virasoro algebra involving irregular states of rank $\frac52$, 
		and discovered that in the case of central charge $c=1$ the corresponding expansions match with  $\mathcal Z$ computed directly from \PIeq.
	\end{itemize}

	This raises a number of questions. 
	First of all, one should be able to give a simple proof of Zak representation
	\eqref{Kyiv} on the Painlevé side. 
	Second, it is important to compare the tau function construction within  the CFT, gauge theory and TR approach. Such comparison would be a key step towards the exploration of  (yet to be mathematically formulated) AGT correspondence for AD gauge theories. Moreover, it is necessary to address several important issues in the construction of the counterparts of $\mathcal Z$ in each of the three setups. 
	On the CFT side, one needs to prove the existence of Whittaker states of half-integer rank and the relevant irregular conformal blocks. 
	In the TR framework, one has to carry out concrete calculations of the free energy and 
	develop tools for extracting the asymptotics with respect to isomonodromic time $t$ from a formal series in the TR counting parameter $\hbar$. 
	Additionally, from a mathematical perspective, when comparing approaches based on TR
	and HAE, it is necessary to rigorously verify the conifold gap property of the corresponding partition function. 
	
	The present work aims to study these questions and formulate a precise correspondence between 
	the different faces of the \PIeq tau function.
	The aforementioned issues and obtained results are explained in further detail below.

	\vspace{+1.em}
	\begin{figure}[t]
		\[
	\xymatrix@C=40pt@R=50pt{
	  & \text{~~2D conformal field theory~~} \ar@/_20pt/[ldd] \ar[d]^{\small \cite{PP23}}  \ar@/^20pt/[rdd] &  \\
	 &\text{P$_{\rm I}$ tau functions}  &  \\ 
	 \text{~~4D gauge theory~~~} \ar[ur]_{\small ~~~~~ \cite{BLMST}} \ar@/^20pt/[ruu]  \ar@/_15pt/[rr] 
	 &  & \text{~~~~~topological recursion} \ar[ul]^{\small \cite{Iwaki19} ~~~~~ } \ar@/_20pt/[luu] \ar@/^15pt/[ll]
	}
	\]
	\vspace{+1.em}
	\caption{Painlevé I from different perspectives}
	\label{fig:duality-diagram}
	\end{figure}
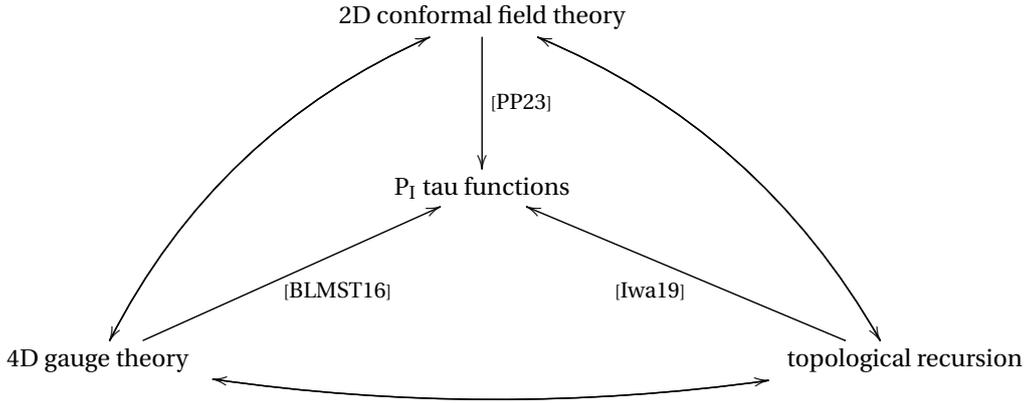

	\subsection{Zak transform from monodromy}
	It was shown in \cite{ILT14} that Zak representation of the \PVIeq tau function naturally follows from the relationship between the isomonodromic deformations and CFT. In the case of \PIeq, a related idea was explored in the TR framework \cite{Iwaki19}. 	In fact, the monodromy of the TR wave function
	(which can be considered as a counterpart of conformal block with degenerate field insertion) 
	is described using difference operators with respect to the parameter $\nu$. 
	The Fourier series arises 
	when this operator-valued monodromy is translated into the usual Stokes data of the linear system associated to \PIeq. On the gauge theory side, the structure \eqref{Kyiv}  turns out to be closely related to the so-called blow-up equations satisfied by the Nekrasov instanton partition functions.

    In the case of \PVIeq (as well as its degenerations to 
    P$_\text{V}$ and P$_\text{III}$), such a representation  can be proved rigorously using the explicit combinatorial series for the tau function deduced from its Fredholm determinant representation \cite{GL16,CGL17,LNR18}. What is missing, however, is a simple explanation based solely on the theory of monodromy preserving deformations.

	In Subsection \ref{subsec:Fourier-series-structure}, 
	we will show that Zak transform structure of \PIeq tau function is a straightforward consequence 
	of the existence of global log-canonical coordinates on the space $\mathcal S$ of \PIeq monodromy data 
	(cf. Proposition \ref{taufou}). 
	The key fact here is a one-to-one correspondence between the Stokes data  and solutions of \PIeq.
	Indeed, two relevant Stokes multipliers are expressed as $e^{2 \pi i \nu}$ and $e^{2 \pi i \rho}$, 
	which trivially makes the corresponding  \PIeq function $q=q\lb t\,|\,\nu,\rho\rb$ periodic in $\nu$ and $\rho$. We will use an extension of the Jimbo-Miwa-Ueno differential to $\mathcal S$ to promote this periodicity of $q$ to quasi-periodicity of the \PIeq tau function.
	It may be expected that the Fourier series expression for all other Painlevé tau functions
	can be proved in a similar manner.

		\subsection{Ramified irregular conformal blocks}
	First examples of irregular conformal blocks were introduced in \cite{G09} shortly after the discovery of the AGT correspondence \cite{AGT09}. On the gauge theory side, these correspond to theories with a lower number of flavors  $N_f<4$ in the weak coupling regime. In the isomonodromic context, such conformal blocks appear in the  description of tau functions of P$_\text{V}$ and P$_{\text{III}}$ equations \cite{GIL13} in the asymptotic regime where the conventional isomonodromic time parameter is small. Their construction involves, in addition to highest weight representations of the Virasoro algebra 
    \beq\label{Virasoro}
	 \left[L_m,L_n\right]=\lb m-n\rb L_{m+n}+\frac{c}{12}m\lb m^2-1\rb \delta_{m+n,0},
	\eeq
    to be denoted by $\mathsf{Vir}$, the so-called Whittaker modules $\mathcal V^{[r]}_{\Lambda}$ that are induced from the joint eigenstates $\bigl|I^{(r)}\lb\Lambda\rb\bigr\rangle$ of a suitably chosen subalgebra \cite{G09,BMT11}. 
		Here, $\Lambda$ denotes a tuple of complex parameters and $r\in \mathbb{N}/2$ is the rank of the Whittaker module. The latter quantity is the CFT counterpart of the Poincaré rank of irregular singularity. For instance, $\mathcal V^{[0]}_{\Delta}$ stands for the usual highest weight representation of $\mathsf{Vir}$ with conformal dimension $\Delta$, generated from the primary state $\bigl|I^{(0)}\lb\Delta\rb\bigr\rangle\equiv \bigl| \Delta\bigr\rangle$.

	The extension of isomonodromy/CFT correspondence to other Painlevé equations has been done in several other cases/asymptotic regimes, cf \cite{Nagoya15,Nagoya18,NU19}. However, it lacks systematic description and remains far from complete:
	\begin{itemize}
		\item First of all, it is not  easy to guess the algebraic construction of conformal block describing a given asymptotic regime.  Even in the P$_\text{V}$ case, where the associated linear system has only two regular singularities and one simplest irregular singular point of rank $1$, an algebraic definition of the conformal block corresponding to one of the two  strong coupling regimes (expansion (A.45) in \cite{BLMST}) is not known.
		\item For higher (integer) rank $r>1$, there appears an extra complication due to the absence of a canonical pairing between $\mathcal V^{[r]}_{\Lambda}$ and generic highest weight representation. This implies that the embedding of the irregular state $\bigl|I^{(r)}\lb\Lambda\rb\bigr\rangle$ into a Whittaker module of lower rank is not unique and in general involves an infinite number of free parameters, subsequently inherited by conformal blocks. A possible way to fix this ambiguity was proposed by Gaiotto and Teschner in \cite{GT12} using Feigin-Fuchs bosonization. In this approach, additional constraints on the Whittaker vectors come from a realization of the positive part $\mathsf{Vir}_+=\bigoplus_{n=0}^\infty \mathbb C L_n$ of the Virasoro algebra in terms of differential operators in $\Lambda$. It is conjectured that these constraints uniquely fix a chain of embeddings $\mathcal V^{[0]}_{\Delta}\supset \ldots \supset\mathcal V^{[r-1]}_{\Lambda'}\supset \mathcal V^{[r]}_{\Lambda}\supset\ldots$ involving only a finite number of new parameters at each step.
		\item Finally, for half-integer rank $r\ge \frac32$ corresponding to ramified irregular singularities, the analog of the latter setup was not available until very recent examples \cite{PP23,HNNT24}. One of the results of the present work is to provide, in Subsection~\ref{subsec_HIR}, an explicit 1st order differential operator realization of $\mathsf{Vir}_+$ for arbitrary half-integer rank, which turns out to be unique up to gauge equivalence.
	\end{itemize}

 The irregular conformal block describing the Painlevé I tau function was defined in \cite{PP23} as a projection of the rank $\frac52$ irregular state $\bigl|I^{(5/2)}\bigr\rangle$ onto the vacuum. Given that the  \PIeq linear system has a single ramified irregular singularity of Poincaré rank $\frac52$, this is totally in line with the above heuristics. The state $\bigl|I^{(5/2)}\bigr\rangle$ depends on a triple  $\Lambda=\lb \Lambda_3,\Lambda_4,\Lambda_5\rb$ of eigenvalues of $L_{3},L_4,L_5$ entering 
    the differential operator realization of $\mathsf{Vir}_+$. However, the irregular conformal block itself depends essentially only on a certain algebraic combination of these parameters. This follows from the Ward identities for global conformal symmetry and is analogous to the well-known fact that the regular 4-point conformal block  nontrivially depends only on the cross-ratio of fields positions. On the Painlevé side, the corresponding invariant plays the role of isomonodromic time. 

   In \cite{PP23}, the irregular state $\bigl|I^{(5/2)}\bigr\rangle$ is assumed  to be a series in $\Lambda_5$ with coefficients in a Whittaker module of rank~2. They are sought in the form of generalized descendants of the corresponding rank~2 irregular state, involving the action of Virasoro generators $L_k$ with $k<0$ and derivatives with respect to $\Lambda_3,\Lambda_4$. Neither the existence nor uniqueness of $\bigl|I^{(5/2)}\bigr\rangle$ has been proven. The general structure of descendants is also unclear; moreover, the ansatz employed in \cite{PP23} to determine them order by order turns out to be largely suboptimal since many of the corresponding coefficients vanish. At the technical level, this makes the computation of rank~$\frac52$ irregular state extremely tedious.

   To address these issues, we develop  an alternative algebraic approach described in Subsection~\ref{SubsecAlgCon}. We introduce a realization of the Virasoro algebra depending on  a single invariant parameter  $\varepsilon$ from the very beginning.
 The rank~$\frac 52$ irregular state is constructed by embedding it into a rank $2$ Whittaker module ${\mathcal V}^{[2]}_{\nu}$ depending on a  parameter $\nu\in\mathbb C$. 
This embedding has a form of a power  series in $\varepsilon$ that can be found from the action of $L_k$'s with $k\ge 3$, and a special linear combination $L_\varepsilon=L_2-2\varepsilon L_1+6\varepsilon^2 L_0$.
The equations following from the action of  $L_{k\ge3}$ do not contain differential operators at all and fix the irregular state 
uniquely up to an overall prefactor depending on $\varepsilon$. The precise statement is given in Theorem \ref{conjWh52} that constitutes the main result of Section \ref{sec_IrrCB}. 
The remaining unknown prefactor, in turn, is uniquely fixed by the action of  $L_\varepsilon$.
This two-step procedure clarifies and drastically simplifies the approach of  \cite{PP23}
allowing for an efficient computation of the corresponding irregular conformal block. 
The results turn out to be consistent with the expansions obtained from Painlevé I and the (refined) holomorphic anomaly equation.

	\subsection{Series expansion of TR partition function} 
	
	TR is a recursive algorithm that constructs analogs of correlation functions 
	and partition functions of matrix models for a given spectral curve, 
	which is a Riemann surface with some additional data. 
	In fact, through TR, deformations of the spectral curve are closely related to isomonodromic deformations. 
	Consequently, as mentioned above, the partition function of TR serves as a building block of the tau function. 
	See \cite{Iwaki19, EGF, MO19, EGFMO} for more details. 
	
	In general, the TR partition function is defined as a formal power series in a parameter $\hbar$, 
	and serves as the generating function of  free energies. 
	While $\hbar$ does not explicitly appear in the CFT approach nor in \PIeq equation,
	it can be introduced through a suitable rescaling of parameters. 
	Therefore, in comparing the TR partition function with
	the irregular conformal block or \PIeq tau function, which both admit asymptotic expansions in the isomonodromic time $t$, 
	it becomes necessary to determine the corresponding large $t$ expansion of the TR partition function. 
	
	In fact, in Subsections \ref{subsec_TRPI} and \ref{sec:conifold-gap}, 
	we will establish an algorithm that, in principle, allows us to compute the series expansion 
	of the TR partition function in $t$ to arbitrary order. 
	As a method for obtaining such an expansion, we provide two approaches: 
	one based on the explicit representation of the free energy in terms of quasi-modular forms 
	(Proposition \ref{prop_FgStructure} and Example \ref{eg:P2-P3}), 
	and the other based on the direct analysis of TR itself (Proposition \ref{prop:expansion-Wgn}). 
	The latter approach, in particular, has the advantage of determining the large $t$ expansion 
	without requiring explicit evaluation of the free energy. 
	Importantly, as far as computationally feasible, we have confirmed that the coefficients of 
	the obtained asymptotic expansion coincide with those of the irregular conformal block of \cite{PP23}.
	Proving the all-order match  remains a challenge for future work.

	The large $t$ expansion of the relevant TR partition function is also useful for verifying the conjectural formulas 
	regarding the resurgence property (as an $\hbar$-series) of the topological string partition function 
	proposed in \cite{gm22, IM24}. 
	In fact, we can confirm  consistency  of these conjectures with the description of the nonlinear Stokes phenomenon for \PIeq solutions   obtained in \cite{Kapaev2004}. 
	This will be discussed in Appendix \ref{appendix:conjecture-resurgence}.

	\subsection{Holomorphic anomaly and conifold gap property}
	The HAE, introduced in \cite{BCOV93v2}, is a general method to compute 
	the free energy of topological string theory. 
	The free energy of TR is known to satisfy the holomorphic anomaly equation \cite{EO07, emo}, 
	and this fact establishes its close relation with the higher genus B-model of topological string theory. 
	In particular, it has been firmly established as the remodeling theorem 
	(the Bouchard-Klemm-Mari\~no-Pasquetti conjecture) 
	that the Gromov-Witten invariants of toric Calabi-Yau 3-folds can be obtained by 
	TR applied to the mirror curve \cite{BKMP08, EO15, FLZ20}. 
	For these examples, comparisons with TR can be made by describing Gromov-Witten invariants 
	via graph counting using methods such as localization, the Givental formula, and the topological vertex. 
    Furthermore, as shown in \cite{HK06, HK09}, 
	the HAE is also effective as an alternative method for calculating the gauge theory partition functions. 
	From this perspective, \cite{BLMST}  identifies the AD partition function as the one constructed through 
	the HAE and compares it with the P$_{\rm I}$ tau function.
	
	The HAE determines the free energy up to the so-called holomorphic ambiguity, 
	and it is common to impose the conifold gap property as a boundary condition 
	when specifying this ambiguity \cite{gv-conifold, HK06}. 
	To perform a comparison between this approach and TR for spectral curves
	which may not be related to toric Calabi-Yau 3-folds, it is necessary to prove 
	that the free energy constructed via TR satisfies the conifold gap property. 
	However, to the best of our knowledge, aside from certain specific cases (\cite{Nor09, IKoT1, IKoT2}), such a proof is not available.
	As seen in Subsection \ref{subsec_TRPI}, the conifold gap property also plays a crucial role in determining 
	the normalization constant when comparing the results of gauge theory \cite{BLMST} with TR.
	
	In Theorem \ref{thm:coni-gap}, 
	as another main result of this paper, 
	we  provide a rigorous proof of the conifold gap property for 
	the Weierstrass elliptic spectral curve, which gives the tau function of P$_{\rm I}$.
	Our strategy is summarized as follows.
	The spectral curve includes the independent variable $t$ of P$_{\rm I}$ as a parameter, and 
	$t = \infty$ can be regarded as the conifold point. 
	At this point, a coalescence of ramification points occurs, 
	which may introduce singularities in the TR correlators. 
	However, by employing an alternative formulation of the recursion using the residue theorem 
	(as discussed in the proof of Proposition \ref{prop:expansion-Wgn}), this difficulty can be circumvented. 
	Furthermore, after applying an appropriate symplectic transformation and taking 
	the large $t$ limit, it can be shown that the Weber curve emerges. 
	By demonstrating that this limit commutes with TR, the conifold gap property is thereby proven.
	Details of the proof are provided in Appendix \ref{appendix:constant-term}. 

    As seen in Subsection \ref{subsec:Fixing-holomorphic-ambiguity}, 
	the conifold gap condition actually provides an overdetermined system for holomorphic ambiguity. 
	Discussions of the existence of solution of this system in  \cite{HKR08, HKPK11} 
	are based on physical considerations; however, it is not clear to us whether these results are applicable to our setting.
	Our  Theorem \ref{thm:sol-HAE-weak} and Corollary \ref{cor:existence-strong-coniforld-gap-HAE}
	guarantee the existence of solution to the HAE that satisfies the conifold gap property, 
	at least for the free energy of the Weierstrass elliptic spectral curve.
	%

	We also discuss the $\beta$-deformed conifold gap property to construct the 
	$\beta$-deformed partition function. The latter is a one-parameter deformation of the original partition function. 
	We prove that the $\beta$-deformed partition function is uniquely determined by imposing HAE and 
	the ``weak'' conifold gap, which requires  vanishing of a smaller number of coefficients than the usual conifold gap. 
	This $\beta$-deformed partition function is expected to coincide with irregular conformal block 
	for generic central charge  $c\ne1$, and this has indeed been verified in \cite{PP23} 
	(see also \cite{FMP23,FGMS23}).
	For reader's convenience, we will briefly review these results in Subsection~\ref{subsec_betaHAE}.
	
	\subsection{Comments on related works}
	
	It should also be noted that, independently of the study of the Painlev\'e equations, 
	the relations between the topics such as gauge theory, CFT, topological string theory 
	in both $A$-model and $B$-model, have already been studied by many authors.
	For example, in \cite{KPW10, DGH11, AFKMY12}, 
	the consistency of the Nekrasov partition functions, conformal blocks, TR, 
	and the topological vertex has been conjecturally confirmed. 
	More recently, several advances have been made in \cite{BBCC24, BCU24}, 
	such as the construction of Whittaker vectors based on the Airy structure,
	which is an algebraic reformulation of TR \cite{KS18}. 
	Our present work extends such comparisons to 
	Argyres-Douglas theories, irregular conformal blocks, and TR. 
	The aforementioned prior research deals with quantities that have a 
	"regular singular-type" series expansions, while here we focus on comparing quantities that 
	exhibit a "ramified irregular singular-type" behavior.
	See also \cite{BE19, FGMS23} for other related results. 
	
	During the preparation of this paper, two related works \cite{BGMT24, BST25} were posted on arXiv. 
	The paper \cite{BGMT24} proposed a new method for directly deriving the HAE
	from the bilinear equations satisfied by the tau function, without going through TR. 
	The paper \cite{BST25}, as an extension of the discussion in \cite{GMS20} on Painlev\'e III$_{3}$, 
	constructs the quantum Painlev\'e tau function for all Painlev\'e equations using 
	the Nekrasov partition function with general $\Omega$-background parameters. 
	These works are consistent with our results and further strengthen the correspondences shown in Figure \ref{fig:duality-diagram}.

\subsection{Further questions} 

Here, we outline a few potential directions for future research.
\begin{itemize}
\item 
It is natural to ask whether our results can be extended 
to other Painlev\'e equations, from P$_{\rm II}$ to P$_{\rm V}$.
\begin{itemize}
\item 
On the CFT side, it is worth investigating whether the algebraic approach developed 
in this paper can be extended as to cover all possible large isomonodromic time behaviors  presented in \cite{BLMST}.
\item
On the side of TR and HAE, describing the free energy in terms of quasi-modular forms 
associated with the elliptic curve, which arises as the classical limit of the corresponding isomonodromic linear system, 
would lead to a deeper understanding of the analogy between Painlev\'e transcendents and elliptic functions. 
We also expect that the conifold gap property can also be demonstrated for spectral curves associated 
with other Painlev\'e equations (\cite{EGF, MO19, EGFMO}) using a similar approach.
\item Different asymptotic behaviors of Painlevé functions can be attributed, from the perspective of TR, 
to the choice of cycles in the definition of the spectral curve, 
and from the perspective of CFT, to the choice of ``irregular pants decomposition''. 
It would be an intriguing challenge to mathematically formalize the latter concept. 
\end{itemize}

\item 
To our best knowledge, a definition of the AD partition function based on methods such as 
instanton counting is not known. 
As mentioned above, in \cite{BLMST} and \cite{PP23}, this quantity
is considered as defined by the HAE and the conifold gap property, under duality assumption. 
Adopting this definition, our result (Theorem \ref{thm:sol-HAE-weak} and Corollary \ref{cor:existence-strong-coniforld-gap-HAE})
rigorously proves the well-definedness of the AD partition function and 
its complete agreement with the TR partition function (for self-dual background). 
Providing a geometric definition of the AD partition function  consistent with 
the above definition is  an important challenge.

\item 
It should be possible to generalize our results to the Painlev\'e I hierarchy containing P$_\text{I}$ equation \eqref{eq:PI} as the first member. The second member of the hierarchy, denoted P$_\text{I}^{\text{2}}$, is also well-known. It appears in the universal description of the gradient catastrophy for systems such as KdV equation \cite{dub08}. It would be interesting to study the asymptotics of the isomonodromic tau function for such systems, which might be related to 1-point conformal block of irregular vertex operator of half-integer rank $r=\lb 2g+3\rb\bigl/2$, where $g>1$ corresponds to the genus of underlying spectral curve.

\item 
A rigorous proof, and even a heuristic derivation, of the \PIeq/CFT correspondence (Conjecture \ref{CFTPI}) remains an open problem. 
 One may attempt to develop an approach similar to \cite{ILT14}, 
 constructing solutions of the \PIeq linear system from  irregular conformal blocks with degenerate fields.
 In order to implement this idea, one needs  to explicitly describe the analytic continuation 
 of such irregular conformal blocks, particularly their Stokes structure.
 In this regard, the work \cite[Section 5]{Iwaki19} where explicit calculation of Stokes multipliers of the isomonodoromic system has been performed
 using TR and exact WKB analysis under various technical assumptions, could be instrumental. 
The method used in \cite[Section 5]{Iwaki19} is based on the so-called Voros connection formula, 
so it can be thought as an ``abelianization of BPZ equation''. 
It would also be related to the recent work \cite{HN24}.
Another approach worth considering is extending to rank $\frac52$ irregular setting the method of \cite{BS14}, 
where blow-up equations equivalent to \PVIeq were derived directly from the representation theory of the Virasoro algebra.

\item 
In Section \ref{subsec_betaHAE}, we have constructed the $\beta$-deformed partition function based on HAE. 
This quantity could potentially be obtained from 
the refined topological recursion \cite{CEM09, CEM11, KO23a, O24}. 
Another approach to  $\beta$-deformed partition functions is proposed in \cite{BST25}
based on blow-up equations.
It would be interesting to compare these approaches and prove 
the $\beta$-deformed conifold gap property in a manner analogous to Section \ref{sec:conifold-gap}, using in particular the results of \cite{KO23b}.

\end{itemize}

\subsection{Outline of the paper}

The paper is organized as follows.
In Section~\ref{sec_FT}, we introduce the linear system associated to Painlev\'e~I, the space of its Stokes data and the (extended) \PIeq tau function. The Zak transform structure of the tau function is proved in Proposition~\ref{taufou}. Section~\ref{sec_IrrCB} develops an algebraic construction of irregular conformal blocks relevant to \PIeq. Its main result, Theorem~\ref{conjWh52}, establishes existence and uniqueness of the corresponding rank~$\frac52$ irregular vector and fixes its algebraic structure. The algebraic expression of conformal block itself is given in \eqref{Ugenericc}--\eqref{L2comms}, and its relation to \PIeq is formulated in Conjecture~\ref{CFTPI}. Proposition~\ref{propLs} provides a differential operator realization of the positive part of the Virasoro algebra for arbitrary half-integer rank. 
Section~\ref{sec_TR} is devoted to the TR approach. To help readers who are new to it, Subsection~\ref{subsec_warmup} introduces the TR formalism with a simpler example of rational spectral curve corresponding to tronqu\'ees solutions of \PIeq. Subsection~\ref{subsec_Weierstrass} implements the TR  on the Weierstrass elliptic curve and determines the general structure of TR correlators and free energies in this case, see Propositions~\ref{prop:pre-quasi-modular} and~\ref{prop_FgStructure}. Subsection~\ref{subsec_TRPI} explains the TR/\PIeq correspondence, namely, how the large-time asymptotic series of the generic \PIeq tau function on the canonical rays can be retrieved from the formal $\hbar$-expansion of the non-perturbative TR partition function. The conifold gap property of the TR free energies is proved in Theorem~\ref{thm:coni-gap}. 
This result is further used in Section~\ref{sec_HAE}, devoted to the holomorphic anomaly approach, to demonstrate existence of solutions of HAE. 
Uniqueness of HAE solutions is established in Theorem~\ref{thm:sol-HAE-weak}. The final Subsection~\ref{subsec_betaHAE} deals with a $\beta$-deformation of the HAE; in particular, Conjecture~\ref{CFTHAEcorr} brings the discussion full circle by relating the corresponding ($\beta$-deformed) partition function to irregular conformal block (with arbitrary central charge $c$) introduced in Section~\ref{sec_IrrCB}. 

\subsection*{Acknowledgements}
The authors are grateful to 
Vincent Bouchard, 
Pierrick Bousseau, 
Tom Bridgeland, 
Bertrand Eynard, 
Pavlo Gavrylenko,
Alessandro Giacchetto, 
Paolo Gregori, 
Qianyu Hao, 
Shinobu Hosono, 
Akishi Ikeda, 
Hiroshi Iritani, 
Reinier Kramer, 
Hajime Nagoya, 
Andrew Neitzke,
Kento Osuga,
Masa-Hiko Saito, 
Nobuo Sato, 
Piotr Su\l{}kowski, 
Atsushi Takahashi
and
Koji Tasaka, 
for valuable discussions. 
We also would like to thank organizers of 
OIST Theoretical Sciences Visiting Program ``Exact Asymptotics: From Fluid Dynamics to Quantum Geometry'', 
and SRS Workshop ``Quantisation of moduli spaces from different perspectives''. 
This work was also supported by the Research Institute for Mathematical Sciences, 
an International Joint Usage/Research Center located in Kyoto University.
The work of KI was supported by
JSPS KAKENHI Grand Numbers 
21H04994, 22H00094, 23K17654, 24K00525.
The work of NI and YZ was supported by Team grant ``Random matrix models: from biomolecules to topological recursions'' from the Foundation for Polish Science.
NI and YZ also acknowledge support from the Simons Foundation.

	\section{Tau function quasi-periodicity\label{sec_FT}}
	\subsection{Extended tau function}
	Recall that the first Painlev\'e equation describes isomonodromic deformations of the rank 2 linear system with a single irregular point of Poincar\'e rank $\frac52$ on the Riemann sphere. Such a system can be transformed  into a canonical form
	\beq\label{linsys}
	\frac{d\Phi}{dz}=A\lb z\rb\Phi, \qquad 
	A\lb z\rb=\lb\begin{array}{cc} -p & z^2+q z+q^2+\frac{t}{2} \\
	4\lb z-q\rb & p \end{array}\rb,
	\eeq
	which corresponds to placing the pole of $A\lb z\rb dz$ at $\infty$, and using the gauge freedom to set the non-diagonalizable coefficient of the most singular term $z^2dz$ equal to $\lb\begin{array}{cc}0 & 1 \\ 0 & 0\end{array}\rb$ and to make the coefficient of $zdz$ off-diagonal. 
	
	A shortcut connection between \eqref{linsys} and \PIeq can be made via the characteristic equation ${\operatorname{det}\lb A\lb z\rb-\lambda\rb=0}$. The latter defines an elliptic curve written in the Weierstrass form,
	\beq\label{spcurve}
	\lambda^2=4z^3+2tz+2H,
	\eeq
	where 
	\beq\label{hamPI}
	H=\frac{p^2}{2}-2q^3-tq.
	\eeq 
	The equations 
	\beq\label{hameqs}
	\ds\frac{dq}{dt}=\ds\frac{\partial H}{\partial p},\qquad  \ds\frac{dp}{dt}=-\ds\frac{\partial H}{\partial q}
	\eeq are equivalent to Painlev\'e~I.  The $g_2$ and $g_3$ invariants of the curve \eqref{spcurve} play the respective roles of \PIeq time and Hamiltonian. 
	
	Let $\mathcal S$ denote the space of P$_\text{I}$ initial conditions parameterized by a pair of complex local coordinates $\mu_a,\mu_b$. The symplectic 2-form
	\beq\label{sympform}
	\Omega=d_{\mathcal S}p\wedge d_{\mathcal S}q=\lb p_{\mu_a}q_{\mu_b}-q_{\mu_a}p_{\mu_b}\rb d\mu_a\wedge d \mu_b
	\eeq
	is obviously closed on $\mathcal S$ and does not depend on time. It will be convenient to choose $\mu_{a,b}$ as Darboux type coordinates so that 
	\beq
	\Omega=2\pi i \,d\mu_a\wedge d\mu_b.
	\eeq 
	\begin{defin}
		The extended \PIeq tau function $\tau\lb t\,|\,\mu_a,\mu_b\rb$ on $\mathbb C\times\mathcal S$ is defined by
		\beq\label{exttau}
		d\ln \tau= H dt+Q_ad\mu_a+Q_b d\mu_b-2\pi i \mu_a d\mu_b,
		\eeq
		with $H$ and $Q_{a,b}$ given respectively by \eqref{hamPI}  and
		\beq
		5Q_k=4tH_{\mu_k}+3pq_{\mu_k}-2qp_{\mu_k},\qquad k=a,b.
		\eeq
	\end{defin}
	
	The first contribution $Hdt$ in \eqref{exttau} is the classical Jimbo-Miwa-Ueno definition \cite{JMU81} of the \PIeq tau function for fixed initial conditions. Its extension to $\mathcal S$ was obtained in \cite{LR16} based on the ideas of \cite{Ber09,IP15} and generalized to arbitrary isomonodromic systems in \cite{ILP16}. Note, however, that the closedness of the 1-form  on the right of \eqref{exttau} can be verified by  plain differentiation. Adding to it a closed 1-form on $\mathcal S$ modifies the tau function by a time-independent prefactor.
	
	\subsection{Zak transform structure}
    \label{subsec:Fourier-series-structure}

	The relation of P$_\text{I}$ to isomonodromy, that will be further discussed in Subsection~\ref{subsec_monodromy}, identifies $\mathcal S$ with the two-dimensional space of Stokes data of the system \eqref{linsys}.
	For now, it suffices to note that there are natural choices of coordinates for which the P$_\text{I}$ function $q=q\lb t\,|\,\mu_a,\mu_b\rb$ is periodic in $\mu_{a,b}$:
	\beq
	q\lb t\,|\,\mu_a+1,\mu_b\rb=q\lb t\,|\,\mu_a,\mu_b+1\rb=q\lb t\,|\,\mu_a,\mu_b\rb.
	\eeq
    The underlying reason is that $p$ and $q$ are uniquely determined by the Stokes coefficients and the latter turn out to be log-canonical coordinates with respect to the sympectic structure on $\mathcal S$ defined by $\Omega$. One can thus expand $q$ into a double Fourier series,
    \beq
    q\lb t\,|\,\mu_a,\mu_b\rb=\sum_{k,l\in \mathbb Z}q_{kl}\lb t\rb e^{2\pi i\lb k\mu_a+l\mu_b\rb}.
    \eeq
    and also write similar expansions for $p$ and $H$. 
    
    
    On the other hand, the \PIeq  tau function $\tau\lb t\,|\,\mu_a,\mu_b\rb$ cannot be made periodic simultaneously in $\mu_a$ and~$\mu_b$. Indeed, while the coefficients $H$ and $Q_{a,b}$ on the right of \eqref{exttau} are periodic, the last term $2\pi i \mu_a d\mu_b$ does represent an obstruction. It implies that, redefining 
    \beq\label{normtau}
    \tau\lb t\,|\,\mu_a,\mu_b\rb\mapsto e^{i\lb\alpha \mu_a+\beta \mu_b\rb}\tau\lb t\,|\,\mu_a,\mu_b\rb
    \eeq 
    with some constants $\alpha,\beta\in\mathbb C$ if necessary, one has
    \begin{subequations}\label{perrelstau}
    \begin{align}
    	\tau\lb t\,|\,\mu_a+1,\mu_b\rb=&\,e^{-2\pi i  \mu_b}\tau\lb t\,|\,\mu_a,\mu_b\rb,\\
    	\qquad \tau\lb t\,|\,\mu_a,\mu_b+1\rb=&\,\tau\lb t\,|\,\mu_a,\mu_b\rb.
    \end{align}
    \end{subequations}
    It follows from the second of these equations that $\tau\lb t\,|\,\mu_a,\mu_b\rb=\sum\limits_{n\in\mathbb Z}\tau_n\lb t\,|\,\mu_a\rb e^{2\pi i n \mu_b}$. The first equation then implies that $\tau_{n+1}\lb t\,|\,\mu_a\rb=\tau_{n}\lb t\,|\,\mu_a+1\rb$. These considerations can be summarized into the following
    \begin{prop}\label{taufou}
    	Let $\mu_a,\mu_b$ be any pair of Darboux coordinates  on $\mathcal S$ with respect to which the \PIeq function $q\lb t\,|\,\mu_a,\mu_b\rb$ is periodic. Then, up to the change of normalization \eqref{normtau}, the tau function defined by \eqref{exttau} can be written as
    	\beq
    	\tau\lb t\,|\,\mu_a,\mu_b\rb=\sum_{n\in\mathbb Z}\mathcal T\lb t\,|\,\mu_a+n\rb e^{2\pi i n \mu_b}.
    	\eeq
    \end{prop}
    
    The above result provides an explanation of the Fourier series structure of the expansions of Painlev\'e tau functions in different asymptotic regimes \cite{BLMST}. Moreover, it defines $\mathcal T\lb t\,|\,\mu_a\rb$ as a genuine function of~$t$ instead of an asymptotic series. Just as $\tau\lb t\,|\,\mu_a,\mu_b\rb$ itself, $\mathcal T\lb t\,|\,\mu_a\rb$ is holomorphic in the entire complex $t$-plane.
    
    \begin{rmk}
    	Given that $H$ is periodic in $\mu_a,\mu_b$, the lack of periodicity of its antiderivative with respect to $t$ may at first sight look surprising. Let us illustrate this phenomenon with a classical example of the theta function with characteristics. Denote
    	\beq
    	\theta\lb t\,|\,\mu_a,\mu_b\rb:
    	=\sum\limits_{n\in\mathbb Z}e^{i\pi \lb\mu_a+n\rb^2 t+2\pi i n \mu_b}=e^{i\pi\mu_a^2t}\vartheta_3\lb \mu_a t+\mu_b\,|\, t\rb,
    	\eeq
    	where $t$ takes values in the upper half-plane $\mathbb H$ and $\vartheta_3\lb z\,|\,t\rb$ is the Jacobi theta function. Obviously,  $\theta\lb t\,|\,\mu_a,\mu_b\rb$  satisfies the quasi-periodicity relations \eqref{perrelstau}, yet its logarithmic derivative $\partial_t\ln \,\theta\lb t\,|\,\mu_a,\mu_b\rb$ is periodic in both $\mu_a$ and $\mu_b$. Let us mention that, up to minor redefinitions, $\theta\lb t\,|\,\mu_a,\mu_b\rb$ coincides with the tau function associated to the Picard family of elliptic solutions of the Painlev\'e VI equation.
    \end{rmk}
	
	\subsection{Stokes data and  $\mathcal Z\lb s\, | \,\nu\rb$\label{subsec_monodromy}}
    The linear system \eqref{linsys} admits a formal fundamental matrix solution of the form
    \beq\label{sform}
    \Phi_{\text{form}}\lb z\rb=\lb\begin{array}{cc} \tfrac12 z^{1/4} & \tfrac12 z^{1/4} \\
    	z^{-1/4} & -z^{-1/4}
    \end{array}\rb\lb\mathbf 1+\sum_{k=1}^\infty g_k z^{-k/2}\rb\, \exp\lb
    \begin{array}{cc}\Theta\lb z\rb & 0 \\ 0 & -\Theta\lb z\rb\end{array}\rb,
    \eeq
	with $\Theta\lb z\rb=\frac45 z^{5/2}+tz^{1/2}$ and the matrix coefficients $g_k$ uniquely determined by \eqref{linsys}. The series in the above formula is divergent. The actual solutions can only be asymptotic to $ \Phi_{\text{form}}\lb z\rb$ as $z\to\infty$ while staying  inside the Stokes sectors
	\beq
	\Omega_k=\left\{z\in\mathbb C:\frac{\lb 2k-3\rb\pi}{5}< \arg z<\frac{\lb 2k+1\rb\pi}{5}\right\},\qquad k=1,\ldots,5.
	\eeq
	
	Introduce five such canonical solutions uniquely defined by the asymptotic condition $\Phi_k\lb z\rb\simeq \Phi_{\text{form}}\lb z\rb$ in~$\Omega_k$. They are related by  constant Stokes matrices $V_k=\Phi_k\lb z\rb^{-1}\Phi_{k+1}\lb z\rb$ that can be parameterized as
	\beq
	V_1=\lb\begin{array}{cc} 1 & iv_3 \\ 0 & 1\end{array}\rb,\quad 
	V_2=\lb\begin{array}{cc} 1 & 0 \\ iv_1 & 1\end{array}\rb,\quad
	V_3=\lb\begin{array}{cc} 1 & iv_4 \\ 0 & 1\end{array}\rb,\quad 
	V_4=\lb\begin{array}{cc} 1 & 0 \\ iv_2 & 1\end{array}\rb,\quad
	V_5=\lb\begin{array}{cc} 1 & iv_0 \\ 0 & 1\end{array}\rb.
	\eeq
	Their triangular form follows from the consistency of the asymptotics of $\Phi_k$ and $\Phi_{k+1}$ on the overlap ${\Omega_k\cap\Omega_{k+1}}$. Since the system \eqref{linsys} has no finite singularities in the $z$-plane, the monodromy of its fundamental solution  around~$\infty$ is trivial, implying the constraint
	\beq\label{nomon}
	V_1V_2V_3V_4V_5=\lb\begin{array}{cc}0 & i \\
	i & 0\end{array}\rb.
	\eeq 
	The matrix on the right  appears due to the presence of fractional powers of $z$ in $\Theta\lb z\rb$ and the leftmost prefactor in \eqref{sform}. The matrix identity \eqref{nomon} is equivalent to scalar relations
	\beq\label{stokesrec}
	v_k=1-v_{k-1}v_{k+1},\qquad k\in\mathbb Z/5\mathbb Z,
	\eeq
	which imply that there are at most two independent monodromy parameters. For example, for $v_0v_1\ne 0$ we can express the other parameters in terms of $v_0$, $v_1$ as
	\beq
	v_2=\frac{1-v_1}{v_0},\qquad v_3= \frac{v_0+v_1-1}{v_0v_1},\qquad v_4=\frac{1-v_0}{v_1}.
	\eeq
	On the other hand, if e.g. $v_1=0$, then the relations \eqref{stokesrec} imply that $v_0=v_2=v_3+v_4=1$ and we obtain a one-dimensional stratum of Stokes data on which $v_0$ ceases to be a suitable  coordinate.
	
	Any choice of a 5-tuple of Stokes coefficients $\mathbf v=\lb v_1,\ldots,v_5\rb\in\mathbb C^5$ satisfying \eqref{stokesrec} uniquely defines a system of the form \eqref{linsys}. In this way, $q$ and $p$ become functions of $t$ and $\mathbf v$ that can be shown to verify the deformation equations \eqref{hameqs}. The space of Stokes data may thus interpreted as the space $\mathcal S$ of \PIeq initial conditions.
	
	Assuming the genericity condition $v_0v_1\ne0$, let us choose a pair of local coordinates $\lb\nu,\rho\rb\in\mathbb C^2$ on $\mathcal S$ by
	\beq
	v_0=e^{2\pi i \nu},\qquad v_1=e^{2\pi i \rho}.
	\eeq
	By construction, the \PIeq function $q=q\lb t\,|\,\nu,\rho\rb$ remains invariant under the shifts $\nu\mapsto\nu+1$ and
	$\rho\mapsto\rho+1$.  A less obvious property proved in \cite[Proposition 3.4]{LR16} is that $\lb\nu,\rho\rb$ are canonical Darboux coordinates with respect to the symplectic form \eqref{sympform}, that is,  $\Omega=2\pi i d\nu\wedge d\rho$. Proposition~\ref{taufou} then implies that the \PIeq tau function normalized as in \eqref{exttau} can be written in the form
	\beq\label{taufourier2}
	\tau\lb t\,|\,\nu,\rho\rb=\sum_{n\in\mathbb Z}\mathcal T\lb t\,|\,\nu+n\rb e^{2\pi i n\rho}.
	\eeq
	
	Starting from the known leading asymptotics of $q=q\lb t\,|\,\nu,\rho\rb$ as $t\to-\infty$ along the negative real axis \cite{Kapaev88,Kapaev93,Takei},  one can compute the corresponding asymptotic expansion of the building block $\mathcal T\lb t\,|\,\nu\rb$. The result is given by \cite{BLMST}
	\begin{gather}\label{bigtau1}
		\mathcal T\lb t\,|\,\nu\rb\simeq \tilde{\mathcal Z}\lb s\,|\,\nu\rb,\qquad  
		\tilde{\mathcal Z}\lb s\,|\,\nu\rb=C\lb \nu\rb s^{-\frac{1}{60}-\frac{\nu^2}{2}}e^{\frac{s^2}{45}+\frac45 i\nu s}
		\mathcal Z\lb s\,|\,\nu\rb, \\
		\label{bigtau2} C\left(\nu\right)=48^{-\frac{\nu^2}{2}}
		\lb 2\pi \rb^{-\frac{\nu}{2}}e^{-\frac{i\pi \nu^2}{4}} G\lb 1+\nu\rb,\qquad s=24^{\frac14}\lb  -t\rb^{\frac54},\\
		\label{Zasseries}\mathcal Z\lb s\,|\,\nu\rb=1+\sum_{k=1}^\infty\frac{\mathcal Z_k\lb\nu\rb}{s^k},
	\end{gather}
	where $G\lb \nu \rb$ is the Barnes G-function and the coefficients $\mathcal Z_k\lb \nu\rb$ are polynomials in $\nu$  of degree $3k$ and parity $\lb-1\rb^k$ that can be systematically computed from \PIeq equation. For instance,
		\begin{subequations}
	 \begin{align}
		\mathcal Z_1\lb\nu\rb=&\,-\frac{i\nu\left(94\nu^2+17\right)}{96},\\
		\mathcal Z_2\lb\nu\rb=&\,-\frac{44180\nu^6+170320\nu^4+
			74985\nu^2+1344}{92160},\\
		\mathcal Z_3\lb\nu\rb=&\,\frac{i\nu\lb
			4152920 \nu^8+ 45777060 \nu^6
			+ 156847302 \nu^4 + 124622833 \nu^2+
			13059000   \rb}{26542080}.
	\end{align} 
\end{subequations}

	A surprising property of the asymptotic series $\mathcal Z\lb s\,|\,\nu\rb$ is a simplification of its logarithmic version. The coefficients $\mathcal E_k\lb\nu\rb$ of 
	\beq\label{PIFEnergy}
	\mathcal E\lb s\,|\,\nu\rb=\ln \mathcal Z\lb s\,|\,\nu\rb=\sum_{k=1}^\infty\frac{\mathcal E_k\lb\nu\rb}{s^k}
	\eeq
	are polynomials in $\nu$ of degree only $k+2$ instead of the naively expected $3k$. It will be useful to record here a few of these coefficients for subsequent comparison:
	\begin{subequations}
		\label{Epsilons}
	\begin{align}
	\label{Eps1}	\mathcal E_1\lb\nu\rb=&\,\mathcal Z_1\lb\nu\rb=-\frac{i\nu\left(94\nu^2+17\right)}{96},\\
		\mathcal E_2\lb\nu\rb=&\,-\frac{ 38585 \nu ^4+18385 \nu ^2+336}{23040},\\
		\mathcal E_3\lb\nu\rb=&\, \frac{i \nu  \left(5326258 \nu ^4+5019530 \nu ^2+541269\right)}{1105920 },\\
		\mathcal E_4\lb\nu\rb=&\, \frac{31386901 \nu ^6+50078328 \nu ^4+14438609 \nu ^2+188160}{1769472},\\
	\label{Eps5}	
		\mathcal E_5\lb\nu\rb=&\, -\frac{i \nu  \left(32160819372 \nu ^6+78779679122 \nu ^4+45102923992 \nu ^2+3752724735\right)}{424673280}.
	\end{align}
	\end{subequations}
    The formal series $\mathcal Z\lb s\,|\,\nu\rb$ and $\mathcal E\lb s\,|\,\nu\rb$ will be called \PIeq \textit{partition function} and \textit{free energy}.

	\begin{rmk}
		Since $G\lb -m\rb=0$ for $m\in\mathbb N$, setting $\nu=0$ eliminates all contributions to the sum \eqref{taufourier2} with negative $n$. Further setting $\rho=i\infty$ eliminates the contributions with $n>0$, so that we are left with
		\beq
		\tau\lb t\,|\,0,i\infty\rb\simeq s^{-\frac1{60}}e^{-\frac{s^2}{45}}\mathcal Z\lb s\,|\,0\rb,
		\eeq
		and the corresponding free energy is given by
		\beq\label{PIFE}
		\mathcal E\lb s\,|\,0\rb=\sum_{k=1}^{\infty}\mathcal E_{2k}\lb 0\rb s^{-2k}=-\frac{7 }{480}s^{-2}+\frac{245 }{2304}s^{-4}-\frac{259553 }{92160}s^{-6}+\ldots
		\eeq
		Somewhat counter-intuitively, $\mathcal Z\lb s\,|\,0\rb$  gives the asymptotic expansion of a \textit{one-parameter family} of \PIeq tau functions   consisting of the so-called \textit{tronqu\'ees} solutions. These correspond to the one-dimensional stratum of $\mathcal S$ with $v_1=0$ described above. The asymptotics of different tronqu\'ees tau functions  (characterized by different values of $v_3$) differ by exponentially small corrections.
	\end{rmk}

	\section{Irregular conformal blocks\label{sec_IrrCB}}
	The aim of this section is to elaborate on a remarkable conjecture of Poghosyan and Poghossian \cite{PP23} that relates the partition function $\mathcal Z\lb s\,|\,\nu\rb$ appearing in the previous sections to a special irregular conformal block of the Virasoro algebra with central charge $c=1$. We develop a drastically simplified algebraic construction of the rank $\frac52$ irregular state and the corresponding irregular conformal block. We also answer a question raised in \cite{HNNT24} on the explicit differential operator realization of the positive part of the Virasoro algebra relevant for constructing irregular states of arbitrary half-integer rank.
	
	\subsection{Gaiotto-Teschner setup}\label{subsec_IrrStates}
	Consider the Heisenberg algebra $\mathsf{H}$ with generators $a_n$ (${n\in\mathbb Z}$) and relations
	\beq
	\left[a_m,a_n\right]=\frac{m}{2}\delta_{m+n,0}.
	\eeq
	The Feigin-Fuchs bosonization of the Virasoro algebra \eqref{Virasoro} is the embedding $\varphi:U\lb\mathsf{Vir}\rb\xhookrightarrow{} U\lb \mathsf{H}\rb $ defined by
	\begin{subequations}
		\begin{align}
			&\varphi\lb L_n\rb=\sum_{k\in\mathbb Z}a_ka_{n-k}+i\lb n+1\rb Q a_n,\qquad n\ne 0,\\
			&\varphi\lb L_0\rb=2\sum_{k>0}a_{-k}a_k+a_0\lb a_0+iQ\rb,
		\end{align}
	\end{subequations}
	where the Virasoro central charge is expressed as  $c=1+6Q^2$. It was realized in \cite{GT12} that the antihomomorphism  	
	\beq
	a_0\mapsto -ic_0,\qquad a_{n}\mapsto -ic_n,\qquad a_{-n}\mapsto \frac{in}{2}\frac{\partial}{\partial c_n},\qquad n>0
	\eeq
	can be used to construct a family of antihomomorphisms of $\mathsf{Vir}_+=\bigoplus\limits_{n\ge0}\mathbb{C}L_n$, indexed by a non-negative integer~$r$, into the algebra of differential operators in $c_1,\ldots, c_r$:
	\beq\label{virdiff1}
	L_n\mapsto \mathcal{L}_n^{(r)}=\begin{cases}
		0,\qquad  & n> 2r,\\
		-\sum\limits_{k=n-r}^r c_k c_{n-k}+\lb r+1\rb Q\,\delta_{n,r} c_r ,\qquad & r\leq n\leq 2r, \\
		-\sum\limits_{k=0}^n c_k c_{n-k}+\lb n+1\rb Qc_n+\sum\limits_{k=1}^{r-n}kc_{k+n}\ds\frac{\partial}{\partial c_k},\qquad &0\leq n<r.
	\end{cases}
	\eeq
	
	The irregular state $\bigl|I^{\lb r\rb}\lb \mathbf{c}\rb\bigr\rangle$, $\mathbf c=\lb c_1,\ldots,c_r\rb$ is defined in \cite{GT12} by means of a recursive procedure in the rank~$r$. It is a solution of 
	\beq\label{Ward_eqs}
	L_n\left|I^{\lb r\rb}\lb \mathbf{c}\rb\right\rangle=\mathcal{L}_n^{(r)}\left|I^{\lb r\rb}\lb \mathbf{c}\rb\right\rangle,\qquad n\ge 0,
	\eeq
	which has an expansion 
	\beq\label{irr_state_dec}
	\left|I^{\lb r\rb}\lb \mathbf{c}\rb\right\rangle=f\lb \mathbf{c}\rb\lb\left|I^{\lb r-1\rb}\lb \mathbf{c'}\rb\right\rangle +\sum_{k=1}^\infty c_r^k \bigl|I^{\lb r-1\rb}_k\lb \mathbf{c'}\rb\bigr\rangle\rb,
	\eeq
	with $\mathbf{c'}=\lb c_1,\ldots,c_{r-1}\rb$. Note in particular that $\left|I^{\lb r\rb}\lb \mathbf{c}\rb\right\rangle$ is a joint eigenstate of $L_r,\ldots,L_{2r}$ and the corresponding eigenvalues are parameterized by $c_0,\ldots,c_r$. To lighten the notation, we omit the dependence of irregular states on parameters other than $\mathbf c$. 	 
	\begin{itemize}
	\item The generalized descendants $\bigl|I^{\lb r-1\rb}_k\lb \mathbf{c'}\rb\bigr\rangle$ appearing in \eqref{irr_state_dec} are linear combinations of states of the form
	\beq\label{notation_young}
	\mathbb{L}_{-\lambda}\frac{\partial^{m_1}}{\partial c_1^{m_1}}\ldots \frac{\partial^{m_{r-1}}}{\partial c_{r-1}^{m_{r-1}}}\left|I^{\lb r-1\rb}\lb \mathbf{c'}\rb\right\rangle,\qquad 
	\mathbb{L}_{-\lambda}=L_{-\lambda_1}\ldots L_{-\lambda_\ell},
	\eeq
	where $\lambda=\lb \lambda_1,\ldots,\lambda_\ell\rb$ is a partition (i.e. $\lambda_1,\ldots,\lambda_\ell\in\mathbb N$ with $\lambda_1\ge\ldots\ge\lambda_\ell$). We will denote by $\mathbb Y$ the set of all partitions identified with Young diagrams.
	\item
	The prefactor $f\lb\mathbf{c}\rb$ is in general singular at $c_r=0$. It is defined up to multiplication by a function analytic in $c_r$ in a neighborhood of $0$ whose value at $c_r=0$ is a non-zero constant independent of $\mathbf{c'}$. A typical ansatz for $f\lb \mathbf c\rb$ is 
	\beq
	f\lb \mathbf c\rb= c_r^{\alpha}\exp \sum_{k=0}^d f_k\lb\mathbf c'\rb c_r^{-k}.
	\eeq 
    \end{itemize}

	The proposal of Gaiotto-Teschner \cite{GT12} is that the equations \eqref{Ward_eqs} and the decomposition \eqref{irr_state_dec} fix both the form of non-analytic prefactor $f\lb \mathbf c\rb$ and the expressions of descendants $\bigl|I^{\lb r-1\rb}_k\lb \mathbf{c'}\rb\bigr\rangle$. For $r=1$, this was established earlier in \cite{G09}. For $r=2,3$, the proposal was confirmed in \cite{GT12,NU19}, yet a rigorous proof remains to be found.
	
	\begin{eg}
		The state $\left|I^{\lb 0\rb}\right\rangle$ is the usual Virasoro highest weight state $|\Delta\rangle$ such that 
		\beq
		L_0|\Delta\rangle=\Delta|\Delta\rangle,\qquad L_{n}|\Delta\rangle=0, \qquad n>0.
		\eeq 
		The Verma module $\mathcal V_{\Delta}$ is spanned by descendants $ \mathbb{L}_{-\lambda}|\Delta\rangle$, $\lambda\in\mathbb Y$. Similarly introduce the dual highest weight state satisfying $\langle \Delta|L_0=\Delta\langle \Delta|$, $\langle \Delta|L_{n<0}=0$ and the dual Verma module $\mathcal V_{\Delta}^*$. There exist a canonical non-degenerate bilinear pairing $\langle\cdot\rangle:\mathcal V_{\Delta}^*\times \mathcal V_{\Delta}\to\mathbb C$ uniquely defined by the condition that $\langle v|L_n\cdot |w\rangle=\langle v|\cdot L_n|w\rangle$ for all $\langle v|\in \mathcal V_{\Delta}^*,|w\rangle\in\mathcal V_{\Delta},n\in\mathbb Z $ and normalization $\langle\Delta|\cdot |\Delta\rangle=1$.
		
		The state $\bigl|I^{(1)}\bigr\rangle\in \mathcal V_{\Delta}$ can be reconstructed from
		the differential equation $\lb\mathcal L_0^{(1)}-\Delta\rb\langle \Delta|I^{(1)}\rangle=0$ solved by $\langle \Delta|I^{(1)}\rangle=\mathrm{const}\cdot c_1^{\Delta-c_0\lb Q-c_0\rb}$, and the only non-zero projections 
		\beq\label{projrank1}
		\langle \Delta|L_1^mL_2^n|I^{(1)}\rangle= 2^m\lb Q-c_0\rb^m \lb -1\rb^n c_1^{m+2n}\langle \Delta|I^{(1)}\rangle.
		\eeq
		Thus
		\beq
		|I^{(1)}\rangle=c_1^{\Delta-c_0\lb Q-c_0\rb}\left[|\Delta\rangle +\sum_{k=1}^{\infty}c_1^k\sum_{\lambda\in\mathbb Y,|\lambda|=k}G_{\lambda}\; \mathbb{L}_{-\lambda}|\Delta\rangle\right],
		\eeq
		where the coefficients $G_{\lambda}=G_{\lambda}\lb c_0,Q,\Delta\rb$ can be systematically computed for any $k$ from \eqref{projrank1} using the bilinear pairing.
		The non-analytic prefactor $c_1^{\Delta-c_0\lb Q-c_0\rb}$ is the simplest instance of $f\lb \mathbf c\rb$. 
	\end{eg}
	
		\subsection{Half-integer ranks\label{subsec_HIR}}
	It was recently realized  that the above construction can be adapted to half-integer ranks. The latter naming means that by analogy with \eqref{virdiff1} we require that $L_{n}\mapsto 0$ for $n>2r\in 2\mathbb N+1$ while keeping the eigenvalue of $L_{2r}$ non-vanishing. An appropriate differential operator realization of $\mathsf{Vir}_+$  was found in \cite{PP23} for $r=\frac52$. In \cite{HNNT24}, it has been suggested that the generalization to arbitrary half-integer ranks can be found by taking a suitable limit of the  integer rank realization. This has led to a conjectural recursive procedure for determining the operators $\mathcal L_{n}^{(r)}$ with $r\in\mathbb N+\frac12$. However,  finding their closed form for any $r$ stayed an open issue.
	
	It turns out that this problem admits a  concise solution that allows to consider the integer and half-integer ranks on equal footing. Loosely speaking, it suffices to set in \eqref{virdiff1} $Q=0$  (since we only care about the realization of $\mathsf{Vir}_+$, this does not imply any constraint on the central charge)  and let the coefficients $c_k$ to be labeled by positive half-integers (which does not alter the form of commutation relations).
	
	\begin{prop}\label{PropLs}
		For $s\in \mathbb N$, the mapping
		\beq\label{propLs}
		L_n\mapsto \mathcal L^{(s-\frac12)}_n=\begin{cases}
			0,\qquad & n\ge 2s,\\
			-\sum\limits_{k=n-s+\frac12}^{s-\frac12} c_{k} c_{n-k},\qquad & n=s,\ldots, 2s-1,\\
			-\sum\limits_{k=\frac12}^{s-\frac12} c_{k} c_{n-k}+
			\sum\limits_{k=\frac12}^{s-n-\frac12}kc_{n+k}\ds\frac{\partial}{\partial c_{k}},\qquad & n=0,\ldots, s-1,
			\end{cases}
		\eeq
		is consistent with the commutation relations of the subalgebra $\mathsf{Vir}_+$. Moreover, as an antihomomorphism into the subalgebra of 1st order differential operators in the eigenvalues of $L_s,\ldots,L_{2s-1}$, such a mapping is unique up to gauge transformations $\mathcal L^{(s-\frac12)}_n\mapsto e^{-F}\circ\mathcal L^{(s-\frac12)}_n\circ e^F$.
		\end{prop}
		\pf Straightforward calculation shows that $\left[\mathcal L^{(s-\frac12)}_m,\mathcal L^{(s-\frac12)}_n\right]=\lb n-m\rb \mathcal L^{(s-\frac12)}_{m+n}$ for any $m,n\ge0$. To demonstrate uniqueness, parameterize $\mathcal L^{(s-\frac12)}_{s},\ldots, \mathcal L^{(s-\frac12)}_{2s-1}$ by $c_{_{1/2}},\ldots, c_{_{s-1/2}}$ as in \eqref{propLs} and note that the differential parts of $ \mathcal L^{(s-\frac12)}_{0},\ldots, \mathcal L^{(s-\frac12)}_{s-1}$ are immediately fixed by their commutation relations with $\mathcal L^{(s-\frac12)}_{s},\ldots, \mathcal L^{(s-\frac12)}_{2s-1}$ acting by multiplication. 
		
		The scalar parts   of $ \mathcal L^{(s-\frac12)}_{0},\ldots, \mathcal L^{(s-\frac12)}_{s-1}$ have to satisfy a system of linear PDEs coming from the commutators $ \mathcal L^{(s-\frac12)}_{0},\ldots, \mathcal L^{(s-\frac12)}_{s-1}$  between themselves. For $n=0,\ldots, s-1$, write $\mathcal L^{(s-\frac12)}_{n}=\tilde{\mathcal L}^{(s-\frac12)}_{n}+D_n$, where $\tilde{\mathcal L}^{(s-\frac12)}_{n}$ is given by the expression of $\mathcal L^{(s-\frac12)}_{n}$ in \eqref{propLs}. Then the compatibility with $\mathsf{Vir}_+$ commutation relations is equivalent to
		\beq\label{PDEs}
		\sum\limits_{k=\frac12}^{s-m-\frac12}kc_{m+k}\ds\frac{\partial D_n}{\partial c_{k}}-
		\sum\limits_{k=\frac12}^{s-n-\frac12}kc_{n+k}\ds\frac{\partial D_m}{\partial c_{k}}=\lb m-n\rb D_{m+n},\qquad 0\leq m< n\leq s-1,
		\eeq
		where $D_k=0$ for $k\ge s$. We can use the gauge freedom to set $D_{s-1}=0$. From \eqref{PDEs} with $m=0,\ldots, s-2$ and $n=s-1$, it follows that $\ds\frac{\partial D_{m}}{\partial c_{_{1/2}}}=0$, and thus the functions $D_0,\ldots,D_{s-2}$ are independent of $ c_{_{1/2}}$. Now we can use the residual gauge freedom to set $D_{s-2}=0$ and iterate the procedure: equations \eqref{PDEs} with $m=0,\ldots, s-3$ and $n=s-2$ shows that $D_0,\ldots, D_{s-3}$ are independent of $ c_{_{3/2}}$ and so on.\epf
	
	In the  most relevant to us rank $\frac52$ case, the nontrivial operators explicitly read
	\beq
	\label{52equationsV2}
	\begin{gathered}
	\mathcal L_5^{(5/2)}=-c_{_{5/2}}^2,\qquad 
	\mathcal L_4^{(5/2)}=-2c_{_{3/2}}c_{_{5/2}},
	\qquad 
	\mathcal L_3^{(5/2)}=-2c_{_{1/2}}c_{_{5/2}}-c_{_{3/2}}^2,	\\
		\mathcal L_2^{(5/2)}=-2c_{_{1/2}}c_{_{3/2}}+\frac{c_{_{5/2}}}{2}\frac{\partial}{\partial c_{_{1/2}}},\qquad
		\mathcal L_1^{(5/2)}=-c_{_{1/2}}^2+\frac{3c_{_{5/2}}}{2}\frac{\partial}{\partial c_{_{3/2}}}+\frac{c_{_{3/2}}}{2}\frac{\partial}{\partial c_{_{1/2}}},\\
	\mathcal L_0^{(5/2)}=\frac{5 c_{_{5/2}}}{2}\frac{\partial}{\partial c_{_{5/2}}}+
	\frac{3 c_{_{3/2}}}{2}\frac{\partial}{\partial c_{_{3/2}}}+
	 \frac{c_{_{1/2}}}{2}\frac{\partial}{\partial c_{_{1/2}}}.
	 \end{gathered}
	\eeq
	They are of course gauge equivalent to the expressions given in \cite{PP23} and \cite{HNNT24}. E.g. parameterizing $\mathcal L_n^{(5/2)}=\Lambda_n$ for $n=3,4,5$, the other  generators become
	\begin{subequations} \label{L52eq2v2}
\begin{align}
	\mathcal L^{(5/2)}_2=&\,\Lambda_5\frac{\partial}{\partial \Lambda_3}-\frac{\Lambda_4\lb\Lambda^2_4-4\Lambda_3\Lambda_5\rb}{8\Lambda_5^2},\\
	\mathcal L^{(5/2)}_1=&\, 3\Lambda_5\frac{\partial}{\partial \Lambda_4}+2\Lambda_4\frac{\partial}{\partial \Lambda_3}+\frac{\lb\Lambda^2_4-4\Lambda_3\Lambda_5\rb^2}{64\Lambda_5^3},
	\\
	\mathcal L^{(5/2)}_0=&\,5\Lambda_5\frac{\partial }{\partial \Lambda_5}+
	4\Lambda_4\frac{\partial }{\partial \Lambda_4}+
	3\Lambda_3\frac{\partial }{\partial \Lambda_3}.
\end{align}
\end{subequations}
   The gauge transformation with $F=\ds\frac{47\Lambda_4^5}{960\Lambda_5^4}-\ds\frac{5\Lambda_3\Lambda_4^3}{24\Lambda_5^3}+\frac{\Lambda_3^2\Lambda_4}{4\Lambda_5^2}$ then yields the realization of \cite[Eqs. (3.41)--(3.46)]{HNNT24}. Further reparameterization $\Lambda_3=-2c_1c_2$, $\Lambda_4=-c_2^2$ combined with the gauge transformation corresponding to $F=\ds\frac{c_1c_2^4\lb c_2^3+4c_1\Lambda_5\rb}{4\Lambda_5^3}+\ds\frac{17c_2^{10}}{960\Lambda_5^4}$ gives
   \begin{subequations}\label{PPLs}
   \begin{gather}
   \mathcal L^{(5/2)}_{5}=\Lambda_5,\qquad
   \mathcal L^{(5/2)}_{4}=-c_2^2,\qquad
   \mathcal L^{(5/2)}_{3}=-2c_1c_2,\qquad
   \mathcal L^{(5/2)}_2=-\frac{\Lambda_5}{2c_2}\frac{\partial}{\partial c_1},\\
   \mathcal L^{(5/2)}_1=-\frac{3\Lambda_5}{2c_2}\frac{\partial}{\partial c_2}+
   \frac{2c_2^3+3c_1\Lambda_5}{2c_2^2}\frac{\partial}{\partial c_1}-\frac{2c_1^2c_2^2}{\Lambda_5},\\
   \mathcal L^{(5/2)}_0=5\Lambda_5\frac{\partial }{\partial \Lambda_5}+
   2c_2\frac{\partial }{\partial c_2}+
   c_1\frac{\partial }{\partial c_1}.
   \end{gather}
   \end{subequations}
   This coincides with \cite[Eq. (2.8)]{PP23} after the sign flip $\Lambda_5\mapsto-\Lambda_5$.

	\subsection{Differential operator approach\label{subsec_Diffopapp}}
	Consider   the rank $2$ irregular state   $\bigl|I^{(2)}\bigr\rangle$ defined by $L_n\bigl|I^{(2)}\bigr\rangle=\mathcal L_n^{(2)}\bigl|I^{(2)}\bigr\rangle$ for all $n\ge0$. In the case $r=2$, the differential operators $\mathcal L^{(2)}_{n}$ in \eqref{virdiff1} are explicitly written as
	\beq
	\begin{gathered}
	\mathcal L^{(2)}_{n\ge 5}=0,\qquad\mathcal L^{(2)}_4=-c_2^2,\qquad \mathcal L^{(2)}_3=-2c_1c_2,\qquad\mathcal L^{(2)}_2=-2c_0c_2-c_1^2+3Qc_2,\\
	\mathcal L^{(2)}_1=2c_1\lb Q-c_0\rb+c_2\frac{\partial }{\partial c_1},\qquad 
	\mathcal L^{(2)}_0=c_0\lb Q-c_0\rb+
	2c_2\frac{\partial }{\partial c_2}+
	c_1\frac{\partial }{\partial c_1}.
	\end{gathered} 
	\eeq
	Let $\langle 0|$ be the Virasoro vacuum state satisfying $\langle 0|L_n=0$ for $n\leq 1$. Define conformal block $\mathcal F^{(2)}\lb c_1,c_2\rb$ by 
	\beq
	\mathcal F^{(2)}\lb c_1,c_2\rb= \bigl \langle 0 \bigr| I^{(2)}\bigr\rangle.
	\eeq
	The pairing between the vacuum module and the rank 2 module is required to satisfy the usual compatibility condition $\langle v|L_n\cdot |w\rangle=\langle v|\cdot L_n|w\rangle$ fixing it up to rescaling. This leads to Ward identities $\mathcal L^{(2)}_{0}\mathcal F^{(2)}=\mathcal L^{(2)}_{1}\mathcal F^{(2)}=0$ which are immediately solved by
	\beq
	\mathcal F^{(2)}\lb c_1,c_2\rb=\operatorname{const}\cdot c_2^{-c_0\lb Q-c_0\rb/2}\exp\frac{\lb c_0-Q\rb c_1^2}{c_2}.
	\eeq
	The constant prefactor is independent of $c_1$, $c_2$ and can be set equal to $1$. Note that if we replace $\langle 0|$ by a generic highest weight state $\langle \Delta|$, the Ward identities alone are no longer sufficient to fix $\langle \Delta|I^{(2)}\rangle$ since the right action of $L_1$ on $\langle \Delta|$ is nontrivial.

	Following the same procedure for $r=\frac52$, let us define conformal block $\mathcal F^{(5/2)}\lb c_1,c_2,\Lambda_5\rb$ by 
	\beq
	\mathcal F^{(5/2)}\lb c_1,c_2,\Lambda_5\rb= \bigl\langle 0 \bigr| I^{(5/2)}\bigr\rangle,
	\eeq
	where $L_n\bigl|I^{(5/2)}\bigr\rangle=\mathcal L_n^{(5/2)}\bigl|I^{(5/2)}\bigr\rangle$ for $n\ge 0$ and the differential operators $\mathcal L^{(5/2)}_n$ are given by \eqref{PPLs}. The equation $\mathcal L^{(5/2)}_{0}\mathcal F^{(5/2)}=0$ implies that if $c_1$, $c_2$ and $\Lambda_5$ are assigned degrees $1$, $2$ and $5$ then $\mathcal F^{(5/2)}\lb c_1,c_2,\Lambda_5\rb$ is homogeneous of degree $0$, so that
	\beq
	\mathcal F^{(5/2)}\lb c_1,c_2,\Lambda_5\rb=G\lb x,y\rb,\qquad \text{with } \quad x=\frac{\Lambda_5}{c_2^{5/2}},
	\; y= 
	\frac{c_1\Lambda_5}{c_2^3}.
	\eeq
	The second constraint $\mathcal L^{(5/2)}_{1}\mathcal F^{(5/2)}=0$ then yields a linear PDE
	\beq\label{exppref}
    \left[15x\frac{\partial }{\partial x}+4\lb 1+6y\rb\frac{\partial }{\partial y}-\frac{8y^2}{x^4}\right]G\lb x,y\rb=0,
	\eeq
	which can be solved by the standard method of characteristics.
	The solution is given by 
	\beq\label{expprefsol}
	G\lb x,y\rb= \exp\left\{-\frac{2 \left(135 y^2+30 y+2\right)}{405 x^4}\right\} F\lb s \rb,\qquad s=\frac{2\lb 1+6 y\rb^{5/4}}{3x^2}.
	\eeq
	
	The arbitrary function $F\lb s\rb$ represents the nontrivial part of the conformal block and has to be determined by other means. Anticipating the result, we write it in the form of a formal asymptotic series at $s=\infty$ (which corresponds to sending $\Lambda_5\to0$):
	\beq\label{CBansatz}
	F\lb s\rb=s^{-2\alpha}e^{\frac{s^2}{45}+\frac{4i\nu s}{5}}\left[1+\sum_{k=1}^\infty F_k s^{-k}\right].
	\eeq
	The role of $\frac{s^2}{45}$ here is to cancel the most singular term proportional to $\Lambda_5^{-4}$ in the exponential prefactor in \eqref{expprefsol} so that $\ln\mathcal F^{(5/2)}\lb c_1,c_2,\Lambda_5\rb=O\lb \Lambda_5^{-3}\rb$ as $\Lambda_5\to 0$. In fact, there occur more substantial cancellations:
	\beq
	\frac{s^2}{45}-\frac{2 \left(135 y^2+30 y+2\right)}{405 x^4}=\frac{2 c_1^3 c_2}{3 \Lambda_5}+O\lb 1\rb\qquad\text{as }\;\Lambda_5\to 0.
	\eeq
	The parameter $\nu$ in \eqref{CBansatz} is free while $\alpha$ and the coefficients $F_k$ are to be found in terms of $\nu$.
	
	After these preparations, let us now revisit the construction of \cite{PP23}. The idea is to look for the rank~$\frac52$ irregular state $| I^{(5/2)}\rangle$ in the form of a series in $\Lambda_5$ similar to the expansion \eqref{irr_state_dec}, 
		\beq\label{irr_state_dec_bis}
	\left|I^{\lb 5/2\rb}\right\rangle=f\lb c_1,c_2,\Lambda_5\rb\,|\Psi\rangle,
	\qquad |\Psi\rangle=
	\sum_{k=0}^\infty \Lambda_5^k \bigl|I^{\lb 2\rb}_k\bigr\rangle,
	\eeq
	where $\bigl|I^{\lb 2\rb}_0\bigr\rangle=\bigl|I^{(2)}\bigr\rangle$. The generalized descendants $	\bigl|I^{\lb 2\rb}_k\bigr\rangle$ have the form
	\beq\label{gendes}
	\bigl|I^{\lb 2\rb}_k\bigr\rangle=\sum_{n_1,n_2\in\mathbb Z_{\ge0}}\sum_{\lambda\in\mathbb Y}A^{[k]}_{\lambda;n_1,n_2}\lb c_1,c_2\rb\mathbb L_{-\lambda}\frac{\partial^{n_1}}{\partial c_1^{n_1}}\frac{\partial^{n_2}}{\partial c_2^{n_2}}\bigl|I^{(2)}\bigr\rangle.
	\eeq
	The above considerations suggest the following form of the singular prefactor $f\lb c_1,c_2,\Lambda_5\rb$ in \eqref{irr_state_dec_bis} as
	\begin{subequations}
	\begin{align}
	f\lb c_1,c_2,\Lambda_5\rb=&\,\lb\mathcal F^{(2)}\lb c_1,c_2\rb\rb^{-1}  s^{-2\alpha}\exp\left\{\frac{s^2}{45}+\frac{4i\nu s}{5}-\frac{2 \left(135 y^2+30 y+2\right)}{405 x^4}\right\}=
	 \\ 
	 \label{fsingb}=&\,c_2^{c_0\lb Q-c_0\rb/2}s^{-2\alpha}\exp\left\{\frac{s^2}{45}+\frac{4i\nu s}{5}
	 -\frac{4c_2^{10}}{405\Lambda_5^4}-\frac{4c_1c_2^7}{27\Lambda_5^3}-
	 \frac{2c_1^2c_2^4}{3\Lambda_5^2}+\frac{\lb Q-c_0\rb c_1^2}{c_2}\right\},
	\end{align}
	\end{subequations}
	with 
	\beq\label{irrCR}
	s=\frac{2c_2^5}{3\Lambda_5^2}\lb 1+\frac{6c_1\Lambda_5}{c_2^3}\rb^{5/4}.
	\eeq
	The  formula \eqref{fsingb} is equivalent to \cite[Eqs. (2.12)--(2.14)]{PP23} where this prefactor was found by invoking a Seiberg-Witten curve analysis instead of solving Ward identities.
	
	Substituting \eqref{irr_state_dec_bis} into the  defining equations of irregular state, we have $L_n|\Psi\rangle=f^{-1}\circ\mathcal L^{(5/2)}_n\circ f\,|\Psi\rangle$. The operators on the right hand side of this relation can be expressed as
	\begin{align}
		f^{-1}\circ\mathcal L^{(5/2)}_n\circ f=
		\begin{cases}
		0,\qquad & n>5,\\
		\Lambda_5, \qquad & n=5,\\
		-c_2^2,\qquad & n=4, \\
		-2c_1c_2,\qquad & n=3, \\	
		-\ds\frac{\Lambda_5}{2c_2}\ds\frac{\partial}{\partial c_1}+\varphi_2,\qquad & n=2,\\
		-\ds\frac{3\Lambda_5}{2c_2}\ds\frac{\partial}{\partial c_2}+
		\ds\frac{2c_2^3+3c_1\Lambda_5}{2c_2^2}\frac{\partial}{\partial c_1}+\varphi_1,\qquad & n=1,\\
		5\Lambda_5\ds\frac{\partial }{\partial \Lambda_5}+
		2c_2\ds\frac{\partial }{\partial c_2}+
		c_1\ds\frac{\partial }{\partial c_1}+c_0\lb Q-c_0\rb,\qquad & n=0,
		\end{cases}	
	\end{align}
	where the functions $\varphi_1$, $\varphi_2$ are readily computed as
	\begin{subequations}
	\begin{gather}
		\varphi_1=\lb Q-c_0\rb\left[ 2c_1+\frac{3\lb 6c_1^2-c_0c_2\rb}{4c_2^3}\Lambda_5\right],\\
	\begin{gathered}
	\varphi_2=\frac{2c_2^6}{27\Lambda_5^2}\left[1+\frac{9c_1\Lambda_5}{c_2^3}-
	\lb1+\frac{6c_1\Lambda_5}{c_2^3}\rb^{3/2}\right]-2i\nu c_2\lb1+\frac{6c_1\Lambda_5}{c_2^3}\rb^{1/4}\\-\frac{\lb Q-c_0\rb c_1}{c_2}\Lambda_5+\frac{15\alpha \Lambda_5^2}{2c_2^4}\lb1+\frac{6c_1\Lambda_5}{c_2^3}\rb^{-1}.
	\end{gathered}
	\end{gather}
	\end{subequations}
	It is important to note that the  $\Lambda_5\to 0$ expansion of $\varphi_2$  does not contain singular terms. Collecting the coefficients of different powers $\Lambda_5^k$ in $L_n|\Psi\rangle=f^{-1}\circ\mathcal L^{(5/2)}_n\circ f\,|\Psi\rangle$ yields a system of recurrence relations for the generalized descendants:
	\begin{subequations}
		\label{recrels1}
	\begin{gather}
		L_{n>5}\bigl|I^{\lb 2\rb}_k\bigr\rangle=0,\qquad 
		L_5 \bigl|I^{\lb 2\rb}_k\bigr\rangle=\bigl|I^{\lb 2\rb}_{k-1}\bigr\rangle,\qquad
		\lb L_4+c_2^2\rb \bigl|I^{\lb 2\rb}_k\bigr\rangle=0,\qquad 
		\lb L_3+2c_1c_2\rb \bigl|I^{\lb 2\rb}_k\bigr\rangle=0,\\
		L_2\bigl|I^{\lb 2\rb}_k\bigr\rangle=-\frac{1}{2c_2}\frac{\partial}{\partial c_1}
		\bigl|I^{\lb 2\rb}_{k-1}\bigr\rangle+\sum_{j=0}^k\varphi_{2,j}\bigl|I^{\lb 2\rb}_{k-j}\bigr\rangle,\\
		\lb L_1 -\mathcal L_1^{(2)}\rb\bigl|I^{\lb 2\rb}_k\bigr\rangle=\lb-\frac{3}{2c_2}\frac{\partial}{\partial c_2}+\frac{3c_1}{2c_2^2}\frac{\partial}{\partial c_1}+\frac{3\lb Q-c_0\rb\lb 6c_1^2-c_0c_2\rb}{4c_2^3}\rb \bigl|I^{\lb 2\rb}_{k-1}\bigr\rangle,\\
		\lb L_0 -\mathcal L_0^{(2)}-5k\rb\bigl|I^{\lb 2\rb}_k\bigr\rangle=0,
	\end{gather}
	\end{subequations}
	where $\varphi_{2,j}$ denote the coefficients in the Taylor expansion $\varphi_2=\sum_{j=0}^\infty \varphi_{2,j}\Lambda_5^j $. For $k=0$, all of these equations are satisfied automatically except for the $L_2$ equation which expresses the corresponding eigenvalue parameter $c_0$ in terms of $\nu$ as $c_0=i\nu+\frac{3Q}{2}$.
	
	One may attempt to solve the system \eqref{recrels1} recursively by restricting possible values of $n_1,n_2,\lambda$ in \eqref{gendes} to a suitable finite subset of $\mathbb Z_{\ge0}^2\times \mathbb Y$. For example, substituting   $\bigl|I^{\lb 2\rb}_1\bigr\rangle$ in the form
	\beq
	\bigl|I^{\lb 2\rb}_1\bigr\rangle=\lb A L_{-1}+B\frac{\partial }{\partial c_2}+C\frac{\partial }{\partial c_1}+D\rb \bigl|I^{\lb 2\rb}\bigr\rangle
	\eeq
	into the relations \eqref{recrels1} with $k=1$, we see that the $L_{n>5}$ equations are satisfied automatically, the $L_5$ equation determines $A=-\ds\frac{1}{6c_2^2}$, the $L_4$ equation gives $B=-5Ac_1=\ds\frac{5c_1}{6c_2^2}$ and the $L_3$ equation yields
	\beq C=\frac{2A\lb3Qc_2-2c_0c_2-c_1^2\rb-Bc_1}{c_2}=\frac{2c_0-3Q}{3c_2^2}-\frac{c_1^2}{2c_2^3}.
	\eeq
	The $L_2$ equation does not produce new constraints while $L_1$ and $L_0$ equations uniquely determine the remaining coefficient\footnote{This expression differs from its counterpart in Eq. (2.16) in \cite{PP23} because our choice of the singular prefactor $f$ differs from the one in  loc. cit. by a non-singular function $\lb 1-\frac{6c_1\Lambda_5}{c_2^3} \rb^{5\rho_5/8}$ and the sign of $\Lambda_5$.}
	\beq
	D=\ds\frac{\lb c_0-Q\rb\lb 22c_1^2+3c_2\lb 8Q-7c_0\rb \rb c_1}{12c_2^4}.
	\eeq
	
	Though in principle similar calculations  can be carried  out at higher orders in $\Lambda_5$, they quickly become quite cumbersome. However, the procedure we have just outlined does not take into account the symmetry of conformal block, namely the fact that $\mathcal F^{(5/2)}\lb c_1,c_2,\Lambda_5\rb$ nontrivially depends only on  the irregular cross-ratio  given by \eqref{irrCR}. In the next subsection, we develop a simpler, purely algebraic construction of irregular state which implements this symmetry and effectively sets $c_1=0$, $c_2=-\frac{i}{2}$ from the very beginning.

   \subsection{Algebraic construction}\label{SubsecAlgCon}
   \subsubsection{Whittaker vector of rank $\frac52$}
   Let $\mathbb C\mathcal |J\rangle$ be a one-dimensional representation of the subalgebra $\bigoplus\limits_{n\ge 2}\mathbb C L_n\subset \mathsf{Vir}$ defined by
   \beq\label{Wh2}
   L_{n>4}|J\rangle=0,\qquad L_4|J\rangle=\frac14 |J\rangle,\qquad L_3|J\rangle=0,
   \qquad   L_2|J\rangle=-\nu |J\rangle,
   \eeq
   with $\nu\in\mathbb C$. The vector $|J\rangle$ is called a rank 2 Whittaker state. The Whittaker module $\mathcal V^{[2]}_{\nu}$ is the corresponding induced representation of $\mathsf{Vir}$,
   \beq\label{PBWbasis}
   \mathcal V^{[2]}_{\nu}=\bigoplus_{m_{0,1}\in \mathbb Z_{\ge 0},\lambda\in\mathbb Y}\mathbb C\,\mathbb L_{-\lambda}L_0^{m_0}L_1^{m_1}|J\rangle.
   \eeq
   
   Our first task to find a vector $|\Psi\rangle\in\mathcal V^{[2]}_{\nu}$ such that
   \beq\label{Wh52}
   L_{n\ge 6}|\Psi\rangle=0,\qquad L_5|\Psi\rangle=\varepsilon |\Psi\rangle,\qquad
     L_4|\Psi\rangle=\frac14 |\Psi\rangle,\qquad L_3|\Psi\rangle=0,
   \eeq
   with $\varepsilon\in\mathbb C$. Any such $|\Psi\rangle$ will be called a Whittaker state of rank $\frac52$. By analogy with the Gaiotto-Teschner construction \eqref{irr_state_dec}, one may look for $|\Psi\rangle$ in the form of a power series in $\varepsilon$:
   \beq\label{Wh52exp}
   |\Psi\rangle=|J\rangle+\sum_{k=1}^{\infty}\varepsilon^k\mathbb G_k|J\rangle,
   \eeq
   where
   \beq \label{Wh52expAux}
   \mathbb G_k=\sum_{m_{0,1}\in\mathbb Z_{\ge0},\lambda\in\mathbb Y} G_{\lambda;m_0,m_1}^{[k]}\mathbb L_{-\lambda}L_0^{m_0}L_1^{m_1},
   \eeq
   and the coefficients $G_{\lambda;m_0,m_1}^{[k]}$ are independent of $\varepsilon$. Rescaling $|\Psi\rangle\mapsto f|\Psi\rangle$ by a power series $f\in\mathbb C[[\varepsilon]]$ such that $f\lb 0\rb=1$ does not change the relations \eqref{Wh52} nor the normalization $|\Psi\rangle_{\varepsilon=0}=|J\rangle$. We fix this ambiguity by requiring that 
   \beq\label{Wh52normalization}
   G_{\emptyset;0,0}^{[k]}=0,\qquad k\in\mathbb N.
   \eeq
   In other words, there are no terms proportional to $|J\rangle$ in the expression of all descendants $\mathbb G_k|J\rangle$ .
   Using the canonical pairing between $\mathcal V^{[2]}_{\nu}$ and Virasoro vacuum module, these conditions can be rewritten as $\langle 0|\Psi\rangle=\langle 0|J\rangle$.
   We will refer to \eqref{Wh52normalization} as the orthogonal gauge.
   
   \begin{rmk}
   	Fixing the eigenvalues of $L_3$ and $L_4$ as in \eqref{Wh2} involves no loss of generality since the vector $| J'\rangle=e^{\alpha L_0}e^{\beta L_1}|J\rangle$ is also a joint eigenstate of $L_2,L_3,L_4$ and
   	\beq
   	L_{n\ge 5}|J'\rangle=0,\qquad L_4|J'\rangle=\frac{e^{4\alpha}}{4} |J'\rangle,\qquad L_3|J'\rangle=
   	\frac{\beta e^{3\alpha}}{2}|J'\rangle,
   	\qquad   L_2|J'\rangle=\lb \frac{\beta^2}{4}-\nu\rb e^{2\alpha} |J'\rangle.
   	\eeq
   	Similar considerations can be applied to a joint eigenstate of $L_3,L_4,L_5$  to adjust its eigenvalues to \eqref{Wh52}. The latter relations are designed precisely so as to reproduce their lower rank counterparts \eqref{Wh2} in the limit $\varepsilon\to0$.
   \end{rmk}
   
   The main result of this subsection is the following
   \begin{theo}\label{conjWh52}
   	The rank $\frac52$ Whittaker state $|\Psi\rangle\in\mathcal V^{[2]}_{\nu}$ defined by \eqref{Wh52}--\eqref{Wh52expAux} and normalization \eqref{Wh52normalization} uniquely exists. Moreover,
   	\beq\label{selectionrulesWH}
   	\mathbb G_k=\sum_{\substack{1\leq\ell\leq k \\ \ell=k\;\mathrm{mod}\;2}}\sum_{\substack{m_0,m_1\in\mathbb Z_{\ge0},\lambda\in\mathbb Y \\ m_1+2m_0+|\lambda|=\ell}}G_{\lambda;m_0,m_1}^{[k]}\mathbb L_{-\lambda}L_0^{m_0}L_1^{m_1}.
   	\eeq
   \end{theo}
   \pf   The meaning of the selection rules in \eqref{selectionrulesWH} is as follows. We define $\operatorname{deg}L_{-k}=k$ for $k>0$, $\operatorname{deg}L_0=2$ and ${\operatorname{deg}L_1=1}$. The Virasoro monomials allowed in $\mathbb G_k$ have total degree $k$ modulo $2$. The latter restriction may be traced back to the automorphism $L_n\mapsto \lb -1\rb^n L_n$, $n\in\mathbb Z$ of the Virasoro algebra. The nontrivial part of the proof concerns the restriction 
   $m_1+2m_0+|\lambda|\le k$, and is given in Appendix~\ref{AppProof}. The proof builds upon and refines the ideas from \cite{Nagoya15}.  \epf
   
   To explain the internal mechanics and practical implications of this statement, let us first note that the joint eigenstate conditions \eqref{Wh52} can be equivalently rewritten as
   \beq\label{Gkrels}
   \left[L_5,\mathbb G_k\right]|J\rangle=
   \mathbb G_{k-1}|J\rangle,\qquad  \left[L_n,\mathbb G_k\right]|J\rangle=0,\qquad n=3,4,6,
   \eeq
   where $\mathbb G_0=1$. The generators $L_{n}$ with $n\ge 7$ can be obtained by commuting $L_3,\ldots ,L_6$, therefore the corresponding equations are satisfied automatically.

   \begin{itemize}
   	\item For $k=1$, Theorem~\ref{conjWh52} predicts that
   	\beq
   	\mathbb G_1=G^{[1]}_{\ytableausetup
   		{mathmode,  boxsize=0.4em}\ydiagram{1}\,;0,0}L_{-1}+G^{[1]}_{\emptyset;0,1}L_1.
   	\eeq
   	This ansatz automatically satisfies $L_6$ and $L_4$ relation
    in \eqref{Gkrels}. The $L_5$ and $L_3$ relations imply respectively that $\frac32G^{[1]}_{\ytableausetup
   		{mathmode,  boxsize=0.3em}\ydiagram{1}\,;0,0}=1$ and
   		$\frac12 G^{[1]}_{\emptyset;0,1}-4\nu G^{[1]}_{\ytableausetup
   			{mathmode,  boxsize=0.3em}\ydiagram{1}\,;0,0}=0$, so that
   			\beq
   			\mathbb G_1=\frac23  L_{-1}+\frac{16\nu}{3}L_1.
   			\eeq
   	\item Likewise, for $k=2$ we expect at most 5 terms in the descendant state,
   	\beq
   	\mathbb G_2=G^{[2]}_{\ytableausetup
   		{mathmode,  boxsize=0.4em}\ydiagram{2}\,;0,0}L_{-2}+G^{[2]}_{\ytableausetup
   		{mathmode,  boxsize=0.4em}\ydiagram{1,1}\,;0,0}L_{-1}^2+G^{[2]}_{\ytableausetup
   		{mathmode,  boxsize=0.4em}\ydiagram{1}\,;0,1}L_{-1}L_1
   		+G^{[2]}_{\emptyset;1,0}L_0+G^{[2]}_{\emptyset;0,2}L_1^2.
   		\eeq
   The $L_5$ equation determines two coefficients $G^{[2]}_{\ytableausetup
   	{mathmode,  boxsize=0.4em}\ydiagram{1}\,;0,1}=\frac{32\nu}{9}$, $G^{[2]}_{\ytableausetup
   	{mathmode,  boxsize=0.4em}\ydiagram{1,1}\,;0,0}=\frac29$. The $L_6$ equation yields a linear relation $2G^{[2]}_{\ytableausetup
   	{mathmode,  boxsize=0.4em}\ydiagram{2}\,;0,0}+\frac{21}{2}G^{[2]}_{\ytableausetup
   	{mathmode,  boxsize=0.4em}\ydiagram{1,1}\,;0,0}=0$ which in turn implies that $G^{[2]}_{\ytableausetup
   	{mathmode,  boxsize=0.4em}\ydiagram{2}\,;0,0}=-\frac76$. The $L_3$ equation gives two constraints. One of them involves $G^{[2]}_{\ytableausetup
   	{mathmode,  boxsize=0.4em}\ydiagram{1}\,;0,1}$, $G^{[2]}_{\ytableausetup
   	{mathmode,  boxsize=0.4em}\ydiagram{1,1}\,;0,0}$ and turns out to be satisfied by the previously determined expressions. Another one involves 3 already found coefficients and $G^{[2]}_{\emptyset;0,2}$ thereby fixing the latter. The $L_4$ equation produces one relation containing the same 3 coefficients and $G^{[2]}_{\emptyset;1,0}$. The final result is
   	\beq\label{G2}
   	\mathbb G_2=-\frac76L_{-2}+\frac29L_{-1}^2+\frac{32\nu}{9}L_{-1}L_1-\frac{103\nu}{9}L_0+\frac{256\nu^2+57}{18} L_1^2.
   	\eeq
   \end{itemize}
   The pattern described above continues for higher values of $k$. Equations \eqref{Gkrels} generate an overdetermined linear system for the coefficients $G^{[k]}_{\lambda;m_0,m_1}$. The key part of the statement concerns the existence of its solution. It was additionally checked by explicit computation of higher-order descendants $\mathbb G_3,\ldots,\mathbb G_{12}$. The expressions of  $\mathbb G_3$, $\mathbb G_4$, $\mathbb G_5$ (in the case $\nu=0$, also $\mathbb G_6$) are given in Appendix~\ref{AppA1}. For generic $\nu$, all of the coefficients in \eqref{selectionrulesWH} turn out to be non-zero.

   \subsubsection{Conformal blocks}\label{SubsectionConfBlocks52}
   Let us denote by $L_{\varepsilon}\in\mathsf{Vir}$ the element
   \beq
   L_{\varepsilon}=L_2-2\varepsilon L_1+6\varepsilon^2 L_0,
   \eeq
   and consider the Virasoro subalgebra $\mathsf{Vir}_{\varepsilon}^{(5/2)}=\bigoplus\limits_{n\ge 3}\mathbb C L_n\oplus\mathbb C L_{\varepsilon}$. There exists an anti-homomorphism of $\mathsf{Vir}_{\varepsilon}^{(5/2)}$ into the algebra of 1st order differential operators in $\varepsilon$:
   \beq
   \begin{gathered}
   	L_n\mapsto \mathcal L_n,\qquad n\in \{\varepsilon\}\cup \mathbb N_{\ge3},\\
   \mathcal L_{n\ge 6}= 0,\qquad \mathcal L_5= \varepsilon,\qquad 
   \mathcal L_4=\frac14,
   \qquad \mathcal L_3= 0,\qquad 
   \mathcal L_{\varepsilon}=30\varepsilon^3\frac{\partial}{\partial\varepsilon}-\nu.
   \end{gathered}
   \eeq
   
   \begin{rmk} The form of $L_{\varepsilon}$ is fixed by the condition for  $\mathcal L_{n\ge 3}$ to reproduce the eigenvalues of $L_n$ in \eqref{Wh52}. Indeed, replace $L_\varepsilon$ by an arbitrary linear combination $\tilde{L}=AL_2+BL_1+CL_0$, then we have 
   	\begin{subequations}
   	\begin{align}
   	\left[L_3,\tilde{L}\right]=&\,AL_5+2BL_4+3CL_3\mapsto A\varepsilon+\frac{B}{2},\\
   	\left[L_4,\tilde{L}\right]=&\,2AL_6+3BL_5+4CL_4\mapsto 3B\varepsilon+C.
   	\end{align}
   	\end{subequations}
   	However, since $\mathcal L_3,\mathcal L_4$ are independent of $\varepsilon$, both of these commutators should be mapped to $0$, which implies that $\tilde L=AL_\varepsilon$. The form of $\mathcal L_{\varepsilon}$ is then determined, up to a gauge transformation, by the remaining commutator $[L_5,L_\varepsilon]=3L_7-8\varepsilon L_6+30\varepsilon^2L_5\mapsto-30\varepsilon^3$. The constant term $-\nu$ is included into the definition of $\mathcal L_{\varepsilon}$ for later convenience.
   	\end{rmk}
   
    By analogy with \eqref{Ward_eqs}, let us now introduce the irregular state $\bigl|I^{(5/2)}\bigr\rangle\in \mathcal V_{\nu}^{[2]}$ by
    \beq\label{RedIrrStateDec}
    L_n\bigl|I^{(5/2)}\bigr\rangle=\mathcal L_n \bigl|I^{(5/2)}\bigr\rangle,\qquad  n\in \{\varepsilon\}\cup \mathbb N_{\ge3}.
    \eeq
    Up to a scalar prefactor, it must coincide with the Whittaker state $|\Psi\rangle$ defined by \eqref{Wh52}--\eqref{Wh52normalization}:
    \beq\label{IrrStateF}
    \bigl|I^{(5/2)}\bigr\rangle=\mathcal F\lb\varepsilon\,|\,\nu\rb\,|\Psi\rangle.
    \eeq 
    Just as our notation suggests, the prefactor $\mathcal F\lb\varepsilon\,|\,\nu\rb$ is nothing but the sought-for irregular conformal block. Indeed, taking into account the expansions \eqref{Wh52exp}--\eqref{Wh52expAux}and  recalling that $\langle 0|L_{n}=0$ for $n\le 1$, $\mathcal F\lb\varepsilon\,|\,\nu\rb$ can be written as
    \beq
    \mathcal F\lb\varepsilon\,|\,\nu\rb=
    \frac{\bigl\langle 0\,\bigl|I^{(5/2)}\bigr\rangle}{\langle 0|J\rangle\;}.
    \eeq
    
    The explicit form of conformal block  is obtained from the $L_\varepsilon$ equation in \eqref{RedIrrStateDec}:
    \beq
    \mathcal F\lb\varepsilon\,|\,\nu\rb L_\varepsilon|\Psi\rangle=\mathcal L_\varepsilon \mathcal F\lb\varepsilon\,|\,\nu\rb |\Psi\rangle.
    \eeq
    The procedure is as follows. In a more elaborated form, the last equation reads
    \beq\label{EqForFcal}
    \left[L_2 -2\varepsilon L_1+6\varepsilon^2 L_0+\nu-30\varepsilon^3\frac{\partial }{\partial\varepsilon}-30\varepsilon^3\frac{\partial
    \ln \mathcal F\lb\varepsilon\,|\,\nu\rb}{\partial\varepsilon}\right] 
    \sum_{k=0}^{\infty}\varepsilon^k\mathbb G_k|J\rangle=0.
    \eeq
    The Virasoro commutation relations together with $L_2|J\rangle=-\nu|J\rangle$ can now be used to express the left hand side in the  previously used Poincar\'e-Birkhoff-Witt basis \eqref{PBWbasis}  of $\mathcal V_{\nu}^{[2]}$. 
    \begin{conj} Irregular state $\bigl|I^{(5/2)}\bigr\rangle$ exists. Equations  \eqref{EqForFcal} admit a solution at every order in $\varepsilon$, with 
    \beq\label{CBauxexp}
    \mathcal F\lb\varepsilon\,|\,\nu\rb=\varepsilon^\alpha \Bigl[1+\sum_{k=1}^\infty
     \mathcal F_{k}\lb\nu\rb\varepsilon^{2k}\Bigr],
    \eeq
    and the coefficients $\mathcal F_{k}\lb\nu\rb$ and $\alpha$ are uniquely determined in terms of $\nu$ and  Virasoro central charge $c$.
    \end{conj}

    The only part of the above statement which is hard to prove is existence. If  $\bigl|I^{(5/2)}\bigr\rangle$ exists (which is confirmed by direct calculation at low orders in $\varepsilon$), then its uniqueness follows from that of $|\Psi\rangle$. Indeed, the irregular conformal block $\mathcal F\lb\varepsilon\,|\,\nu\rb$ can be easily computed by projecting the relation \eqref{EqForFcal} onto the vacuum  $\langle 0|$:
    \beq\label{CBalgexp}
    30\varepsilon^3\frac{\partial
    	\ln \mathcal F\lb\varepsilon\,|\,\nu\rb}{\partial\varepsilon}=\frac{\langle 0|\lb L_2+\nu\rb|\Psi\rangle}{\langle 0|J\rangle},
    \eeq
    which implies that the coefficients of the series
    \beq\label{Ugenericc}
    \mathcal U\lb \varepsilon\,|\,\nu\rb=\ln \mathcal F\lb\varepsilon\,|\,\nu\rb=\mathcal U_{\operatorname{ln}}\lb\nu\rb\ln\varepsilon+\sum_{k=1}^{\infty}\mathcal U_{k}\lb \nu\rb\varepsilon^{2k}.
    \eeq
    can be calculated as
    \beq
    \label{L2comms}
    \mathcal U_{\operatorname{ln}}\lb\nu\rb=\alpha=\frac{\langle 0|\left[L_2,\mathbb G_{2}\right]|J\rangle}{30
    	\langle 0|J\rangle},\qquad
    \mathcal U_{k}\lb\nu\rb=\frac{\langle 0|\left[L_2,\mathbb G_{2k+2}\right]|J\rangle}{60k
    	\langle 0|J\rangle},\qquad k\in\mathbb N.
    \eeq
    
    For any practical purpose, equations \eqref{CBalgexp}--\eqref{L2comms} can be considered the definition of $\mathcal F\lb\varepsilon\,|\,\nu\rb$. The Whittaker state $|\Psi\rangle$ computed up to $o\lb \varepsilon^{12}\rb$ allows us to obtain the coefficients
        \begin{subequations}
        	\label{Ucoefs}
    	\begin{align}
    \label{G2ME}\mathcal U_{\operatorname{ln}}\lb\nu\rb=&\,\nu^2+\frac{1}{30}	-\frac{7\lb c-1\rb}{360},\\
    \label{G4ME}
    \mathcal U_{1}\lb\nu\rb=&\,\frac{\nu\lb 94\nu^2+17\rb}{2}-\frac{77\nu\lb c-1\rb}{24},\\
    \label{G6ME}
    \mathcal U_{2}\lb\nu\rb=&\,\frac{38585 \nu ^4+18385 \nu ^2+336}{10}-\frac{7 (c-1) \left(1433520 \nu ^2-14497 c+124825\right)}{17280},\\
    \mathcal U_{3}\lb\nu\rb=&\,\frac{\nu  \left(5326258 \nu ^4+5019530 \nu ^2+541269\right)}{10} -\frac{49 (c-1) \nu  \left(25241040 \nu ^2-880567 c+8708575\right)}{8640},\\
    \nonumber\mathcal U_{4}\lb\nu\rb=&\,3 \left(31386901 \nu ^6+50078328 \nu ^4+14438609 \nu ^2+188160\right)+\frac{49\lb c-1\rb^2 \left(60185465 \nu ^2+7172941\right)}{960} \\&\,-\frac{7\lb c-1\rb \left(686028420 \nu ^4+509742465 \nu ^2+21788387\right)}{120}
    -\frac{791845439 \lb c-1\rb^3}{34560},\\
    \nonumber\mathcal U_{5}\lb\nu\rb=&\,\frac{3\nu\left(32160819372 \nu ^6+78779679122 \nu ^4+45102923992 \nu ^2+3752724735\right) }{5} \\
    \label{G12ME}&\,-\frac{7\lb c-1\rb \nu \left(51339657369 \nu ^4+72168517621 \nu ^2+12853630497\right) }{30}   \\
    \nonumber &\,+\frac{7 \lb c-1\rb^2 \nu  \left(1354433276168 \nu ^2-34511035493 c+636095673641\right)}{5760}. 
     \end{align}
		 \end{subequations}

    \begin{rmk} The subset of descendants contributing to  matrix elements on the right of \eqref{CBalgexp} is relatively small. For example, out of 22 Virasoro monomials in the expression \eqref{G4} of $\mathbb G_4$, non-vanishing contributions to $\mathcal U_1\lb \nu\rb$ in \eqref{G4ME} come from only 5 terms: $L_{-2}$, $L_1^2$, $L_0$, $L_0L_1^2$ and $L_0^2$. In the case of $\mathbb G_6$ there are 47 more monomials, yet only 2 of them, namely $L_0^2L_1^2$ and $L_0^3$, contribute to $\mathcal U_2\lb \nu\rb$ given by \eqref{G6ME}. In the general case, one has $2k+3$ contributions to $\mathcal U_{k}\lb \nu\rb$, since
    	\beq
    	\frac{\langle 0|\left[L_2,\mathbb G_{2k}\right]|J\rangle}{
    		\langle 0|J\rangle}=\frac{c}{2}G^{[2k]}_{\ytableausetup
    		{mathmode,  boxsize=0.4em}\ydiagram{2}\,;0,0}+\sum_{\ell=1}^k
    		2^{\ell-2}\lb G^{[2k]}_{\emptyset\,;\,\ell-1,2} -4 \nu\, G^{[2k]}_{\emptyset\,;\,\ell,0}\rb.
    	\eeq
    \end{rmk}
    
    We are now finally ready to give an algebraic variant of the proposal of \cite{PP23}.
    \begin{conj}\label{CFTPI}
    For $c=1$, irregular conformal block $\mathcal{F}\lb\varepsilon\,|\,\nu\rb$ defined by \eqref{CBalgexp}  reproduces the Painlev\'e I partition function $\mathcal Z\lb s\,|\,\nu\rb$:
    \beq\label{CFTPIstatement}
    \mathcal Z\lb s\,|\,\nu\rb=\varepsilon^{-\nu^2-\frac1{30}}
    \mathcal{F}\lb\varepsilon\,|\,\nu\rb\Bigl|_{c=1},
    \eeq
    under identification $ \varepsilon^{-2}=48is$.
    \end{conj} 	
    A straightforward check of this statement can be made by comparing the expansion coefficients \eqref{Eps1}--\eqref{Eps5} of the \PIeq free energy $\mathcal E\lb s\,|\,\nu\rb$ and the corresponding coefficients \eqref{G4ME}--\eqref{G12ME} in the expansion of $\mathcal U\lb \varepsilon\,|\,\nu\rb$.
    
    \begin{rmk} The prefactor on the right of \eqref{CFTPIstatement} was introduced to compensate $\varepsilon^{\alpha}$ in \eqref{CBauxexp} thereby ensuring that the corresponding asymptotic expansion starts from $1$. It turns out that this singular prefactor nicely combines with other quantities  in the Fourier representation of the \PIeq tau function. As $t\to -\infty$, one has simply
    \beq
    	\mathcal T\lb t\,|\,\nu\rb\simeq \lb 2\pi \rb^{-\frac{\nu}{2}}G\lb 1+\nu\rb \exp\left\{-\tfrac{\varepsilon^{-4}}{103680}+\tfrac{\nu\varepsilon^{-2}}{60}\right\}
    \mathcal{F}\lb\varepsilon\,|\,\nu\rb\Bigl|_{c=1}. 
    \eeq
    where $\mathcal T\lb t\,|\,\nu\rb$ is the function that appears in the Fourier decomposition \eqref{taufourier2}.
    \end{rmk}

     We have developed an algebraic construction of the rank $\frac52$ Whittaker state $|\Psi\rangle$ and the corresponding irregular  state $\bigl|I^{(5/2)}\bigr\rangle$ by embedding them into a Whittaker module $\mathcal V^{[2]}_{\nu}$ of rank 2. An analog of this construction in the case of embedding into a Whittaker module of rank $\frac32$ is outlined in Appendix~\ref{AppA3}.

    \section{Topological recursion\label{sec_TR}}
    The initial input to the Eynard-Orantin topological recursion (TR) \cite{EO07} is a spectral curve --- a tuple 
    $\lb\Sigma,x,y\rb$ of a compact Riemann surface $\Sigma$  with a Torelli marking (i.e. a fixed symplectic basis of the first homology group $H_1\lb\Sigma,\mathbb Z\rb$), and a pair $\lb x,y\rb$ of meromorphic functions on $\Sigma$. 
    
    In this work, the pair $\lb x,y\rb$ is assumed to satisfy an irreducible algebraic relation of the form 
    \begin{equation} \label{eq:sp-poly}
    	y^2 = P_3\lb x\rb,
    \end{equation}
    with $P_3\lb x\rb $ a cubic polynomial. 
    The equation \eqref{eq:sp-poly} defines an algebraic curve, 
    which we denote by ${\mathcal C}$, and $\Sigma$ is its normalization. 
    When there can be no confusion, ${\mathcal C}$ will also be referred to 
    as a spectral curve. We regard $x : \Sigma \to {\mathbb P}^1$ as a double covering map, 
    and denote by $\sigma : \Sigma \to \Sigma$ the covering involution
    satisfying $x\lb z\rb = x\lb\sigma\lb z\rb\rb$ and $y\lb z\rb \ne y\lb \sigma\lb z\rb\rb$ for generic $z\in\Sigma$. A zero of $dx$ is called a ramification point. The set of ramification points will be denoted by $R ~ (\subset \Sigma)$.
    
    The TR output is a doubly-indexed sequence of meromorphic multi-differentials $W_{g,n}\lb z_1,\ldots, z_n\rb$ ($g\ge0$, $n\ge1$) on $\Sigma$, called correlators, and a sequence of numbers $F_g$ ($g\ge0$), called free energies. These quantities have attracted interest of researchers 
    in various fields such as enumerative geometry, topological string theory, 
    integrable systems, quantum topology, and others.
    
      The \PIeq partition function $\mathcal Z\lb s\,|\,\nu\rb$ turns out to be related to a special case of the generating series for free energies. We start this section by a brief review of the prior results \cite{IS16} on the TR/\PIeq correspondence for $\nu=0$. In order to extend it to generic monodromy, we need an implementation of the topological recursion on the Weierstrass elliptic curve \cite{BCD18}. We then explain how the $\hbar$-expansion of the TR free energy reproduces the  large \PIeq time asymptotic expansion of $\ln\mathcal Z\lb s\,|\,\nu\rb$.
    
    \subsection{A warm-up example\label{subsec_warmup}}
    The TR implementation dramatically simplifies in the case $\Sigma=\mathbb P^1$. In our setting, this situation corresponds to a degenerate \PIeq elliptic curve \eqref{spcurve},
    \begin{equation} \label{eq:deg-ec}
    	{\mathcal C}_{\rm deg} ~:~ y^2 = 4x^3 + 2 t x + u_0  = 4\lb x-q_0\rb^2 \lb x+ 2q_0\rb,
    \end{equation}
    where $t,u_0$ are parameterized by $q_0$ as
    \begin{equation} \label{eq:u0}
    	t=-6q_0^2,\qquad u_0 = 8q_0^3.
    \end{equation}
    As one might guess, the parameter $t$ will later be  identified with the independent variable of Painlev\'e I. The quantity~$q_0$ will appear as the leading term of its formal power series solution. The spectral curve is a rational parameterization of
    $	{\mathcal C}_{\rm deg}$:
    \begin{equation} \label{eq:sp-degPI}
    	x\lb z\rb = z^2 - 2q_0, \qquad y\lb z\rb = 2z\lb z^2 - 3q_0\rb, 
    \end{equation} 
    where $z$ is the coordinate of ${\mathbb P}^1$. The involution map $\sigma$ for \eqref{eq:sp-degPI} is given by 
    	$\sigma\lb z\rb = -z$,
    which induces the covering involution $\lb x,y\rb \mapsto \lb x, -y\rb$ of \eqref{eq:deg-ec}.

    The TR algorithm works as follows:
    \begin{itemize}
    	\item Its
    starting point are the expressions
    \beq
    W_{0,1}\lb z_1\rb=y\lb z_1\rb dx\lb z_1\rb,\qquad W_{0,2}\lb z_1,z_2\rb=\frac{dz_1dz_2}{\lb z_1-z_2\rb^2}.
    \eeq
    The 1st differential is canonically associated to the spectral curve \eqref{eq:sp-poly} and the 2nd one is the unique normalized symmetric bidifferential on $\mathbb P^1$ with a double pole along the diagonal and no other poles.    
    \item The other multi-differentials  are computed by induction in $\chi_{g,n}=2g-2+n$ as 
    \begin{equation}
    	W_{g,n}\lb z_{1},\dots,z_{n}\rb = 
    	\sum_{r \in R}
    	\Res_{z=r} \, K\lb z_1,z\rb R_{g,n}\lb z, z_2, \dots, z_n\rb.  
    	\label{eq:top-rec}
    \end{equation}
    Here 
    \begin{equation} \label{eq:R-gn}
    	R_{g,n}\lb z, z_2, \dots, z_n\rb = 
    	W_{g-1,n+1}\lb z, \sigma(z),z_2, \dots, z_n\rb + 
    	\sum'_{\substack{g_{1}+g_{2}=g \\ I \sqcup J = \{2,\dots,n \}}}
    	W_{g_{1}, |I|+1}\lb z,z_{I}\rb \, W_{g_{2}, |J|+1}\lb  \sigma(z), z_{J}\rb,
    \end{equation}
    and the recursion kernel $K\lb z_1,z\rb$ is 
    given by
    \begin{equation} \label{eq:rec-ker}
    	K(z_1,z) = 
    	\frac{\displaystyle\int^{w=z}_{ w=\tilde z} W_{0,2}\lb z_{1}, w\rb}{\lb y\lb z\rb - y\lb  \sigma(z)\rb\rb  dx(z)} \, ,
    \end{equation}
    where $\tilde z$ denotes an arbitrary basepoint\footnote{One can also replace the integral
    	$\int^{w=z}_{ w=\tilde z}$ by its standard symmetrized version $\frac12\int^{w=z}_{ w=\sigma\lb z\rb}$ ; see \cite[Footnote 3]{BHLMR2013} and \cite[Remark~3.10]{BE16}. The non-symmetrized version is more straightforward to extend to genus 1 case.}.
    We use the convention for a tuple of variables as 
    $z_{I} = (z_{i_1}, \dots, z_{i_{k}})$ if 
    $I = \{i_1, \dots, i_k \}$.  
    The prime in the right hand side of \eqref{eq:R-gn} 
    means that only the indices satisfying $\lb g_i, I_i\rb \ne \lb 0, \emptyset\rb$ are taken into account
    in the summation, which implies that $W_{0,1}$ does not appear  at all in the computation of $W_{g,n}$.
    \item 
    The $g$th free energy $F_g$ is defined for $g\ge 2$ by
    \beq\label{Fgdef}
    F_g=\frac{1}{2-2g}\sum_{r\in R}\Res_{z=r} \,\Phi\lb z\rb W_{g,1}\lb z\rb,
    \eeq
    where $\Phi$ is any primitive of $W_{0,1}$. The free energies $F_0$, $F_1$ may also be defined, but in a different manner, see \cite[Section 4.2]{EO07}  for details.
    \end{itemize}
    
    In the case at hand, the set $R$ consists of only one ramification point,  $R = \left\{ 0 \right\}$, and we may set
    \beq
    K\lb z_1,z\rb=\frac{dz_1/dz}{8 z^2\lb z^2-3q_0\rb\lb z_1-z\rb},\qquad \Phi\lb z\rb=\frac{4}5z^5-4q_0z^3.
    \eeq
    The differentials $W_{g,n}$ can now be easily calculated by residues.
    For example, at the first three steps of the recursion we obtain
    \begin{subequations}
    \begin{align}	
    &\chi_{g,n}=1:\qquad 
    W_{0,3}\lb z_1,z_2,z_3\rb=\frac{dz_1dz_2dz_3}{12q_0 z_1^2z_2^2z_3^2},\qquad W_{1,1}\lb z_1\rb=\ds\frac{\lb 1+3q_0z_1^{-2}\rb dz_1}{288q_0^2z_1^2},
    \\
    &\chi_{g,n}=2:\qquad  \begin{aligned} &W_{0,4}\lb z_1,z_2,z_3,z_4\rb=\frac{1+3q_0\lb z_1^{-2}+z_2^{-2}+z_3^{-2}+z_4^{-2}\rb}{144q_0^3z_1^2z_2^2z_3^2z_4^2}dz_1dz_2dz_3dz_4,\\ &W_{1,2}\lb z_1,z_2\rb=\frac{2+6q_0\lb z_1^{-2}+z_2^{-2}\rb+3q_0^2\lb 5z_1^{-4}+5z_2^{-4}+3z_1^{-2}z_2^{-2}\rb}{3456q_0^4z_1^2z_2^2}dz_1dz_2, 
    \end{aligned}	\\
    &\chi_{g,n}=3:\qquad W_{2,1}\lb z_1\rb=\frac{7 \left(4 z_1^8+12 q_0 z_1^6+36 q_0^2 z_1^4+87 q_0^3 z_1^2+135 q_0^4\right) dz_1}{1990656 q_0^7 z_1^{10}},\qquad \ldots
    \end{align}
    \end{subequations}
    Next, using $W_{g,1}$ to compute free energies, one finds e.g. that
    \beq\label{FgDeg}
    F_{2}=\frac{7q_0^{-5}}{207360},\qquad  F_{3}=\frac{245 q_0^{-10}}{429981696},\qquad F_4=\frac{259553
    q_0^{-15}}{7430083706880 },\qquad \ldots
    \eeq
    
    Before proceeding, let us draw the reader's attention to some of the general properties of $W_{g,n}$ and $F_g$ that will be relevant for us (see \cite{EO07} for more details).
    \begin{itemize}
    	\item The differentials $W_{g,n}$ are symmetric in all variables $z_1,\ldots,z_n$, even though this symmetry is not manifest in equations \eqref{eq:top-rec}--\eqref{eq:rec-ker}.
    	\item In each  variable $z_i$, $W_{g,n}$ with $\chi_{g,n}\ge 1$ have poles only at ramification points. The pole order is at most $6g+2n-4$ and the residues at these points vanish.
    	\item $W_{g,n}$ and $F_g$ possess a homogeneity property. In the above case, it means e.g. that for  $g\ge 2$ the free energies are expressible in the form
    	\beq
    	F_g=c_gq_0^{5-5g},
    	\eeq
    	with some constant $c_g\in\mathbb Q$.
    \end{itemize}
    
    Comparing the numerators of the expressions \eqref{FgDeg} and the successive contributions to the tronqu\'ee  \PIeq  free energy in \eqref{PIFE}, one may guess that these quantities coincide under a suitable identification of parameters. Indeed, the following TR/\PIeq relation has  been proved in \cite{IS16}:
    \begin{theo}\label{thm:0-par-tau}
    The TR free energy for the curve ${\mathcal C}_{\rm deg}$ defined by \eqref{eq:deg-ec},
    \beq
    \mathcal U_{\rm \,TR}^{\rm \,deg}\lb s\rb=\sum_{g=2}^\infty F_g,\qquad s^2+432q_0^5=0,
    \eeq	
    coincides, as a formal power series in $s^{-2}$, with the \PIeq free energy $\mathcal E\lb s\,|\,0\rb$ in \eqref{PIFE}. 
    \end{theo}
    
    In the context of TR, it is customary to introduce a counting parameter $\hbar$ and define the free energy as
    $\sum_{g\ge 0}F_g\hbar^{2g-2}$. We postpone the inauguration of $\hbar$ to the next subsection, since $q_0$ in our introductory example plays a completely equivalent role. The discussion of $F_0$ and $F_1$ is also postponed to minimize the amount of fuss. We simply mention that the corresponding expressions agree with the prefactor $s^{-1/60}e^{s^2/45}$ in \eqref{bigtau1}. 
    
    Let us also mention that the degenerate curve ${\mathcal C}_{\rm deg}$ agrees with the spectral curve of the $\lb 2,3\rb$-minimal string, whose relation to \PIeq is known since \cite{BK1990,DS1990,GM1990}.

    \subsection{Weierstrass case\label{subsec_Weierstrass}}
    \subsubsection{General setup}
    Consider a smooth Weierstrass elliptic curve defined by 
    \begin{equation} \label{eq:spcurve}
    	{\mathcal C} ~:~ y^2 = 4x^3 - g_2 x - g_3, \qquad\qquad g_2=- 2 t, \quad g_3=- u,
    \end{equation}
    and fix a symplectic basis $\{A, B\}$ of $H_1({\mathcal C}, {\mathbb Z})$. The quantities  defined by the topological recursion depend on the choice of $\{A, B\}$. We will choose a specific basis when  we will discuss asymptotic expansions as $|t|\to \infty$. 
    
    Introduce the differential $W_{0,1}=ydx$ on $\mathcal C$ and consider its  normalized period 
    \begin{equation} \label{eq:nu}
   	\nu = \frac{1}{2 \pi i} \oint_A y\, dx.
   \end{equation}
    The parameter $u=u\lb t,\nu\rb$ in \eqref{eq:spcurve} will be considered as a function of two independent variables $t$ and $\nu$, 
    locally defined by the implicit relation \eqref{eq:nu}.  Throughout this subsection, 
    we assume that $\nu \ne 0$ and $\lb t,\nu\rb$ vary 
    in a domain where the discriminant 
    \begin{equation} \label{eq:discriminant}
    	\Delta = g_2^3-27g_3^2=-8t^3 -27 u^2 
    \end{equation}
    does not vanish. We use the following notation for the periods of the differentials of the first and second kind on $\mathcal C$:
     \begin{equation}\label{periodsAB}
     	\omega_{\ast} = \oint_{\ast} \frac{dx}{y}
     	,\qquad
     	\eta_{\ast} = - \oint_{\ast} \frac{x\, dx}{y},
     	\qquad \ast = A, B .
     \end{equation}
     These periods are holomorphic in both $t$ and $\nu$ 
     on the domain which we consider. 
     
     Let $L = {\mathbb Z} \omega_A + {\mathbb Z} \omega_B$ be the lattice generated by 
     the periods $\omega_{A,B}$, and let
     \begin{equation}
     	\wp(z) = \frac{1}{z^2} 
     	+ \sum_{\omega \in L \setminus \{(0,0) \}} 
     	\left( \frac{1}{\lb z - \omega\rb^2} - \frac{1}{\omega^2} \right) 
     \end{equation}
     denote the associated Weierstrass $\wp$-function. 
     Then, the (non-degenerate) \PIeq spectral curve is given by 
     \begin{equation} \label{eq:sp-PI}
     	\Sigma = {\mathbb C}/L\,, \qquad x\lb z\rb = \wp\lb z\rb, \qquad y\lb z\rb = \wp'\lb z\rb.
     \end{equation}
     Indeed, it is well-known that \eqref{eq:sp-PI} gives a meromorphic parameterization, 
     known as the Weierstarss parameterization, of the elliptic curve \eqref{eq:spcurve}. 
     The involution map 
     \begin{equation}
     	\sigma\lb z\rb \equiv - z  \mod L
     \end{equation}
     also induces the covering involution $\lb x,y\rb \mapsto \lb x, -y\rb$ of \eqref{eq:spcurve}.
     In this case, the set $R$ of ramification points 
     consists of half-periods: 
     \begin{equation} \label{eq:r123}
     	R = \left\{ r_1, r_2, r_3 \right\},\qquad\text{with}\qquad
     	r_1 = \frac{\omega_A}{2}, \quad
     	r_2 = \frac{\omega_B}{2}, \quad
     	r_3 = \frac{\omega_A+\omega_B}{2}.
     \end{equation} 
     Note that the $\wp$-function values at half-periods, $e_i = \wp\lb r_i\rb$ ($i=1,2,3$), give the three simple zeros 
     of the right hand side of \eqref{eq:spcurve} and satisfy ${e_1+e_2+e_3=0}$.   
     The $A$-cycle specified in the definition of the spectral curve is represented by 
     the closed loops encircling $e_2$ and $e_3$.
Sometimes it will be convenient for us to parameterize the curve $\mathcal{C}$, its discriminant $\Delta$ and elliptic invariants $g_2$, $g_3$ as
     \beq\label{invariants_es}
     \begin{gathered}
     y^2=4\lb x-e_1\rb\lb x-e_2\rb\lb x-e_3\rb,\qquad
     \Delta=16\lb e_1-e_2\rb^2 \lb e_2-e_3\rb^2 \lb e_3-e_1\rb^2,
     \\ g_{2}=-4\lb e_1e_2+e_2e_3+e_3e_1\rb,\qquad g_3=4e_1e_2e_3.
     \end{gathered}
     \eeq
     
     The last ingredient we need to be able to run the topological recursion is an expression for $W_{0,2}$. In general, for a given spectral curve $\lb \Sigma,x,y\rb$, it is given by
     a meromorphic fundamental bidifferential 
     $B = B\lb z_1,z_2\rb$,  
     called the Bergman bidifferential. It is uniquely defined by the following conditions: (i) $B$ is symmetric, (ii) all of its $A$-periods  vanish, and (iii) $B$ has a double pole along the diagonal and no other poles, and such that in any local coordinate $z$, 
     	\begin{equation}
     		B\lb z_1, z_2\rb \sim \left( \frac{1}{\lb z_1 - z_2\rb^2} + {\text{regular}
     	} \right) dz_1 dz_2,\qquad \text{as }z_1\to z_2.
     	\end{equation}
     In the case of the Weierstrass elliptic curve $\mathcal C$, this leads to the expression
     \beq\label{W02ell}
     W_{0,2}\lb z_1,z_2\rb=\lb \wp\lb z_1-z_2\rb+S\rb dz_1dz_2,
     \qquad S=\frac{\eta_A}{\omega_A},
     \eeq
     where $\omega_A,\eta_A$ are defined in \eqref{periodsAB}. The rest of the 
     TR equations of the previous subsection are universal. In particular,  it follows from \eqref{eq:rec-ker} and \eqref{W02ell} that the recursion kernel is given by
     \beq
     K\lb z_1,z\rb=\frac{\lb P\lb z_1-\tilde z\rb-P\lb z_1-z\rb\rb dz_1}{2\wp'\lb z\rb^2dz},\qquad P\lb z\rb=-\zeta\lb z\rb+S z.
     \eeq
     Here 
     \beq
     \zeta\lb z\rb=\frac{1}{z} 
     + \sum_{\omega \in L \setminus \{(0,0) \}} 
     \left( \frac{1}{z - \omega} + \frac{1}{\omega} +\frac{z}{\omega^2} \right) 
     \eeq
     is the Weierstrass zeta function satisfying $\zeta'\lb z\rb=-\wp\lb z\rb$. There is however a subtlety: one has to fix a fundamental domain $D$ by a choice of representatives of the homology basis $\left\{A,B\right\}$, and require that $\tilde z,z\in D$.

     \subsubsection{Correlators}
     In order to illustrate the issues one has to deal with in the elliptic case, let us compute the differentials $W_{0,3}$ and $W_{1,1}$ which correspond to $\chi_{g,n}=1$. From the TR defining relations \eqref{eq:top-rec}--\eqref{eq:R-gn}, we find e.g.
     \begin{align}
     \begin{aligned}
     &W_{0,3}\lb z_1,z_2,z_3\rb=\sum_{i=1,2,3}\Res_{z=r_i} K\lb z_1,z\rb 
     \lb W_{0,2}\lb z,z_2\rb  W_{0,2}\lb -z,z_3\rb+W_{0,2}\lb z,z_3\rb  W_{0,2}\lb -z,z_2\rb\rb=\\
     &=\sum_{i=1,2,3}\Res_{z=r_i} \frac{\lb P\lb z_1-z\rb- P\lb z_1-\tilde z\rb\rb\lb\wp\lb z-z_2\rb+S\rb \lb\wp\lb z+z_3\rb+S\rb}{2\wp'\lb z\rb^2}dz_1dz_2dz_3+\lb z_2\leftrightarrow z_3\rb.
     \end{aligned}
     \end{align}
     The denominator of the above expression has a double zero at every ramification point $r_i\in R$. Moreover, since both $\wp'\lb z\rb^2$ and
     $\lb\wp\lb z-z_2\rb+S\rb \lb\wp\lb z+z_3\rb+S\rb+\lb z_2\leftrightarrow z_3\rb$
     are even elliptic functions of $z$, their expansions around $z=r_i$ contain only even powers of $z-r_i$. The contribution coming from the residue at $z=r_i$ is therefore determined by the derivative $P'_z\lb z_1-z\rb\bigl|_{z=r_i}=-\wp\lb z_1-r_i\rb-S$, so that
     \beq\label{W03ell01}
     W_{0,3}\lb z_1,z_2,z_3\rb=\frac{4}{\Delta}\sum_{i=1,2,3}
     \lb 3e_i^2-g_2\rb \lb \wp\lb z_1-r_i\rb+S\rb \lb \wp\lb z_2-r_i\rb+S\rb \lb \wp\lb z_3-r_i\rb+S\rb  dz_1dz_2dz_3.
     \eeq
     In a completely similar manner, one finds
     \begin{align}
     \nonumber	W_{1,1}\lb z_1\rb
     	=&\,\sum_{i=1,2,3}\Res_{z=r_i} K\lb z_1,z\rb W_{0,2}\lb z,\sigma\lb z\rb\rb=
     	\sum_{i=1,2,3}\Res_{z=r_i}\frac{
     		\lb P\lb z_1-z\rb-P\lb z_1-\tilde z\rb\rb\lb\wp\lb 2z\rb+S\rb}{2\wp'\lb z\rb^2}dz_1=\\
     \label{W11ell01}		
     	=&\,\frac{2}{\Delta}\sum_{i=1,2,3}
     	\lb 3e_i^2-g_2\rb \lb \frac{\wp''\lb z_1-r_i\rb}{24}+\lb S-e_i\rb\lb \wp\lb z_1-r_i\rb+S\rb\rb dz_1,
     \end{align}
     At the last step of the residue computation leading to \eqref{W03ell01} and \eqref{W11ell01}, we have used the identities
     \beq
     \frac{1}{\wp''\lb r_i\rb^{2}}=\frac{4g_2-12e_i^2}{\Delta},
     \qquad  e_i=\ds \frac{\wp^{(4)}\lb r_i\rb}{12\wp''\lb r_i\rb},
     \eeq
      which can be deduced from the  differential equation $\wp''\lb z\rb=6\wp^2\lb z\rb-g_2/2$. 
      
      The formulas \eqref{W03ell01}--\eqref{W11ell01} reproduce   \cite[Eqs. (A.1)--(A.2)]{BCD18} up to an overall sign\footnote{It is not clear to authors whether this typo also affects higher-level differentials presented in \cite{BCD18}.}. This is, however, not the end of the story, since we need to express the result in terms of elliptic invariants $g_{2,3}$ instead of the ramification points $e_{1,2,3}$. Using the addition formula for $\wp\lb z\rb$, we can write, e.g.
      \beq
      \wp\lb z-r_1\rb-e_1=\frac{12e_1^2-g_2}{4\lb\wp\lb z\rb-e_1\rb}=\frac{\lb 12e_1^2-g_2\rb\lb x\lb z\rb-e_2\rb \lb x\lb z\rb-e_3\rb}{y^2\lb z\rb}.
      \eeq
      After substitution of such expressions into \eqref{W03ell01}--\eqref{W11ell01},  $W_{0,3}$ and $W_{1,1}$ transform into  linear combinations of differentials whose coefficients are symmetric polynomials in $e_{1,2,3}$. The latter can be algorithmically expressed as polynomials in $g_{2,3}$ using the relations \eqref{invariants_es} and  $4e_i^3= g_2 e_i+g_3$. E.g., $W_{1,1}$ is given by
      \beq\label{W11ell02}
      W_{1,1}\lb z_1\rb=\lb-\frac{12g_2 x\lb z_1\rb^2+36g_3 x\lb z_1\rb+g_2^2}{32y\lb z_1\rb^4}-\frac{x\lb z_1\rb+4S}{8 y\lb z_1\rb^2}+\frac{g_2\lb g_2-12S^2\rb}{4\Delta}\rb dz_1.
      \eeq
      
      The  procedure outlined above can in principle be continued at higher values of $\chi_{g,n}$. The analogs of \eqref{W03ell01} and \eqref{W11ell01} are relatively easy to compute from truncated Laurent expansions of $K\lb z_1,z\rb$ and $W_{g,n}\lb z_1,\ldots,z_n\rb$ at the ramification points. The most computationally involved step is (by far!) the  final symmetrization. We have calculated the  expressions of $W_{0,3}$, $W_{1,2}$ and $W_{2,1}$, but even for $W_{0,3}$ the counterpart of \eqref{W11ell02} is too lengthy to be usefully displayed here.
      
      To exhibit the homogeneity of $W_{g,n}$, let us define a scaling operator 
      \begin{equation} \label{eq:degree}
      	\varrho_r : \lb z,t,\nu\rb \mapsto \lb r^{- \frac{1}{5}} z, r^{\frac{4}{5}} t, r \, \nu \rb,
      	\qquad r \in {\mathbb R}_{>0}.
      \end{equation}
      Recall that $t=-g_2/2$, $u=-g_3$ and note that e.g. $\varrho_r\lb u,S\rb=\bigl( r^{\frac65}u,r^{\frac25}S\bigr)$. For a function $f\lb z,t,\nu\rb$ and a differential $f\lb z,t,\nu\rb \, dz$, we define 
      \begin{equation}
      	\lb \varrho_r\rb_\ast f\lb z,t,\nu\rb = f\lb r^{-\frac{1}{5}} z, r^{\frac{4}{5}} t, r \, \nu\rb, 
      	\qquad 
      	\lb \varrho_r\rb_\ast f\lb z,t,\nu\rb \, dz = r^{-\frac{1}{5}}f\lb r^{-\frac{1}{5}} z, r^{\frac{4}{5}} t, r \, \nu\rb \, dz, 
      \end{equation}
      and generalize this to multi-differentials in a natural way. Then we have 
      \begin{equation} \label{eq:degree-Wgn}
      	\lb \varrho_r\rb_\ast W_{g,n} = r^{-\chi_{g,n}} \, W_{g,n},
      \end{equation}
      which can be proved by induction in $\chi_{g,n}$.
      
      The definition \eqref{W02ell} of $W_{0,2}$ 
      and the recursion relation \eqref{eq:top-rec} imply that $W_{g,n}$ with $\chi_{g,n} \ge 0$ are polynomial in $S$.
      The general structure of $W_{g,n}$ is described by 
      
      \begin{prop} \label{prop:pre-quasi-modular}
      		The correlator $W_{g,n}$ of the spectral curve \eqref{eq:sp-PI} with $\chi_{g,n} \ge 1$ 
      		has an expression of the form 
      		\begin{equation} \label{eq:qm-Wgn}
      			W_{g,n}\lb z_1,\ldots, z_n\rb = \frac{w_{g,n}\lb t,u,x\lb z_1\rb, \ldots, x\lb z_n\rb, S\rb}
      			{\Delta^{2g-2+n}  \prod_{i=1}^{n} y\lb z_i\rb^{6g-4+2n}} \, dz_1 \ldots dz_n.
      		\end{equation}
      		Here, the numerator takes the form 
      		\begin{equation} \label{eq:qm-Wgnk}
      			w_{g,n} = \sum_{k=0}^{3g-3+2n} w_{g,n,k}
      			\lb t,u,x\lb z_1\rb, \ldots, x\lb z_n\rb\rb 
      			S^k, 
      		\end{equation}
      		where $w_{g,n,k}$ is a polynomial in its arguments with constant coefficients, 
      		and satisfies 
      		\begin{equation} \label{eq:Pol-gnk-degree}
      			\lb \varrho_r\rb_\ast w_{g, n, k} 
      			= r^{\frac{1}{5}\lb 14g-14-2k+2n(9g+3n-2)\rb} \, w_{g,n, k}.
      		\end{equation}
      \end{prop}
      
      \pf
      	A proof of a similar result was given in \cite[Theorem 4.4]{FRZZ19}, where 
      	the mirror curves of some local Calabi--Yau 3-folds were discussed.
      	Although our setting is different from \cite{FRZZ19}, 
      	the proof can be carried out along the same lines. 
      	For convenience of the reader, we give a sketch of it. 
      	
      	It may be shown   by induction that $W_{g,n}$ is a polynomial in $S$ 
      	with degree at most $3g-3+2n$.
      	When expressing $W_{g,n} / \lb dz_1 \ldots dz_n\rb$ as a rational function of $x\lb z_i\rb$ and $y\lb z_i\rb$, 
      	it is understood from the above considerations that the denominator is 
      	represented by even powers of $y\lb z_i\rb \, (= x'\lb z_i\rb$), and 
      	the numerator is a polynomial in $x\lb z_i\rb$.
      	When calculating according to the TR formula \eqref{eq:top-rec}, 
      	the discriminant of the denominator of \eqref{eq:qm-Wgn} arises 
      	when putting the residues at ramification points over a common denominator 
      	(by multiplying the series coefficients of $1/y(z)$ in the recursion kernel \eqref{eq:rec-ker}). 
      	Its degree, equal to the recursion level $\chi_{g,n}=2g-2+n$, can also be verified by induction. 
      	The relation \eqref{eq:Pol-gnk-degree} follows from the homogeneity \eqref{eq:degree-Wgn}. 
      \epf
     
      \subsubsection{Free energies and quasi-modular forms}
      For generic \PIeq  monodromy data, the contribution of the free energies $F_0$, $F_1$ to the function $\mathcal Z\lb s\,|\, \nu\rb$ turns out to be nontrivial. These quantities are separately defined in \cite{EO07} as
      \begin{equation} \label{eq:expression-F0-F1}
      	F_0  = \frac{g_2 g_3}{10} + \frac{\nu}{2} \, \oint_B y dx, \qquad
      	F_1 = - \frac{1}{12} \ln \left( \omega_A^6 \, {\Delta} \right).
      \end{equation}
      
      The free energies  $F_{g}$ with $g\ge 2$ are determined by the same universal formula \eqref{Fgdef} as in the degenerate case, with the primitive $\Phi\lb z\rb$ now given by 
      \begin{lemma}
      	We have
      \begin{equation}\label{PhiEll}
      	\Phi\lb z\rb=\int^z y\,dx =\frac{1}{5}\lb 2\wp\lb z\rb\wp'\lb z\rb+2g_2\zeta\lb z\rb-3g_3z\rb.    
      \end{equation}
      \end{lemma}
      \pf Differentiate the right hand side and then use the equations $\wp''=6\wp^2-g_2/2$ and $\lb\wp'\rb^2=4\wp^3-g_2\wp-g_3$ to check that the result is indeed $\wp'\lb z\rb^2=y\lb z\rb x'\lb z\rb$. \epf
      \begin{cor}\label{Cor_W01periods}
      	The periods of $W_{0,1}$ can be expressed as
      	\beq\label{W01periods}
      	\oint_* ydx=\frac{2g_2\eta_*-3g_3\omega_*}{5},\qquad *=A,B.
      	\eeq
      \end{cor}
      \pf Straightforward consequence of \eqref{PhiEll} and quasiperiodicity relations  $\zeta\lb z+\omega_*\rb=\zeta\lb z\rb+\eta_*$.  \epf
      
      Using $\Phi\lb z\rb$, we have obtained surprisingly simple explicit answers for $F_2$ and $F_3$:
       \begin{subequations}
       	\label{F2and3}
     \begin{align}
     	&F_2=\frac{1}{\Delta^2}\left\{
     	\frac{15 }{2}g_2^2 S^3
     	+\frac{297 }{4 }g_2 g_3 S^2
     	+\frac{3}{8 } \left(14 g_2^3+297 g_3^2\right)S
     	+\frac{897 }{80 }g_2^2 g_3\right\},\\
     	&\label{F3ell}
     	\begin{aligned}
     	F_3=&\,\frac{1}{\Delta^4}\left\{405 g_2^4 S^6
     	+9720 g_2^3 g_3 S^5
     	+\frac{81}{4 } \left(56 g_2^3+4293 g_3^2\right) g_2^2 S^4
     	+\frac{27  }
     	{4 } \left(3121 g_2^3+49005 g_3^2\right)g_2 g_3 S^3\right.\\
     	&\qquad +\frac{9  }{80 }\left(1079028 g_2^3 g_3^2+8023 g_2^6+3207600 g_3^4\right)S^2
     	+\frac{729  }
     	{80}\left(899 g_2^3+22463 g_3^2\right)g_2^2 g_3 S\\
     	&\left.\qquad +\frac{1}
     	{2240 } \left(29176497 g_2^3 g_3^2+171350 g_2^6+174573630 g_3^4\right)g_2\right\}.
     	\end{aligned}
     \end{align}
     \end{subequations}
     The first-principles TR derivation of the last equation requires the knowledge of $W_{3,1}$. Although we have not succeeded to compute a symmetrized version of this differential similar to \eqref{W11ell02}, it is still possible to find an analog of the representation \eqref{W11ell01} in terms of $e_{1,2,3}$ and plug it into \eqref{Fgdef} to compute $F_3$. The  resulting 6-fold sum over  ramification points coming from the iterated residue computation can then be symmetrized (i.e. expressed in terms of $g_{2,3}$) which leads to \eqref{F3ell}.
     
     For free energies, the homogeneity property takes the form 
     \begin{equation} \label{eq:degree-Fg}
     	\lb \varrho_r\rb_\ast F_g = r^{-2\lb g-1\rb} F_g.
     \end{equation}
     In the special case $g=1$ it is understood that $\lb \varrho_r\rb_\ast F_1$ is equal to $F_1$ up to an additive constant.
     Proposition~\ref{prop:pre-quasi-modular} implies
     that for  $g \ge 2$ the  free energy $F_g$ of the Weierstrass curve $\mathcal C$
     has an expression of the form 
     \begin{equation} \label{eq:rational-expression-Fg}
     	F_{g} = \frac{1} {\Delta^{2g-2}}\sum_{k=0}^{3g-3} f_{g, k}\lb t,u\rb 
     	S^k,
     \end{equation}
     where $f_{g, k}$ is a polynomial in $t$ and $u$ (equivalently, $g_2$ and $g_3$) with constant coefficients
     satisfying 
     \begin{equation} \label{eq:degree-Pol-gk}
     	\lb \varrho_r\rb_\ast f_{g, k} 
     	= r^{\frac{1}{5}(14g-14-2k)} \, f_{g, k}.
     \end{equation}

     In order to fully benefit from the homogeneity, we now define another parameterization. Introduce the elliptic nome and period ratio by
     \beq\label{nome_ratio}
     q=e^{2\pi i\tau},\qquad \tau = \frac{\omega_B}{\omega_A},\qquad \Im \tau>0.
     \eeq
     Note that we use the conventions where $q$ is the square of the usual nome. The quantities $S$, $g_2$ and $g_3$ parameterizing TR free energies with $g\ge2$  can then be rewritten as
     \begin{subequations}\label{eq:t-u-and-eisenstein} 
     \begin{align} 
     	\label{XE2}
     	S =&\, \frac {\pi^2 }{3 \omega_A^2}E_2,\\
     \label{g2E4}	g_2 =&\, \frac {4\pi^4 }{3 \omega_A^4} E_4,\\
     \label{g3E6}	g_3 =&\, \frac {8\pi^6 }{27 \omega_A^6}E_6,
     \end{align}
     \end{subequations}
     where $E_{2n}=E_{2n}\lb q\rb$ denote the Eisenstein series defined by
     \begin{align}\label{eisensteindef}
     	E_{2n}\lb q\rb = 1 - \frac{4n}{B_{2n}} \sum_{k=1}^{\infty} \frac{k^{2n-1} q^k}{1 - q^k},
     	\qquad n\in \mathbb N. 
     \end{align}
     In this definition, $B_{2n}$ stand for Bernoulli numbers (hence e.g. $B_2=\frac16$, $B_4=-\frac1{30}$, $B_6=\frac{1}{42}$). It is useful to keep in mind that later in the paper we will be interested in the asymptotic regime where $q$ is small. 
     
     The Eisenstein series  $E_2$ is a quasi-modular form with weight 2 and depth 1, while $E_4$ and $E_6$ are modular forms with weights $4$ and $6$, respectively. This means that they have a certain behavior under the transformations  
     \beq
     \tau\mapsto \frac{\alpha\tau+\beta}{\gamma\tau+\delta},\qquad
     \lb\begin{array}{cc}\alpha  & \beta \\ \gamma & \delta\end{array}\rb\in\mathrm{SL}_2\lb\mathbb Z\rb.
     \eeq

     Using \eqref{nome_ratio}--\eqref{eq:t-u-and-eisenstein} and Corollary~\ref{Cor_W01periods}, the normalized $A$-period of $W_{0,1}$-differential (which will play a role of the monodromy parameter in the relation to \PIeq) can now be expressed as
     \beq\label{numodular}
     \nu=\lb\frac{2\pi i}{\omega_A}\rb^5\frac{E_6-E_2E_4}{360}.
     \eeq
     Similar calculations, together with the Riemann bilinear identity $\eta_A\omega_B-\eta_B\omega_A=2\pi i$, lead to the following representation of the genus 0 free energy:
     \beq\label{F0modular}
     F_0=\nu^2\left[i\pi\tau+\frac{E_4\lb 6E_2E_4-11E_6\rb}{\lb E_6-E_2E_4\rb^2}\right].
     \eeq
     Let us also express the discriminant  $\Delta = g_2^3-27 g_3^2$
     of the elliptic curve $\mathcal C$  as
     \begin{equation} \label{eq:delta-eta}
     	\Delta = 
     	 \left( \frac{2\pi i}{\omega_A} \right)^{12} 
     	\frac{E_4^3 - E_6^2}{1728} 
     	= \left( \frac{2\pi i}{\omega_A} \right)^{12} \eta^{24},
     \end{equation}
     where $\eta = \eta\lb q\rb$ is the Dedekind eta function:
     \begin{equation}
     	\eta\lb q\rb = q^{\frac{1}{24}} \, \prod_{k=1}^{\infty} \lb 1 - q^k\rb.
     \end{equation}
     The free energy $F_1$ is therefore given (here and below, $F_1$ will be considered only up to an irrelevant additive numerical constant) by
     \beq\label{F1modular}
     F_1=-\frac1{12}\ln \lb\frac{2\pi}{\omega_A}\rb^6\lb E_4^3-E_6^2\rb.
     \eeq

        The quasi-modular structure of  $F_{g}$ with $g\ge 2$ is described by the following statement.
    \begin{prop}[{cf. \cite[Theorem 4.4]{FRZZ19}}]\label{prop_FgStructure}
    	The $g$th free energy of the spectral curve \eqref{eq:sp-PI} with $g \ge 2$ has the following expression: 
    	\begin{equation} \label{eq:Fg-quasi-modular-expression}
    		F_g = \left( \frac{\xi}{E_4} \right)^{2g-2} P_{g}\lb E_2, E_4, E_6\rb,
    	\end{equation}
    	where 
    	\begin{equation}\label{eq:Yukawa}
    		\xi=
    		- \frac{1}{2 \eta^{24}} \, \left( \frac{\omega_A}{2\pi i} \right)^5 \, {E_4},
    	\end{equation}
    	and $P_{g}$ is a polynomial in its arguments of the following form
    	\begin{equation} \label{eq:quasi-modular-part}
    		P_{g}\lb E_2, E_4, E_6\rb = 
    		\sum_{k = 0}^{3\lb g-1\rb} 
    		\Biggl( \sum_{\substack{a, b \ge 0 \\ 4a + 6b + 2k= 14\lb g-1\rb}} 
    		C_{a,b,k} \, E_4^a E_6^b \; \Biggr) \; E_2^k,
    	\end{equation}
    	with some constants $C_{a,b,k} \in {\mathbb C}$. 
    	In particular, $P_{g}\lb E_2, E_4, E_6\rb$ is 
    	a quasi-modular form of weight $14\lb  g-1\rb$ and depth $3\lb g-1\rb$.
    \end{prop}
    \pf Straightforward consequence of Proposition \ref{prop:pre-quasi-modular}, quasi-modular expressions \eqref{eq:t-u-and-eisenstein}, \eqref{eq:delta-eta}, and scaling behavior $\lb \varrho_r\rb_\ast \omega_A=r^{-\frac15}\omega_A$. \epf
    \begin{eg} \label{eg:P2-P3}
    	It follows from the expressions \eqref{F2and3}  that
    	\begin{subequations}
    		\label{F23modular}
    	\begin{align}
    		\label{F23modularA}&P_2\lb E_2,E_4, E_6\rb=-\frac{25 E_4^2 E_2^3 +165 E_4 E_6 E_2^2+15  \left(14 E_4^3+11 E_6^2\right) E_2+299 E_4^2 E_6}{207360},\\
    		\label{F23modularB}&\begin{aligned}P_3\lb E_2,E_4, E_6\rb=&\;\Bigl(525 E_4^4 E_2^6 +8400 E_4^3 E_6 E_2^5+315 E_4^2 \left(56 E_4^3+159 E_6^2\right) E_2^4 +70 E_4 E_6 \left(3121 E_4^3+1815 E_6^2\right)E_2^3\Bigr. \\
    			&\,+21 \left(8023 E_4^6+39964 E_4^3 E_6^2+4400 E_6^4\right)E_2^2+42  E_4^2 E_6 \left(24273 E_4^3+22463 E_6^2\right)E_2\\
    			&\,\Bigl.+\left(171350 E_4^6+1080611 E_4^3 E_6^2+239470 E_6^4\right) E_4\Bigr)\bigl/5016453120\bigr. .
    			\end{aligned}
    	\end{align}
    	 \end{subequations}
    \end{eg}
    
    The above discussion can be summarized as follows. We would like to compute the TR free energies $F_g=F_g\lb t,\nu\rb$ of the Weierstrass elliptic curve $\mathcal C$ as functions of a pair of parameters  $\lb t,\nu\rb$, where $t$ is one of the two curve moduli and $\nu$ is one of the two periods of $W_{0,1}=ydx$. One may use another pair of parameters, $\lb q,\omega_A\rb$, where $q$ is the elliptic nome \eqref{nome_ratio} and $\omega_A$ is the period defined by \eqref{periodsAB}. The relation between the two pairs is given by  \eqref{g2E4} and \eqref{numodular}. The main advantage of the 2nd pair is that the dependence of $F_g$ on $\omega_A$ and $q$ has a factorized form (except for $F_1$ for which the separation is additive) described by \eqref{eq:Fg-quasi-modular-expression}--\eqref{eq:quasi-modular-part}.

    \subsection{TR/\PIeq correspondence\label{subsec_TRPI}}

    Let us introduce an auxiliary parameter $\hbar$ and define the TR partition function
    \begin{equation} 
    	\mathcal Z_{\mathrm{TR}}\lb t,\nu; \hbar\rb = \exp \sum_{g \ge 0} \hbar^{2g-2} F_g\lb t,\nu\rb. 
    	\label{eq:TR-Z}
    \end{equation} 
    
    The main result on the TR/\PIeq correspondence is the following.
    \begin{theo}[{\cite[Theorem 4.3]{Iwaki19}}] \label{thm:2-par-tau}
    	The Fourier series 
    	\begin{equation} \label{eq:TR-tau}
    		\tau_{\mathrm{TR}}\lb t,\nu,\rho;\hbar\rb
    		= \sum_{k \in {\mathbb Z}} e^{2 \pi i k \rho / \hbar} 
    		\, \mathcal Z_{\mathrm{TR}}\lb t, \nu + k \hbar; \hbar\rb
    	\end{equation}
    	gives a formal series-valued tau function for $\hbar$-Painlev\'e I equation. That is, 
    	 \begin{equation}
    		 q\lb t, \nu, \rho; \hbar\rb = - \hbar^2 \frac{d^2}{dt^2} \ln \, \tau_{\rm TR}\lb t,\nu,\rho;\hbar\rb 
    		 \end{equation}
    	 gives a formal solution of 
    	 \begin{equation}\label{hbarPI}
    		 \hbar^2 \dfrac{d^2q}{dt^2} = 6 q^2 + t.
    		 \end{equation} 
    \end{theo}

    The parameters $\nu$ and $\rho$ in \eqref{eq:TR-tau} are regarded as arbitrary integration constants, and hence, \eqref{eq:TR-tau} gives the tau function corresponding to general 2-parameter formal solution of \PIeq. 
    The above result can be considered analogous to the construction of Painlev\'e tau functions by conformal blocks \cite{GIL12, GIL13,ILT14, Nagoya15, Nagoya18} and partition functions of $\mathcal N=2$ supersymmetric 4D gauge theories \cite{BLMST}. 
    We also note that the TR construction of tau functions has been  generalized 
    to other Painlev\'e equations \cite{EGF, MO19, EGFMO}. 
    
    Note also that this representation can be written as a 
    formal power series in $\hbar$ whose coefficients 
    are written as differential polynomials of 
    the $\vartheta$-function associated with the spectral curve; 
    this is also known as the non-perturbative partition function introduced in \cite{emo}. 
    It also follows that the corresponding 2-parameter formal solution 
    of \PIeq has a series expansion whose leading term involves  $\wp$-functions 
    which appeared in the definition \eqref{eq:sp-PI} of the \PIeq spectral curve. They 
    faithfully reproduce the form of Boutroux's elliptic asymptotic solutions
    \cite{Bou1913}.
    
    The key to the proof of Theorem \ref{thm:2-par-tau} as well as its degenerate version, Theorem \ref{thm:0-par-tau}, lies in a relationship between the 
    isomonodromic deformations 
    of linear differential equations and TR via quantum curves and exact WKB analysis.  
    In fact, a conjectural formula for explicit description of the Stokes data of the linear system associated with \PIeq by using $(\nu, \rho)$ was proposed in \cite[Section 5]{Iwaki19}, under the assumption of Borel summability of the partition function, which has not yet been rigorously proven. The reader is referred to  \cite{IS16, Iwaki19} for more details.
    We note that this observation is consistent with the discussion in Section \ref{sec_FT} on the proof of existence for the Fourier series expression of the tau function.

    Actually, $\hbar$ in \eqref{hbarPI} is not a genuine new parameter since it can be absorbed by rescaling $\lb t,q\rb \mapsto \lb \hbar^{\frac45}t,\hbar^{\frac25} q\rb$. One may therefore expect a relation between the TR  asymptotic series  $\mathcal Z_{\mathrm{TR}}\lb t,\nu; \hbar\rb$ in $\hbar$ defined by \eqref{eq:TR-Z}  and the Painlev\'e I  asymptotic series $\tilde{\mathcal Z}\lb s\,|\,\nu\rb$ in $s^{-1}$ defined by \eqref{bigtau1}. 
    Recall that $s=24^{\frac14}\lb-t\rb^{\frac54}$ stands for \PIeq time parameterization from  \eqref{bigtau2}.
      We are now going to explain that, indeed,
    \beq\label{TRPIshort}
    \mathcal Z_{\mathrm{TR}}\lb t,\nu; \hbar\rb\simeq \tilde{\mathcal Z}\lb \hbar^{-1}s\,|\,\hbar^{-1}\nu\rb,
    \eeq
    in a very precise sense. More explicitly, we want to show that 
    \beq\label{TRPIvell}
    \sum_{g \ge 0} \hbar^{2g-2} F_g\lb t,\nu\rb\simeq\frac{1}{\hbar^2}\left[ \frac{s^2}{45}+\frac{4i\nu s}{5}-\frac{\nu^2}{2}\ln48is\right]+\ln G\lb 1+\frac{\nu}{\hbar}\rb-\frac{\nu\ln2\pi}{2\hbar}
    -\frac{1}{60}\ln\frac{s}{\hbar}+\mathcal E\lb\frac{s}{\hbar}\,\Bigl|\,\frac{\nu}{\hbar}\rb,
    \eeq
    where $\mathcal E\lb s\,|\,\nu\rb$ is the \PIeq free energy introduced in \eqref{PIFEnergy}, and the left hand side is interpreted as a double expansion in both $\hbar$ and $s^{-1}$.

    Our first task then is to compute the large $t$ expansion of $F_{g}\lb t,\nu\rb$ at fixed $\nu$. To this end, let us express $\omega_A$ from \eqref{g2E4} in terms of $t=-g_2/2$ and $E_4$:
    \beq\label{omegaAequation}
    \frac{\omega_A}{2\pi }=-\left[\frac{E_4}{24\lb -t\rb}\right]^{\frac14}.
    \eeq
    Here we have picked one of the four roots of \eqref{g2E4}; the other three options lead to equivalent results. Substituting \eqref{omegaAequation} into \eqref{numodular}, we find that
    \beq\label{nome_equation}
    \frac{i\nu}{s}=\frac{\lb E_6-E_2E_4\rb}{15 E_4^{5/4}}.
    \eeq
    Equation \eqref{nome_equation} should now be solved for the nome $q$ in terms of $\nu/s$. The solution may be written in the form of a convergent expansion  whose coefficients are readily found from the Eisenstein series definition \eqref{eisensteindef}. We have, e.g. 
    \beq\label{nome_expansion}
    q=-\frac{1}{48}\frac{i\nu}{s}-\frac{47}{384}\frac{\nu^2}{s^2}+\frac{1793 }{2304}\frac{i\nu^3}{s^3}+\frac{6861827 }{1327104}\frac{\nu^4}{s^4}-\frac{1500888517 }{42467328}\frac{i\nu^5}{s^5}+O\lb s^{-6}\rb,
    \eeq
    as $s\to\infty$. Substituting the expansion of $q$ back into \eqref{omegaAequation}, we determine the large $s$ expansion of $\omega_A$:
    \beq
    \lb\frac{\omega_A}{2\pi }\rb^5=
    \frac{\lb-1\rb}{24s}\left[1-\frac{25 }{4}\frac{i\nu}{s}-\frac{2675 }{64}\frac{\nu^2}{s^2}+\frac{667675 }{2304}\frac{i\nu^3}{s^3}+\frac{303033725 }{147456}\frac{\nu^4}{s^4}-\frac{2182577825}{147456}\frac{i\nu^5}{s^5}+O\lb s^{-6}\rb\right].
    \eeq
    Finally, plugging $\omega_A$ and $q$ into the formulas \eqref{F0modular}, \eqref{F1modular}, \eqref{eq:Fg-quasi-modular-expression}, \eqref{F23modular}, we obtain the corresponding expansions of the TR free energies. In particular,
    \begin{subequations}
    \begin{align}
    	&\begin{aligned} F_0=&\,\frac{s^2}{45}+\frac{4 i \nu  s}{5}+ \frac{\nu ^2}{2} \ln \left(-\frac{i \nu }{48 s}\right)-\frac{3\nu^2}{4}-\frac{47 i \nu ^3}{48 s}-\frac{7717 \nu ^4}{4608 s^2}+\frac{2663129 i \nu ^5}{552960 s^3}
        \label{expansion_F0}
    	\\ &\,+\frac{31386901 \nu ^6}{1769472 s^4}-\frac{2680068281 i \nu ^7}{35389440 s^5}+O\lb s^{-6}\rb,	\end{aligned}\\
    	& F_1=-\frac1{60}\ln\lb\nu^5s\rb-\frac{17 i \nu }{96 s}-\frac{3677 \nu ^2}{4608 s^2}+\frac{501953 i \nu ^3}{110592 s^3}+\frac{2086597 \nu ^4}{73728 s^4}-\frac{39389839561 i \nu ^5}{212336640 s^5}+O\lb s^{-6}\rb,\\
    	& F_2=\, -\frac{1}{240 \nu ^2}-\frac{7}{480 s^2}+\frac{20047 i \nu }{40960 s^3}+\frac{14438609 \nu ^2}{1769472 s^4}-\frac{5637865499 i \nu ^3}{53084160 s^5}
    	+O\lb s^{-6}\rb,\\
    	& F_3=\,\frac{1}{1008 \nu ^4}+\frac{245}{2304 s^4}-\frac{9265987 i \nu }{1048576 s^5}+O\lb s^{-6}\rb.
        \label{expansion_F3}
    \end{align}
    \end{subequations}
    Using the asymptotic expansion of the Barnes G-function $G\lb 1+z\rb$ as $z\to\infty$ (with $|\arg z|<\pi$),
    \beq\label{BarnesGas}
    G\lb 1+z\rb\simeq\frac{z^2}{2}\ln z-\frac{3z^2}{4}+\frac{z\ln 2\pi}{2}-\frac{\ln z}{12}+
    \zeta'\lb -1\rb+\sum_{g\ge 2}\frac{B_{2g}}{4g\lb g-1\rb z^{2g-2}},
    \eeq
    we find full agreement between both sides of \eqref{TRPIvell} up to order $O\lb s^{-6},\hbar^6\rb$ in the double expansion.  
    
    Let us note the following:
   \begin{itemize}
   	\item Assuming the TR/\PIeq relation \eqref{TRPIshort}, the knowledge of free energies $F_0,\ldots, F_g$ completely determines the coefficients  $\mathcal E_k\lb\nu\rb$, and thereby also the coefficients $\mathcal U_k\lb\nu\rb$ of $c=1$  conformal block $\mathcal F\lb \varepsilon\,|\,\nu\rb$, up to $k=2g-1$. This order upgrades to $k=2g+1$ if in addition we also allow the use of much more easily computable free energies of the degenerate curve $\mathcal C^{\mathrm{deg}}$.
   	\item For the expansions in $\hbar$ and $s^{-1}$ in \eqref{TRPIvell} to be consistent, the free energy $F_g$ with $g\ge 2$ must have the large $s$ asymptotic expansion
   	\beq\label{conigap}
   	F_g=\frac{\kappa_g}{\nu^{2g-2}}+O\lb s^{-\lb 2g-2\rb}\rb,
   	\eeq
   	with
   	\beq\label{eq:kappa_g}
   	\kappa_g=\frac{B_{2g}}{4g\lb g-1\rb}.
   	\eeq
   	This behavior, namely the absence of  $s^{-1},\ldots, s^{-\lb 2g-3\rb}$ contributions in the asymptotic series, is known as the \textit{conifold gap property}. 
    We will give a proof of this property in Section \ref{sec:conifold-gap} and Appendix \ref{appendix:constant-term}.
   \end{itemize}

   In topological string theory, the behavior \eqref{conigap} is believed to be universal    up to overall model-dependent prefactors (see e.g. \cite{gv-conifold}). 
   The constant term therein is nothing but the free energy of the Gaussian matrix model.  Since the topological recursion is intimately related to matrix models, 
   we may expect that the TR free energy satisfies the 
   conifold gap behavior universally. 
   One of the main contributions of this paper is to provide (see next subsection) a rigorous proof of this property  for our spectral curve ${\mathcal C}$. 
   The conifold gap behavior also plays a fundamental role in solving the holomorphic anomaly equation discussed in Section~\ref{sec_HAE}.

     \subsection{Asymptotic series for correlators and free energies}
     \label{sec:conifold-gap}

     \subsubsection{Rescaling of the spectral curve} \label{subsec:rescale}
   
     Let us first look at the large $t$ behavior of the spectral curve \eqref{eq:spcurve}. 
     It follows from \eqref{g2E4}--\eqref{g3E6} and \eqref{nome_expansion} that 
     $g_2$ and $g_3$ diverge as $s \to \infty$.
     To absorb this, we perform a rescaling and introduce new variables $(X,Y)$ by  
\begin{equation}  \label{eq:scaling-data}
(x,y) =  (\gamma^2 X, \gamma^3 Y), \qquad \gamma=\sqrt{3a}(-t)^{1/4}, \qquad a^2=1/6.
\end{equation} 
Then, the spectral curve \eqref{eq:spcurve} becomes 
\begin{equation} \label{eq:rescaled-sp-curve}
{\mathcal C}^{\rm res} ~:~ Y^2 = 4X^3 - \frac{4}{3}X + U, 
\end{equation}
with 
\begin{equation} \label{eq:U}
U=\gamma^{-6}u=(3a)^{-3}(-t)^{-3/2}u.
\end{equation} 
The condition \eqref{eq:nu} is equivalent to 
\begin{equation} \label{eq:rescaled-A-period}
\oint_A Y dX = 2 \pi i \nu \Lambda, \qquad
\Lambda = \gamma^{-5} = (3a)^{-5/2}(-t)^{- \frac{5}{4}}.
\end{equation}
For comparison with \eqref{bigtau1}--\eqref{Zasseries}, 
we should take $a = - \sqrt{1/6}$, 
in which case the branch of the square root is chosen as 
$\sqrt{3a} = i (3/2)^{1/4}$ so that
\begin{equation} \label{eq:Lambda-t}
\Lambda = -\frac{4i}{3s}.
\end{equation} 
In the remainder of this section, we adopt this branch.
(The other choice $a = \sqrt{1/6}$ corresponds to a different asymptotic behavior 
of tau function, which we will consider in Appendix \ref{appendix:conjecture-resurgence}.)  

Using relation \eqref{eq:Lambda-t}, the large $s$ expansions derived 
in the previous section can be translated into small $\Lambda$ expansions. 
For example, \eqref{nome_expansion} implies 
\begin{equation} \label{nome_expansion_2}
q = \frac{\nu \Lambda}{64} + \frac{141 \nu^2 \Lambda^2}{2048} + \frac{5379 \nu^3 \Lambda^3}{16384}
+ \frac{6861827 \nu^4 \Lambda^4}{4194304} + \frac{4502665551 \nu^5 \Lambda^5}{536870912}
+ O(\Lambda^6).
\end{equation}
Below, we will introduce a method to derive the series expansion in $\Lambda$ 
of the quantity calculated by topological recursion, based on this equation.

Before that, let us see that the $A$-cycle is represented by the vanishing cycle 
on ${\mathcal C}^{\rm res}$ when $\Lambda \to 0$, as follows.
Let $\tilde{e}_i = \gamma^{-2} e_i$ ($i=1,2,3$) be the zeros of 
the right hand side of \eqref{eq:rescaled-sp-curve}, 
and recall that the $A$-cycle is represented by a closed cycle around 
$\tilde{e}_2$ and $\tilde{e}_3$.
For computation of asymptotic behavior of these roots, we can use
\begin{equation}
   \tilde{e}_1= - \frac{\theta_3^4 + \theta_4^4}{3 E_4^{1/2}},\quad
   \tilde{e}_2=\frac{\theta_2^4 + \theta_3^4}{3 E_4^{1/2}},\quad
   \tilde{e}_3= - \frac{\theta_2^4 - \theta_4^4}{3 E_4^{1/2}}, 
\end{equation}
where
\begin{equation}
    \theta_2=\sum_{k \in {\mathbb Z}}Q^{(k+1/2)^2}, \quad 
    \theta_3=\sum_{k \in {\mathbb Z}}Q^{k^2}, \quad 
    \theta_4=\sum_{k \in {\mathbb Z}}(-1)^k Q^{k^2}, \quad Q=q^{1/2}.
\end{equation}
Using the series expansion \eqref{nome_expansion_2} of $q$, we have
\begin{subequations}
\begin{align} \label{tilde-e1}
\tilde{e}_1 & 
= 
-\frac{2}{3} + \nu \Lambda + \frac{49 \nu^2 \Lambda^2}{32} 
+ \frac{4543 \nu^3 \Lambda^3}{1024} + \frac{1053525 \nu^4 \Lambda^4}{65536} 
+ \frac{274381151 \nu^5 \Lambda^5}{4194304}
+ O\lb \Lambda^{6} \rb, \\
\tilde{e}_2 & = \frac{1}{3} + (\nu \Lambda)^{1/2} - \frac{\nu \Lambda}{2} 
 + \frac{25}{64} (\nu \Lambda)^{3/2} -\frac{49 \nu^2 \Lambda^2}{64} 
 + \frac{8139}{8192} (\nu \Lambda)^{5/2} - \frac{4543 \nu^3 \Lambda^3}{2048} +  O\lb \Lambda^{7/2} \rb,
\label{tilde-e2}
\\
\tilde{e}_3 & = \frac{1}{3} - (\nu \Lambda)^{1/2} - \frac{\nu \Lambda}{2} 
 -  \frac{25}{64} (\nu \Lambda)^{3/2} -\frac{49 \nu^2 \Lambda^2}{64} 
 - \frac{8139}{8192} (\nu \Lambda)^{5/2} - \frac{4543 \nu^3 \Lambda^3}{2048} +  O\lb \Lambda^{7/2} \rb.
\label{tilde-e3}
\end{align}
\end{subequations}
These behaviors are consistent with the equalities 
\begin{equation}
   4(\tilde{e}_1 \tilde{e}_2+\tilde{e}_1 \tilde{e}_3+\tilde{e}_2 \tilde{e}_3) = -\frac{4}{3}, \qquad 
   -4\tilde{e}_1 \tilde{e}_2 \tilde{e}_3=U.
\end{equation}
Moreover, we can show that the curve \eqref{eq:rescaled-sp-curve} has 
a limit when $\Lambda \to 0$, and the limiting curve 
\begin{equation} \label{eq:singular-elliptic-curve}
{\mathcal C}^{\rm res}_{\rm deg} ~:~ Y^2 = 4(X-1/3)^2 (X+2/3)
\end{equation}
is a degenerate elliptic curve, which is equivalent to 
the degenerate curve ${\mathcal C}_{\rm deg}$ under the same rescaling \eqref{eq:scaling-data}.
The $A$-cycle tends to small closed circle around the double point $1/3$.

\subsubsection{Series expansion of correlators} \label{subsec:TR-expansion}

Let us now compute the series expansion of the correlator $W_{g,n}$ in the limit $\Lambda \to 0$. 
Below, we will use the rescaled variable $X$ to express the correlators, 
and regard $X$ (and its copies) as a $\Lambda$-independent variable(s) in analyzing the small $\Lambda$ expansion. 

Using \eqref{g2E4}--\eqref{g3E6} and \eqref{nome_expansion_2}, we have 
\begin{equation}
U = \frac{8}{27} - 4 \nu \Lambda + \frac{15 \nu^2 \Lambda^2}{8}  + \frac{705 \nu^3 \Lambda^3 }{256} 
+ \frac{115755 \nu^4 \Lambda^4 }{16384} + \frac{23968161 \nu^5 \Lambda^5 }{1048576}
+ O(\Lambda^{6}).
\end{equation}
This allows us to see that $Y$ has the following series expansion
\begin{equation} \label{eq:series-Y}
Y = \sqrt{4X^3-4X/3+U} = \sum_{k \ge 0} Y^{[k]}(X) \, (\nu \Lambda )^k
\end{equation}
when $\Lambda \to 0$, where $Y^{[k]}$ are meromorphic functions on 
the degenerate curve ${\mathcal C}^{\rm res}_{\rm deg}$ 
and are explicitly given as 
\begin{equation}
Y^{[0]} = 2(X-1/3)\sqrt{X+2/3}, \quad 
Y^{[1]} = - \frac{1}{(X-1/3) \sqrt{X+2/3}}, \quad
Y^{[2]} = \frac{135X^3-45X-62}{288(X-1/3)^3 (X+2/3)^{3/2}},\ \dots
\end{equation}
These are independent of $\nu$ and $\Lambda$, and can be determined at any order in principle. 
The relation \eqref{eq:rescaled-A-period} implies that the
$A$-cycle tends to negatively-oriented small closed circle around $1/3$ so that
\begin{equation} \label{eq:branch-sqrt}
\frac{1}{2 \pi i} \oint_A  Y^{[k]} \, dX = - \Res_{X=1/3} Y^{[k]} \, dX =
\begin{cases} 
1 & k = 1, \\
0 & k \ne 1.
\end{cases} 
\end{equation}
The series expansion \eqref{eq:series-Y} immediately implies that 
$W_{0,1} = y dx = \Lambda^{-1} Y dX$ admits a series expansion of the form
\begin{equation} \label{eq:ser-W01}
W_{0,1}(X)  = \Lambda^{-1} \sum_{k \ge 0}  W_{0,1}^{[k]}(X) \, (\nu\,\Lambda)^k.
\end{equation}
Here, by definition, $W_{0,1}^{[k]} = Y^{[k]} \, dX$.

Next, let us compute the expansion of $W_{0,2}$. 
Using the addition formula for the $\wp$-function, 
we can express the Bergman bidifferential in \eqref{W02ell} as follows: 
\begin{equation} \label{eq:Bergman-via-XY}
W_{0,2}(X_1, X_2) 
= \left( \frac{1}{4}\Bigl( \frac{Y_1 + Y_2}{X_1 - X_2} \Bigr)^2 
- X_1 - X_2 + \frac{\widetilde{\eta}_A}{\widetilde{\omega}_A} \right) 
\frac{dX_1 \, dX_2}{Y_1 \, Y_2}.
\end{equation}
Here, $(X_i, Y_i) = (X_i, Y(X_i))$ ($i=1, 2$) 
denotes the copies of rescaled variables, and 
\begin{equation}  \label{eq:rescaled-omega-eta-A}
\widetilde{\omega}_A = \oint_{A} \dfrac{dX}{Y}, 
\quad
\widetilde{\eta}_A = - \oint_{A} \dfrac{X dX}{Y}.
\end{equation}
Combining the series expansion \eqref{eq:series-Y} of $Y$ with
\begin{equation}
\frac{\widetilde{\eta}_A}{\widetilde{\omega}_A} 
= \gamma^{-2} \frac{\pi^2}{3 \omega_A^2} E_2
= - \frac{1}{3} + \frac{3 \nu \Lambda}{4} + \frac{45 \nu^2 \Lambda^2}{32} 
+ \frac{19035 \nu^3 \Lambda^3}{4096} + \frac{2430855 \nu^4 \Lambda^4}{131072} 
+ \frac{1366185177 \nu^5 \Lambda^5}{16777216} + O(\Lambda^6), 
\end{equation}
which can also be computed by \eqref{nome_expansion_2}, 
we can obtain a series expansion of $W_{0,2}$ of the form 
\begin{equation} \label{eq:expansion-W02}
W_{0,2}(X_1, X_2) = \sum_{k \ge 0} W_{0,2}^{[k]}(X_1, X_2) \, (\nu \, \Lambda)^k.
\end{equation}
Here, the coefficients $W_{0,2}^{[k]}(X_1, X_2)$ are meromorphic bidifferentials 
on ${\mathcal C}^{\rm res}_{\rm deg}$ and are independent of $\nu$ and $\Lambda$. 
The first few coefficients are given explicitly as 
\begin{subequations}
\begin{align}
W_{0,2}^{[0]}(X_1, X_2) & = \frac{4 + 3X_1 + 3X_2 + 6 \sqrt{X_1+2/3} \sqrt{X_2 + 2/3}}
{12(X_1 - X_2)^2 \sqrt{X_1 + 2/3} \sqrt{X_2 + 2/3}}  \, dX_1 dX_2 
\\[+.5em]
W_{0,2}^{[1]}(X_1, X_2) & = \frac{dX_1 dX_2}{432 \, (X_1-1/3)^2(X_2-1/3)^2 (X_1+2/3)^{3/2}(X_2+2/3)^{3/2}}
\notag  \\
& \hspace{-2.5em}
\times \Bigl( 
34 + 66 (X_1 + X_2) + 36 (X_1^2 + X_2^2) + 171 X_1 X_2 + 27 X_1 X_2(X_1+X_2) + 81 X_1^2 X_2^2 \Bigr).
\label{eq:W02-2}
\end{align}
\end{subequations}


For general $(g,n)$, we have 
 
\begin{prop} \label{prop:expansion-Wgn}
For each  $g \ge 0$ and $n \ge 1$, 
the correlator $W_{g,n}$ of the spectral curve ${\mathcal C}$
has a convergent series expansion at $\Lambda = 0$ of the following form:
\begin{equation}\label{eq:expansion-Wgn}
W_{g,n}(X_1, \dots, X_n) = 
\Lambda^{2g-2+n} \,  \sum_{k \ge 0} 
W_{g,n}^{[k]}(X_1, \dots, X_n) \, (\nu \, \Lambda)^{k}.
\end{equation}
The coefficients $W_{g,n}^{[k]}$ are meromorphic multi-differentials 
on the degenerate curve ${\mathcal C}^{\rm res}_{\rm deg}$. 
\end{prop}

\begin{proof}
Since we have already seen that $W_{0,1}$ and $W_{0,2}$ admit a series expansion, 
we will now show by induction that $W_{g,n}$ with $2g-2+n \ge 1$ also admit a series expansion. 

Let ${W}^{\rm res}_{g,n}$ be the correlator of the type $(g,n)$ defined by the topological recursion 
from the rescaled spectral curve ${\mathcal C}^{\rm res}$ defined in \eqref{eq:rescaled-sp-curve}. 
It follows from the relation $W_{0,1} = \Lambda^{-1} {W}^{\rm res}_{0,1}$ 
and the recursion relation \eqref{eq:top-rec} that 
\begin{equation}
W_{g,n}(X_1, \dots, X_n)  = \Lambda^{2g-2+n} {W}^{\rm res}_{g,n}(X_1, \dots, X_n)  
\label{eq:scaling-degree-Wgn}
\end{equation}
hold for all $g$ and $n$. 
Therefore, our task is to show the existence of the power series expansion of ${W}^{\rm res}_{g,n}$
when $\Lambda \to 0$. 
However, we should note that a collision of ramification points (i.e., $\tilde{e}_2, \tilde{e}_3 \to 1/3$) 
happens when $\Lambda \to 0$.
Since the integrand of the defining formula \eqref{eq:top-rec} for the correlator
has singularities at those points,  the original definition \eqref{eq:top-rec} 
is not useful to derive a series expansion. 

To overcome this difficulty, we use the following alternative expression of \eqref{eq:top-rec} 
which can be derived using the residue theorem in a manner analogous to \cite[Lemma B.1]{Iwaki19}:
\begin{align} 
W^{\rm res}_{g,n}(X_1, \dots, X_n) 
& = 
\frac{dX_1}{\widetilde{\omega}_A \, Y(X_1)} \,  
\oint_{X \in A} \frac{{R}^{\rm res}_{g,n}(X, X_2, \dots, X_n)}
{2 W_{0,1}^{\rm res}(X)} 
- 2 
\sum_{i=1}^{n}  \Res_{X=X_i} \frac{\int^{X'=X}_{X'=\infty} {W}^{\rm res}_{0,2}(X_1, X')}
{2 W_{0,1}^{\rm res}(X)} 
\, {R}^{\rm res}_{g,n}(X, X_2, \dots, X_n).
\label{eq:alternative-TR-scaling}
\end{align}
Here, ${R}^{\rm res}_{g,n}$ is obtained from $R_{g,n}$ given in \eqref{eq:R-gn} 
by replacing $W_{g,n}$ by ${W}^{\rm res}_{g,n}$, and 
we have chosen the base point in \eqref{eq:rec-ker} as $\tilde{z} = 0$.
Under the induction hypothesis, ${R}^{\rm res}_{g,n}$ also has a series expansion of the form 
\begin{equation}
{R}^{\rm res}_{g,n}(X, X_2, \dots, X_n)
= \sum_{k \ge 0} R_{g,n}^{[k]}(X, X_2, \dots, X_n)(\nu \Lambda)^{k}.
\end{equation}
The coefficients $R_{g,n}^{[k]}(X, X_2, \dots, X_n)$ are 
meromorphic multi-differentials on ${\mathcal C}^{\rm res}_{\rm deg}$. 
The aforementioned collision of ramification points 
does not cause the integration cycles in both terms of \eqref{eq:alternative-TR-scaling} to be pinched.
Therefore, the integration over the $A$-cycle and residue calculation at $X = X_i$ can be performed term-wise. 
In particular, the expansion coefficients of the $A$-cycle integral of the first term 
can be obtained term by term by calculating the residue at $X = 1/3$, similarly to \eqref{eq:branch-sqrt}.
Therefore, we can inductively verify that $W_{g,n}$ have a convergent series expansion 
of the form \eqref{eq:expansion-Wgn} when $\Lambda$ tends to $0$. 
This completes the proof of Proposition \ref{prop:expansion-Wgn}. 
\end{proof}

From the proof, it can also be seen that the coefficient $W_{g,n}^{[k]}$ 
with $2g-2+n \ge 1$ has poles only at branch points, infinity, and at two preimages of $X_i = 1/3$. 
These coefficients are independent of $\nu$ and $\Lambda$ and are explicitly computable at any order. 
This approach also enables us to derive the series expansion without requiring the explicit form 
of the correlator such as \eqref{W11ell02}. For example, the above algorithm provides
\begin{equation}
W_{0,3}^{[0]}(X_1, X_2, X_3) = 
\frac{dX_1 dX_2 dX_3}{32(X_1+2/3)^{3/2}(X_2+2/3)^{3/2}(X_3+2/3)^{3/2}}, \qquad
W_{1,1}^{[0]}(X_1) =  
\frac{(3X_1+5) dX_1}{192(X_1+2/3)^{5/2}}.
\end{equation}

\subsubsection{Series expansion of $F_g$ and conifold gap property}
\label{sec:series-exp-Fg}

We have already seen in \eqref{expansion_F0}--\eqref{expansion_F3} that 
the first few free energies $F_g$ admit a series expansion at $\Lambda = 0$. 
The main claim of this section, 
which justifies the previous properties \eqref{conigap}--\eqref{eq:kappa_g}, 
is the following. 

\begin{theo}\label{thm:coni-gap}
For each $g \ge 2$, the $g$-th free energy $F_g$ of the spectral curve ${\mathcal C}$ 
has a convergent series expansion at $\Lambda = 0$ of the following form:
\begin{equation} \label{eq:expansion-Fg}
F_g  = 
\frac{B_{2g}}{4g(g-1) \, \nu^{2g-2}} + \Lambda^{2g-2} \sum_{k \ge 0} F_g^{[k]} \, (\nu \, \Lambda)^k. 
\end{equation}
The coefficients $F_g^{[k]} \in {\mathbb C}$ are 
independent of $\nu$ and $\Lambda$ and are explicitly computable at any order. 
\end{theo}

\begin{proof}
Since $\xi$, given in \eqref{eq:Yukawa}, has an expansion 
\begin{equation}
\xi = \frac{1}{\nu} \left[1  + \frac{141 \nu \Lambda}{32} + \frac{23151 \nu^2 \Lambda^2}{1024} 
 + \frac{7989387 \nu^3 \Lambda^3}{65536} + \frac{1412410545 \nu^4 \Lambda^4}{2097152} 
 + \frac{506532905109 \nu^5 \Lambda^5}{134217728} +O(\Lambda^6)
 \right], 
\end{equation} 
the expression \eqref{eq:Fg-quasi-modular-expression} implies that 
$\nu^{2g-2} F_g$ with $g \ge 2$ has a convergent series expansion with respect to $(\nu \Lambda)$. 
The non-trivial aspects of the statement are as follows:
\begin{itemize}
\item 
The gap property \eqref{conigap}; i.e., the absence of the coefficients of 
$\Lambda^{1}, \dots, \Lambda^{2g-1}$ in \eqref{eq:expansion-Fg}. 

\item 
The expression \eqref{eq:kappa_g} of leading term. 
\end{itemize}

The gap property can be shown as follows. Here we employ the variational formula
\begin{equation} \label{eq:vari-Fg}
\frac{\partial F_g}{\partial t} 
= \Res_{z=0} \frac{W_{g,1}(z)}{z} 
= 2 \Res_{X=\infty} \left( \Lambda^{2g-1 - \frac{1}{5}} \sum_{k \ge 0} 
X^{\frac{1}{2}}  W_{g,1}^{[k]}(X) \, (\nu \Lambda)^{k}  \right),
\end{equation}
which is valid for $g \ge 1$ (see \cite[Theorem 5.1]{EO07} and \cite[Proposition 3.4]{Iwaki19}). 
Since the residue cycle around $X = \infty$ is not pinched by 
the merging pair $\tilde{e}_2, \tilde{e}_3$ of ramification points, 
we can take the residue in the right hand side of \eqref{eq:vari-Fg} term-wise. 
Thus we can verify that $\partial F_g/ \partial t$ has a series expansion at $\Lambda = 0$, 
and the absence of the coefficients of 
$\Lambda^{1}, \dots, \Lambda^{2g-1}$ in \eqref{eq:expansion-Fg}. 
Moreover, the variational formula \eqref{eq:vari-Fg} fixes all series coefficients of $F_g$ by
\begin{equation} \label{eq:Fgk-formula}
\frac{15(2g-2+k)}{8} \, F_g^{[k]} = 
2 \Res_{X=\infty} X^{\frac{1}{2}} W_{g,1}^{[k]}(X) \qquad (k \ge 0).
\end{equation} 
except for a $t$-independent leading term. 

Next, let us prove the formula \eqref{eq:kappa_g} for the leading term of $F_g$. 
The proof is technical and would be lengthy if presented in full, 
so here we will outline several key equalities, 
with the detailed derivation provided in Appendix \ref{appendix:constant-term}.
To prove \eqref{eq:kappa_g}, we apply the following symplectic transformation 
$(x,y) \mapsto (\widetilde{X}, \widetilde{Y})$ 
to the (unrescaled) spectral curve ${\mathcal C}$: 
\begin{equation} \label{eq:symplectic-tr-to-Weber}
(x,y) = \left( \frac{{e}_2+ {e}_3}{2} + \alpha \widetilde{X}, \alpha^{-1} \widetilde{Y} \right), 
\qquad 
\alpha = \left(-\frac{1}{24 {e}_1}\right)^{1/4},
\end{equation}
where $e_1, e_2, e_3$ are the zeros of the right hand side of \eqref{eq:spcurve}. 
It transforms the spectral curve ${\mathcal C}$ into 
\begin{equation} \label{eq:spcurve-Weber-expression}
\widetilde{\mathcal C} ~:~\widetilde{Y}^2 = (1 + 16 \alpha^5  \widetilde{X} ) 
\left( \frac{\widetilde{X}^2}{4} - \frac{({e}_2 - {e}_3)^2}{16 \alpha^2}\right).
\end{equation}
It is known (and easily verified in this case) that 
the symplectic transformation \eqref{eq:symplectic-tr-to-Weber} 
preserves the correlator $W_{g,n}$ and free energy $F_{g}$ (cf. \cite[\S 7]{EO07}). 
The key observation is the following: 
\begin{equation}
\lim_{\Lambda \to 0}(1 + 16 \alpha^5  \widetilde{X} ) 
\left( \frac{\widetilde{X}^2}{4} - \frac{(e_2-e_3)^2}{16 \alpha^2}\right)
= \frac{\widetilde{X}^2}{4} - \nu.
\end{equation}
This follows immediately from the behavior \eqref{tilde-e1}--\eqref{tilde-e3} of $e_i$. 
Therefore, the limit of the spectral curve $\widetilde{\mathcal C}$ becomes
\begin{equation} \label{eq:Weber-curve}
{\mathcal C}_{\rm Web} ~:~ \widetilde{Y}^2 = \frac{\widetilde{X}^2}{4} - \nu, 
\end{equation}
which is called the Weber curve. 
As is well-known, the free energy $F_{g}^{\rm Web} $ of 
the Weber curve ${\mathcal C}_{\rm Web}$ is given by 
\begin{equation} \label{eq:Weber-free-energy-explicit}
F_{g}^{\rm Web} = \frac{B_{2g}}{4g(g-1) \, \nu^{2g-2}} \qquad (g \ge 2).
\end{equation}
(See \cite[Theorem 1]{Nor09}, \cite[Theorem 4.9]{IKoT1} for the proof). 
Therefore, the desired property \eqref{eq:kappa_g} is established by proving that 
the topological recursion and limit of the spectral curve commute, that is, 
\begin{equation} \label{eq:limit-Fg-appendix}
\lim_{\Lambda \to 0} F_g = F_{g}^{\rm Web}
\quad (g \ge 2).
\end{equation}
The details of this final step are provided in Appendix \ref{appendix:constant-term}. 
\end{proof}

\begin{rmk}
In cases where a family of spectral curves depending on a parameter arises, as discussed in this paper, 
the commutativity between the limit with respect to the parameter and the application of 
topological recursion often becomes a significant issue.
Indeed, in \cite{BouSul12}, an example is provided where the limit of the spectral curve and the topological recursion do not commute.
The commutativity with the limit has been discussed in a recent paper \cite{BBCKS23}, 
where the "constant genus assumption" is imposed, 
meaning that the genus of the spectral curve remains fixed during the limiting process. 
However, the example treated in this paper provides an example where the elliptic curve 
reduces to the Weber curve with genus $0$, yet the topological recursion commutes with the limit.
\end{rmk}

\begin{rmk}
Slightly digressing from the main discussion, let us make a brief yet meaningful observation here.
As stated in the Proposition \ref{prop:expansion-Wgn}, 
the coefficients $W_{g,n}^{[k]}$ define multi-differentials on the degenerate curve ${\mathcal C}_{\rm deg}^{\rm res}$. 
In fact, the constant term $W_{g,n}^{[0]}$ coincides with the correlator 
of the degenerate curve ${\mathcal C}_{\rm deg}$ discussed in Section \ref{subsec_warmup} 
(through the rescaling in equation \eqref{eq:scaling-data}). 
More specifically, if we denote by $W_{g,n}^{\rm deg}$ the correlator
of ${\mathcal C}_{\rm deg}$, then the following holds:
\begin{equation} \label{eq:Wgn-0-and-deg}
W_{g,n}^{\rm deg}(z_1, \dots, z_n) = \Lambda^{2g-2+n} W_{g,n}^{[0]}(X_1, \dots, X_n)
\end{equation}
under the change of coordinate $X_i = (z_i/ \gamma)^2 - 2/3$ for $i = 1, \dots, n$.
This relation can be easily proved by comparing the first terms of the 
$\Lambda$-expansion on both sides of \eqref{eq:alternative-TR-scaling}. 
(It can be similarly shown that $W_{g,n}^{\rm deg}$ satisfies the same relation; see \cite[Theorem 3.11]{IS16}.) 
Combining the variational formula \eqref{eq:vari-Fg} with \eqref{eq:Wgn-0-and-deg}, 
we can also prove that the the free energy $F_{g}^{\rm deg}$ 
of the degenerate curve ${\mathcal C}_{\rm deg}$ appears in the expansion 
\eqref{eq:expansion-Fg} of $F_g$ as 
\begin{equation} \label{eq:Fgdeg-as-subleading}
F_{g}^{\rm deg} = \Lambda^{2g-2} F_{g}^{[0]} \quad (g \ge 2).
\end{equation}
\end{rmk}

     \section{Holomorphic anomaly equation\label{sec_HAE}}
     
     This section is devoted to the so-called holomorphic anomaly equation (HAE), 
     given in \eqref{eq:HAE} below. 
     We review a method of recursive computation of the TR free energies, based on the HAE and 
     the conifold gap property shown 
     in Theorem \ref{thm:coni-gap}.
     The method itself is a well-known technique that originated in \cite{BCOV93v2} and has been widely employed in the context of topological string theory (higher genus B-model), supersymmetric gauge theory and quantum mechanics \cite{HK06,HK09,KW10,HK10,HKPK11,CM16,FMP23,FGMS23}. We are concerned here with a proof of existence and uniqueness of the HAE solutions. As explained below, with the rigorous proof of the conifold gap property 
     the existence is established as an application. 
     Besides, we introduce special bases for holomorphic ambiguities and use them to prove uniqueness of the HAE solution.   
     
     \subsection{General setup}
 
     In what follows, we will mean by $\partial F_g / \partial E_2$  the derivative of 
     the expression \eqref{eq:Fg-quasi-modular-expression} for $F_g$ obtained by the differentiation of the polynomial $P_g\lb E_2,E_4,E_6\rb$ with respect to $E_2$, regarding the prefactor of $P_g$ as a constant.
     The following theorem relates $\partial F_g / \partial E_2$ to lower genus free energies.

     \begin{theo}[{\cite[Theorem 6.1]{EO07}}] \label{thm:TR-HAE}
     	For $g \ge 2$, the $g$-th free energy of the \PIeq spectral curve \eqref{eq:sp-PI} satisfies the
     	holomorphic anomaly equation 
     	\begin{equation} \label{eq:HAE}
     		\frac{\partial F_g}{\partial E_2} =  - \frac{1}{24} 
     		\left[ \frac{\partial^2 F_{g-1}}{\partial \nu^2} 
     		+ \sum_{h=1}^{g-1} \frac{\partial F_{h}}{\partial \nu} 
     		\, \frac{\partial F_{g-h}}{\partial \nu} \right]. 
     	\end{equation}
     \end{theo}
     
     \pf
     	Our Bergman bidifferential \eqref{W02ell} is written as 
     	\begin{equation}
     		W_{0,2}\lb z_1, z_2\rb = \left( \omega_A^2 \wp(z_1 - z_2) + 2 \pi i \kappa \right) \, 
     		\frac{dz_1}{\omega_A} \, \frac{dz_2}{\omega_A}, \qquad
     		\kappa = \frac{\eta_A \omega_A}{2\pi i}.
     	\end{equation}
     	As shown in \cite[Theorem 6.1]{EO07}, the derivatives of $W_{g,n}$ and $F_g$ 
     	with respect to $\kappa$ can be written again by certain period integrals 
     	of $W_{g,n}$ with different $(g,n)$.
     	The explicit formula for $\partial F_g / \partial \kappa$ reads 
     	\begin{align}
     		2\pi i \frac{\partial F_g}{\partial \kappa} & = 
     		\frac{1}{2} \oint_{z_1 \in B} \oint_{z_2 \in B} W_{g-1,2}(z_1,z_2) 
     		+ \frac{1}{2} \sum_{h=1}^{g-1} \left( \oint_{z_1 \in B} W_{h, 1}(z_1)  \right)
     		\, \left( \oint_{z_2 \in B} W_{g-h, 1}(z_2)  \right) \notag  \\
     		& = \frac{1}{2} \left[ \frac{\partial^2 F_{g-1}}{\partial \nu^2}
     		+ \sum_{h=1}^{g-1} \frac{\partial F_{h}}{\partial \nu} 
     		\, \frac{\partial F_{g-h}}{\partial \nu} \right],
     		\qquad g \ge 2. 
     		\label{eq:EO-HAE}
     	\end{align}  
     	Here we have used the variational formula $\ds\frac{\partial F_g}{\partial\nu}=\ds\oint_B W_{g,1}$, where the derivative with respect to $\nu$ is computed at fixed $t$, see 
     	\cite[Theorem 5.1]{EO07}. The result \eqref{eq:HAE} then immediately follows from the relation \eqref{XE2} rewritten as $2 \pi i \kappa = {\pi^2 E_2} / {3}$.
     \epf
     
     The HAE \eqref{eq:HAE} can be used for very effective recursive computation of $F_g$ starting from $F_1$ given by 
     \eqref{eq:expression-F0-F1}. Just as in the previous section, the free energies $F_g$ can be considered as functions of variables $\lb \nu,t\rb$ or,
     equivalently, of variables $\lb\omega_A,q\rb$. We will rewrite 
     the $\nu$-derivatives for fixed $t$ on the right of \eqref{eq:HAE} in terms of the derivative
     $D_\tau=q\frac{\partial}{\partial q}$, 
     \begin{equation}\label{eq:nuqchain}
     	\frac{\partial}{\partial \nu} f\lb q\rb =\xi\lb q\rb \cdot D_\tau f(q),
     	\qquad \xi(q)=\frac{\partial}{\partial \nu}\ln q .
     \end{equation}
     Let us show that $\xi\lb q\rb$ is given by \eqref{eq:Yukawa}. We will need the differentiation formulas for the Eisenstein series with respect to the nome $q$,
     \begin{equation}\label{eq:diffE}
     	D_\tau E_2 = \frac{1}{12}\left(E_2^2 -E_4\right),\qquad  
     	D_\tau E_4 = \frac{1}{3}\left(E_2 E_4 -E_6\right), \qquad 
     	D_\tau E_6 = \frac{1}{2}\left(E_2 E_6 -E_4^2\right),
     \end{equation}
     The relation between $\nu$ and $q$ can be deduced from two ways to compute the ratio $g_3^2/g_2^3$:
     \begin{equation}
     	\frac{u^2}{\lb -2t\rb^3}=\frac{g_3^2}{g_2^3}=\frac{E_6^2}{27E_4^3}.
     \end{equation}
     The derivative of the left hand side  with respect to $\nu$ includes 
     \begin{equation}
     	\frac{\partial u}{\partial \nu}=\frac{4\pi i}{\omega_A},
     \end{equation}
     coming from the differentiation of \eqref{eq:nu}. The derivative of the right hand side
     can be computed using the chain rule \eqref{eq:nuqchain} together with \eqref{eq:diffE}. 
     Equating the two results leads to the expression \eqref{eq:Yukawa} for the quantity $\xi\lb q\rb$, known as the Yukawa coupling in topological string theories. Note that   $\xi\lb q\rb$ actually depends not only on the nome $q$ but also on $\omega_A$.
     	
     	Additionally, we will need the $\nu$-derivative of $\omega_A$ for fixed $t$. It follows from the $\nu$-derivative of 
     	$g_2=-2t$ given by \eqref{g2E4} that
     	\begin{equation}\label{eq:dernu-omega}
     		\frac{\partial}{\partial \nu} \ln \omega_A = \frac14 \frac{\partial}{\partial \nu} \ln E_4(q).
     	\end{equation}
     	Equations \eqref{eq:dernu-omega} and \eqref{eq:nuqchain} allow to rewrite 
     	the right hand side of \eqref{eq:HAE} in terms of derivatives in the nome $q$ only.
     	
     	The relation \eqref{eq:HAE} thus determines the coefficients $C_{a,b,k}$ with $k>0$
     	in the expression \eqref{eq:Fg-quasi-modular-expression} for $F_g$.
     	The remaining coefficients $C_{a,b}:=C_{a,b,0}$
     	appear in the $E_2$-independent part of $F_g$, the so-called holomorphic ambiguities
     	\begin{equation} \label{eq:hol-ambiguity}
     		H_g=\left( \frac{\xi}{E_4}  \right)^{2g-2} h_g,
     		\qquad
     		h_g= \sum_{\substack{a,b \ge 0 \\ 4a+6b=14 \lb g-1\rb}} C_{a,b} \, E_4^a E_6^b \in \mathcal{M}_{14\lb g-1\rb},
     	\end{equation}
     	where $\mathcal{M}_k$ is the space of $\mathrm{SL}_2\lb \mathbb{Z}\rb$ modular forms of weight $k$. 
     	In the next subsection, we will discuss how to fix them using the gap condition.
   
       \subsection{Fixing holomorphic ambiguity}
       \label{subsec:Fixing-holomorphic-ambiguity}
   Here we show that the conifold gap properties 
       \eqref{eq:expansion-Fg}
       of free energies are sufficient to fix the holomorphic ambiguities \eqref{eq:hol-ambiguity},  
       and give a recursive algorithm for the computation of $F_g$. 
       The crucial point here is that after the initial constant term in the asymptotic expansion of $F_g$, 
       the coefficients of all powers from 
       $\Lambda^{1}$ up to $\Lambda^{2g-3}$  
       vanish (indicating the presence of a gap). 
       This gives a set of constraints on the coefficients $C_{a,b}$
       in \eqref{eq:hol-ambiguity}. 
       
      The gap condition in terms of $\Lambda$-expansion can be reformulated in terms of $q$-expansion 
      with the help of relation \eqref{nome_expansion_2} and its inverse.  
      Indeed, the gap property of $F_g$  described by Theorem~\ref{thm:coni-gap}, 
      is equivalent to vanishing of the coefficients of $q^k$ with $k=1,\ldots,2g-3$, in the $q$-expansion of $\nu^{2g-2} F_g$: 
       	\begin{equation} \label{eq:q-expansion-Fg}
       		\nu^{2g-2} F_g  = \frac{B_{2g}}{4g(g-1)} + q^{2g-2} \sum_{k \ge 0} \widetilde F_g^{[k]} \, q^k . 
       	\end{equation}
       	Here the coefficients $\widetilde F_g^{[k]}\in\mathbb C$ are obtained by re-expanding 
	    \eqref{eq:expansion-Fg} as a $q$-series by using the inverse of the relation \eqref{nome_expansion_2}. 
      The fact that $\nu^{2g-2}F_g$ depends only on $q$   follows 
       from \eqref{eq:Fg-quasi-modular-expression} and the relation 
       \begin{equation}\label{eq:nuxi}
       	\frac{\nu \xi}{E_4}= \frac{E_2E_4-E_6}{720\eta^{24}}=1+42q+840 q^2+11340 q^3+\ldots,
       \end{equation}
       which, in turn, is a consequence of \eqref{numodular} and \eqref{eq:Yukawa}.

       We will use the conifold gap condition described by \eqref{eq:q-expansion-Fg}
       to fix the holomorphic ambiguity $h_g$ given by \eqref{eq:hol-ambiguity} which is 
       a modular form of weight $14(g-1)$. 
       By this  we refer to fixing $h_g$ through the requirement on $\nu^{2g-2}F_g$ to have the coefficients 
       of $q^j$, $j=0,1,\ldots,2g-3$ in accordance with the right hand side of \eqref{eq:q-expansion-Fg}. 
       To implement this, the key idea is to choose a convenient special basis in the corresponding space of modular forms $\mathcal{M}_k$:
       \begin{equation}\label{nice_basis}
       	V_j= E_4^{a_0-3j}E_6^{b_0} \Delta^j\lb q\rb, \qquad j=0,\ldots, \dim \mathcal{M}_k-1,  
       \end{equation}
       where $\Delta\lb q\rb=\eta^{24}\lb q\rb$ is a cusp form of weight 12, and
       \begin{equation}
       	b_0=\begin{cases}
       		1, \quad k=2 \mod 4,\\
       		0, \quad k=0 \mod 4,
       	\end{cases}
       	\qquad a_0=\frac{k-6b_0}{4},\qquad 
       	\dim \mathcal{M}_k =1+[a_0/3].
       \end{equation}
       The cusp form $\Delta\lb q\rb$ is related to, but should not be confused with, the discriminant $\Delta$ of the Weierstrass curve, cf. \eqref{eq:delta-eta}.
      The main reason for choosing the basis \eqref{nice_basis} is its expansion in powers of $q$: 
       \begin{equation}\label{eq:Vj-expansion}
       	V_j=q^j\lb 1+O\lb q\rb\rb.
       \end{equation}
       In this basis, the holomorphic ambiguity takes the form
       \begin{equation}
       	h_g=\sum_{j=0}^{d_g-1} \alpha_j V_j,
       	\qquad
       	d_g=\dim \mathcal{M}_{14\lb g-1\rb}.
       \end{equation}
       
       The procedure of using HAE \eqref{eq:HAE}  to find all $F_g$ is recursive starting from $F_2$ and using $F_1$. 
       At each recursion step, $F_g$ can be presented in the form  \eqref{eq:Fg-quasi-modular-expression} with 
       unknown $h_g=P_g\lb 0,E_4,E_6\rb$ which is the $E_2$-free part of the polynomial $P_g$. 
       The conifold gap condition gives a system of  $2g-2$ linear constraints on the coefficients 
       $\alpha_j$, ${j=0,\ldots,d_g-1}$. Since the dimension $d_g$ is given by
       \beq
       d_g=\begin{cases}
       	\left[\frac{7g-4}{6}\right],\qquad & g\in2\mathbb N,\vspace{0.1cm}\\
       	\left[\frac{7g-1}{6}\right],\qquad & g\in 2\mathbb N+1,
       \end{cases}
       \eeq
        for any $g\ge 2$ we have $d_g<2g-2$ and the system of constraints is overdetermined. 
       Indeed, already the first $d_g$ constraints corresponding to the first $d_g$ powers of $q$ uniquely fix all $\alpha_j$  due to the unitriangular form of the matrix of coefficients following from \eqref{eq:Vj-expansion}. All the other constraints must be satisfied automatically to ensure the existence of solutions of HAE supplemented with the conifold gap condition. The existence should follow from the properties of modular forms entering the procedure, though we do not have a  direct proof of this fact.
       However, we do know already that the TR free energies $F_g$ satisfy
       both  the conifold gap property (cf. Theorem \ref{thm:coni-gap}) 
       and the holomorphic anomaly equation (cf. Theorem \ref{thm:TR-HAE}), 
       therefore a solution of the overdetermined system of constraints necessarily exists! 
       
       Summarizing, we have the following
       \begin{theo} \label{thm:sol-HAE-weak}
       	For any $g \ge 2$, the holomorphic anomaly equations \eqref{eq:HAE} 
       	supplemented with the  \underline{weak} conifold gap constraint 
       	\begin{equation} \label{eq:q-expansion-Fg-weak}
       		\nu^{2g-2} F_g  = \kappa_g + O\lb q^{d_g}\rb,  
       	\end{equation}
       	where $d_g=\dim \mathcal{M}_{14\lb g-1\rb}$ and $\kappa_g$ is given by \eqref{eq:kappa_g}, 
       	define $F_g$ uniquely by recursion in $g$ starting from $F_1$ given by \eqref{eq:expression-F0-F1}. 
       \end{theo}
       
       \begin{cor} \label{cor:existence-strong-coniforld-gap-HAE}
       	For each $g \ge 2$, the free energies $F_g$ defined by Theorem~\ref{thm:sol-HAE-weak} 
       	satisfy the conifold gap property \eqref{eq:q-expansion-Fg} with more vanishing coefficients in the $q$-expansion of
       	$\nu^{2g-2} F_g$ than is required by \eqref{eq:q-expansion-Fg-weak}. 
       \end{cor}

       	In light of  many works on topological string theory such as \cite{HK06,HK09}, 
       	the above claim is believed to hold true generally. 
       	However, to our knowledge, its mathematically rigorous proof 
       	is not known (at least for the TR free energy of our spectral curve \eqref{eq:sp-PI}). 
       	Below we illustrate the method
       	by the computation of $F_2$ and $F_3$.
       	From these examples alone, it becomes apparent that 
       	the existence of solutions for the overdetermined system,
       	stated in Theorem \ref{thm:sol-HAE-weak}  and Corollary \ref{cor:existence-strong-coniforld-gap-HAE}, 
       	is highly nontrivial. Note that analogous computations 
       	have already appeared in e.g. 
       	 \cite[Section 4.1]{FGMS23} and \cite[Section 3]{PP23} (in the unrefined limit). Reproducing them here serves mainly pedagogical purposes.

       \begin{eg}[Computation of $F_2$]
       	We have already obtained an explicit formula for $F_2$ in \eqref{eq:Fg-quasi-modular-expression}, \eqref{eq:Yukawa}, \eqref{F23modularA} by an \textit{ab initio} TR calculation. Let us now rederive it using the aforementioned method. 
       	     The holomorphic anomaly equation \eqref{eq:HAE} for $F_2$ is 
       \begin{equation}\label{eq:HAEg2}
       	\frac{\partial F_2}{\partial E_2} 
       	= {- } \frac{1}{24} 
       	\left[ \frac{\partial^2 F_1}{\partial \nu^2} + \lb \frac{\partial F_1}{\partial \nu} \rb^2 \right]
       	=  \left( \frac{\xi}{E_4} \right)^{2} \frac{\partial P_2}{\partial E_2} .
       \end{equation}
       Our task is to find the polynomial $P_2=P_2\lb E_2,E_4,E_6\rb$. 
       The right hand side of HAE gives its $E_2$-derivative 
       \begin{equation}\label{eq:derP2}
       	\frac{\partial P_2}{\partial E_2} = - \frac{5E_2^2E_4^2 + 22 E_2 E_4 E_6 + 14 E_4^3 + 11 E_6^2}{13824}  .
       \end{equation}
       To derive this expression, we used that
       \begin{equation}
       	\frac{\partial  F_1 }{\partial \nu }= -\frac{\xi}{E_4} \cdot\frac{E_2E_4+E_6}{24},
       \end{equation}
       for $F_1$ given by  \eqref{eq:expression-F0-F1}, and the $\nu$-derivatives \eqref{eq:nuqchain}, \eqref{eq:dernu-omega}  allowing to compute 
       \begin{equation}
       	\frac{\partial}{\partial \nu } \lb\frac{\xi}{E_4}\rb=-\left(\frac{\xi}{E_4}\right)^2 \frac{7E_2E_4 +5E_6}{12}.
       \end{equation}
       Integrating \eqref{eq:derP2} with respect to $E_2$ leads to 
       \begin{equation}\label{eq:F2modular}
       	F_2 = \left(\frac{\xi}{E_4}\right)^2 
       	\left( - \frac{5E_2^3E_4^2 + 33 E_2^2 E_4 E_6 + E_2 \lb 42 E_4^3 + 33 E_6^2\rb}{41472} + \alpha_0 E_4^2 E_6 \right).
       \end{equation}
       This equation contains an unknown constant $\alpha_0$ multiplying $E_4^2E_6$, the only possible $\mathrm{SL}_2\lb \mathbb Z\rb$ modular form of weight~$14$. 
       To fix  $\alpha_0$, let us look at the series expansion of $\nu^2 F_2$ at $q = 0$ :
       \begin{equation}
       	\nu^2 F_2 = 
       	\alpha_0-\frac{113 }{41472 } 
       	+ \left(60 \alpha_0+\frac{299}{3456}  \right)  q +  O(q^2),
       \end{equation}
       which can be obtained from \eqref{eq:nuxi}. 
       The standard (strong) conifold gap property \eqref{eq:q-expansion-Fg} then imposes two constraints,
       \begin{align}\label{eqsalpha0}
       	\begin{cases} \displaystyle
       		\alpha_0-\frac{113 }{41472 }   & = 
       		\dfrac{B_{2g}}{4g\lb g-1\rb } \Bigl|_{g=2} = 
       		\displaystyle
       		-\frac{1}{240} , 
       		\\[+1.em]
       		\displaystyle
       		60 \alpha_0+\frac{299}{3456}  & =   0,
       	\end{cases}
       \end{align}
       whereas the weak conifold gap condition \eqref{eq:q-expansion-Fg-weak} would only imply the first of them.
       Both equations in \eqref{eqsalpha0} are of course equivalent, and their solution reads
       \begin{equation}
       	\alpha_0 = - \frac{299}{207360}. 
       \end{equation}
       This fixes the  holomorphic ambiguity in the expression \eqref{eq:F2modular} for $F_2$, 
       which agrees with \eqref{F23modularA}. 
       The $q$-expansion of $\nu^2F_2$ becomes 
       \begin{equation} \label{eq:F2qexp}
       	\nu^2F_2 =
       	-\frac{1}{240} +\frac{168 q^2}{5} +\frac{70353 q^3}{2} +\frac{25505004 q^4}{5}+ \frac{1790001618 q^5}{5}
       	+O\lb q^6\rb.
       \end{equation}
       
       \end{eg}

       \begin{eg}[Computation of $F_3$]
       
       Let us find the next contribution to free energy, $F_3$, following the same strategy. 
       The holomorphic anomaly equation
       \begin{equation}
       	\frac{\partial F_3}{\partial E_2} = - \frac{1}{24} \left[ \frac{\partial^2 F_2}{\partial \nu^2} 
       	+ 2 \frac{\partial F_2}{\partial \nu} \frac{\partial F_1}{\partial \nu} \right]
       \end{equation}
       fixes $p_3$ in the decomposition 
       \begin{equation}\label{eq:F3modular}
       	F_3=\left(\frac\xi{E_4}\right)^4 \lb p_3+h_3\rb,
       \end{equation}
       where $p_3$ combines all terms with $k\ge 1$ in the sum \eqref{eq:quasi-modular-part}, and $h_3$ is the holomorphic ambiguity: 
       \begin{equation}\label{eq:h3modular}
       	h_3=\alpha_0 V_0+\alpha_1 V_1+\alpha_2 V_2= \alpha_0 E_4^7 +\alpha_1 E_4^4\Delta\lb q\rb +\alpha_2 E_4\Delta^2\lb q\rb.
       \end{equation}
       The explicit formula for $p_3$ reads 
       \begin{align}
       	716636160 p_3= &\, 
       	75E_2^6 E_4^4 
       	+ 1200 E_2^5 E_4^3 E_6 
       	+ 45E_2^4 E_4^2 \lb 56E_4^3+159E_6^2\rb 
       	+ 10E_2^3 E_4 E_6  \lb 3121 E_4^3 + 1815 E_6^2\rb \\
       \nonumber	&\,+ 3 E_2^2 \lb 8023 E_4^6 + 39964 E_4^3 E_6^2 + 4400 E_6^4\rb
       	+ 6 E_2 E_4^2 E_6 \lb 24273 E_4^3 + 22463 E_6^2\rb. 
       \end{align}
       Using it together with \eqref{eq:nuxi}, we can compute the small $q$ expansion of $\nu^4 F_3 $. 
       Imposing the weak gap condition on the coefficients of $1$, $q$, and $q^2$, we obtain  a system of linear equations for $\alpha_{0,1,2}$ with a unitriangular matrix of coefficients
       \begin{align*}
       	\begin{cases} \displaystyle
       		\alpha_0  +\frac{497887}{716636160}   & = 
       		\dfrac{B_{2g}}{4g\lb g-1\rb } \Bigl|_{g=3} = 
       		\displaystyle
       		\frac{1}{1008},  
       		\\[+1.em]
       		\displaystyle
       		1848 \alpha_0+\alpha_1-\frac{72929}{5971968} & =   0,
       		\\[+1.em]
       		\displaystyle
       		1520904 \alpha_0+1104 \alpha_1+\alpha_2-\frac{48826219}{29859840} & =   0.
       	\end{cases}
       \end{align*}
       Its solution reads
       \begin{equation}
       	\alpha_0=\frac{1491431}{5016453120},\qquad \alpha_1= -\frac{222793}{414720},
       	\qquad \alpha_3=\frac{3421}{24},
       \end{equation}
       fixing thereby the holomorphic ambiguity \eqref{eq:h3modular} in the decomposition \eqref{eq:F3modular}.
       Finally, the small $q$ expansion for $\nu^4F_3$ becomes
       \begin{equation} \label{eq:F3qexp}
       	\nu^4F_3=\frac{1}{1008} +564480 q^4 +1614901401 q^5 +614954090130 q^6 + \frac{778563472560150 q^7}{7}+O\lb q^8\rb.
       \end{equation}
       It contains an extra vanishing coefficient of $q^3$, in agreement with the strong conifold gap property for $F_3$. 
       \end{eg}

       While we proved the existence of solutions by demonstrating 
       that the TR free energy  satisfies the conifold gap property (Theorem \ref{thm:coni-gap}), 
       we believe it would be also interesting to have a more direct proof 
       based on the properties of quasi-modular forms.

       
       \subsection{Refined holomorphic anomaly\label{subsec_betaHAE}}
       
       Remarkably, the HAE construction described above admits a one-parameter deformation \cite{KW10,HK10} (see also recent works \cite{FMP23,FGMS23} and references therein)  where $\hbar$ is replaced by by two complex parameters $\epsilon_1$ and $\epsilon_2$. For generic values of $\epsilon_1,\epsilon_2$, it is convenient to introduce the notation 
       \begin{equation}
       	\beta=\sqrt{-\frac{\epsilon_1}{\epsilon_2}}, \qquad  \hbar^2=\epsilon_1 \epsilon_2.
       \end{equation}
       In the description of this deformation, the  free energy $F_0$ remains undeformed and is defined by \eqref{eq:expression-F0-F1} which explicitly evaluates to \eqref{F0modular}.
       The free energy $F_1$ is given by  a $\beta$-deformation 
       \begin{equation}\label{eq:expression-F1b}
       	F_1^{(\beta)}=-\frac12 \ln \omega_A -\frac{\beta^2+\beta^{-2}}{24} \ln \frac{\Delta\lb q\rb}{\omega_A^{12}}. 
       \end{equation}
       This expression is used as an initial condition for the same HAE recursion \eqref{eq:HAE}, and defines  $F_g^{(\beta)}$ with $g\ge 2$. 
       Using 
       that $\Delta=\lb 2\pi/\omega_A\rb^{12} \Delta\lb q\rb$, it is easy to check that the free energy $F_1^{(\beta=1)}$ reduces to $F_1$ given by \eqref{eq:expression-F0-F1}  (up to an additive constant).
       The corresponding generalization of Theorem~\ref{thm:sol-HAE-weak} is given by 
       \begin{theo}\label{thm:sol-HAE-weak-b}
       	For any $g \ge 2$, the holomorphic anomaly equations \eqref{eq:HAE} 
       	supplemented with the  $\beta$-deformed \underline{weak} conifold gap condition 
       	\begin{equation} \label{eq:q-expansion-Fg-weak-b}
       		\nu^{2g-2} F_g^{(\beta)}  = \kappa_g^{(\beta)} + O\lb q^{d_g}\rb,  
       	\end{equation}
       	where $d_g=\dim \mathcal{M}_{14\lb g-1\rb}$ and 
       	\begin{equation}\label{eq:kappa-g-b}
       		\kappa_g^{(\beta)}= - \lb 2g-3\rb! \sum_{h=0}^g \hat B_{2h} \hat B_{2g-2h} \beta^{2g-4h},
       		\qquad \hat B_m=\lb1- 2^{1-m}\rb\frac{B_m}{m!},
       	\end{equation} 
       	define $F_g^{(\beta)}$ uniquely by recursion in $g$ starting from $F_1^{(\beta)}$ given by \eqref{eq:expression-F1b}.
       \end{theo}
       
       \begin{eg} For $g=2$, the HAE recursion yields the following expression:
       \begin{align}
       	F_2^{(\beta)}=&\,-\left(\frac{\xi}{E_4}\right)^2  \left[   \frac{5 E_2^3 E_4^2}{41472} 
       	+\frac{\left(2 \beta^2+7 +2\beta^{-2}\right) E_2^2 E_4E_6}{13824 } 
       	+   \frac{ \left(6 \beta^2-5+6 \beta^{-2}\right) E_2 E_4^3}{6912}\right.\\
       	\nonumber &\,\left.+   \frac{ \left(\beta^4+8 \beta^2-7 +8 \beta^{-2}+\beta^{-4}\right) E_2E_6^2}{13824}+\frac{\left(237 \beta^4-330 \beta^2+485 -330 \beta^{-2}+237 \beta^{-4}\right) E_4^2 E_6}{207360}\right].
       \end{align}      
       Its small $q$ expansion  has the form
       \begin{equation}
       	\nu^2F_2^{(\beta)}=-\frac{7 \beta^4+10+7\beta^{-4}}{5760 } 
       	+\frac{7 \left(14497 \lb \beta^4+ \beta^{-4}\rb-39600 \lb \beta^2+ \beta^{-2}\rb+52510\right) }{480} q ^2 + 
       	O\lb q^3\rb ,
       \end{equation}
       which exhibits the anticipated gap (vanishing coefficient of $q$).
       \end{eg}
       
       Motivated by the TR/\PIeq correspondence \eqref{TRPIshort}, let us introduce the $\beta$-deformed \PIeq partition function by
       \beq\label{betadefZ}
       \tilde{\mathcal Z}^{(\beta)}\lb \hbar^{-1}s\,|\,\hbar^{-1}\nu\rb= \exp\sum_{g \ge 0} \hbar^{2g-2} F_g^{(\beta)}\lb t,\nu\rb.
       \eeq
       Taking into account Conjecture~\ref{CFTPI} relating the \PIeq partition function to $c=1$ Virasoro conformal blocks, it is natural to guess that its $\beta$-deformed version $\tilde{\mathcal Z}^{(\beta)}\lb s\,|\,\nu\rb$ is similarly related to the irregular block  \eqref{CBalgexp}--\eqref{L2comms} for generic central charge. The precise statement is as follows\footnote{We recall once again that $F_1$ in \eqref{eq:expression-F1b} is defined up to an additive constant independent of $\nu$ and $s$; it may, however, depend on $\beta$. }.
       \begin{conj}\label{CFTHAEcorr}
       	The large $s$ asymptotic expansion of the $\beta$-deformed partition function $\tilde{\mathcal Z}^{(\beta)}\lb s\,|\,\nu\rb$ defined by \eqref{betadefZ} is given by
       	\begin{gather}
       	\tilde{\mathcal Z}^{(\beta)}\lb s\,|\,\nu\rb=C^{\lb\beta\rb}\lb\nu\rb e^{\frac{s^2}{45}+\frac{4}{5}i\nu s}\mathcal F\lb\varepsilon\,|\,\nu\rb,     		
       	\end{gather}
       	where $\mathcal F\lb\varepsilon\,|\,\nu\rb$ is the irregular conformal block defined by \eqref{CBalgexp}--\eqref{L2comms}, with
       	\beq\label{beta_identification}
       	c=1-6\lb\beta-\beta^{-1}\rb^2,\qquad \varepsilon^{-2}=48is.
       	\eeq
       	The $s$-independent prefactor $C^{\lb\beta\rb}\lb\nu\rb$ is given by
       	\beq
       	C^{\lb\beta\rb}\lb\nu\rb=\Gamma_{\beta}^{-1}\left(\frac{\tilde Q}{2}  +\nu\right),\qquad \tilde Q= \beta+\beta^{-1},
       	\eeq
       	 where $\Gamma_\beta(z)$ is expressed in terms of the Barnes gamma and zeta functions (see e.g. \cite{Spr09} for the corresponding definitions)
       	\begin{equation}\label{eq:BarnesGamma}
       		\ln \Gamma_\beta\lb z\rb =
       		\ln \frac{\Gamma_2\lb z;\beta,\beta^{-1}\rb}{\Gamma_2\lb \tilde Q/2;\beta,\beta^{-1}\rb}=
       		\left. \frac{\partial}{\partial s} \left(\zeta_2\lb s;\beta,\beta^{-1},z\rb - \zeta_2\lb s;\beta,\beta^{-1},\tilde Q/2\rb \right)\right|_{s=0}.
       	\end{equation}
       \end{conj}
       
             We have checked that previously computed coefficients $\mathcal U_{\ln}$, $\mathcal U_1,\ldots,\mathcal U_5$ of the expansion \eqref{Ugenericc} indeed coincide with the expressions predicted by Conjecture~\ref{CFTHAEcorr} under identification \eqref{beta_identification}. This comparison involves the calculation of $\beta$-deformed free energies up to $F_3^{(\beta)}$.
       
        The role of the $\Gamma_\beta$-function in the above is
       to reproduce the leading terms $\kappa_g^{(\beta)}$ of expansions \eqref{eq:q-expansion-Fg-weak-b} of $F_g^{(\beta)}$. As $z\to\infty$ with $|\arg z|<\frac{\pi}{2}$, we have
       \begin{equation}\label{betaBarnesAs}
       	\ln \Gamma^{-1}_{\beta}\left(\frac{\tilde Q }{2} +z\right) \simeq
       	\frac{z^2}{2}\ln z- \frac{3z^2}{4}-\frac{\beta^2+\beta^{-2}}{24} \ln z +
       	c_0\lb \beta\rb+\sum_{g\ge 2}   \frac{\kappa_g^{(\beta)}}{z^{2g-2}}.
       \end{equation}
       This follows from the integral representation 
       \beq
       \ln \Gamma^{-1}_{\beta}\left(\frac{\tilde Q }{2} +z\right)=\int_0^\infty\frac{dx}{x}
       \left[\frac{1-e^{-zx}}{4\sinh \frac{\beta x}{2}\sinh\frac{x}{2\beta}}+\frac{z^2}{2}e^{-x}-\frac{z}{x}\right],\qquad \Re z>0.
       \eeq
       Indeed, decompose the integral on the right into a sum of three contributions:
       \beq\begin{gathered}
       	\int_0^\infty\left[\lb\frac{z^2}{2}e^{-x}-\frac{z}{x}\rb-\lb \frac{1}{x^2}-\frac{\beta^2+\beta^{-2}}{24}\rb e^{-zx}+\lb \frac{1}{x^2}-\frac{\beta^2+\beta^{-2}}{24}e^{-x}\rb\right]\frac{dx}{x}\\
       	-
       	\int_0^\infty \left(\frac{1}{x^2}-\frac{\beta^2+\beta^{-2}}{24}e^{-x}-\frac{1}{4\sinh \frac{\beta x}{2}\sinh\frac{x}{2\beta}}\right)\frac{dx}{x}+
       	\int_0^\infty \left(\frac{1}{x^2}-\frac{\beta^2+\beta^{-2}}{24}-\frac{1}{4\sinh \frac{\beta x}{2}\sinh\frac{x}{2\beta}}\right)\frac{e^{-z x}}{x}\,dx.
       	\end{gathered}
       \eeq
       The first of these integrals is elementary and evaluates to $\frac{z^2}{2}\ln z- \frac{3z^2}{4}-\frac{\beta^2+\beta^{-2}}{24} \ln z$, which is exactly the divergent part of the asymptotics in \eqref{betaBarnesAs}. The second integral gives its finite part $c_0\lb\beta\rb$. Finally, in the third integral one may use the well-known expansion
       \beq
       \frac{1}{2\sinh x}=\sum_{g=0}^\infty\hat{B}_{2g}x^{2g-1}
       \eeq
       to obtain the last sum in \eqref{betaBarnesAs}. It seems plausible that the asymptotic expansion \eqref{betaBarnesAs} holds in a wider sector $|\arg z|<\pi$, similarly to \cite[Proposition 8.11]{Spr09}. In the non-deformed case, the $\Gamma_\beta$-function reduces to the usual Barnes G-function, $\left.\Gamma^{-1}_{\beta}\left(\tilde Q /2 +\nu\right) \right|_{\beta=1} = \lb 2\pi\rb^{-\nu/2} G\lb 1+\nu\rb$.

    \appendix
    
    \section{Irregular states}\label{AppIrrStates}
    \subsection{Higher-order descendants\label{AppA1}}
    We record here the higher order contributions to the rank $\frac52$ Whittaker state $|\Psi\rangle$ in \eqref{Wh52exp}.\vspace{0.2cm}\\
    \underline{Order 3}:
    \beq\label{G3}
    \begin{aligned}
    	\mathbb G_3=&\,\frac{343}{135}L_{-3}-\frac79 L_{-2}L_{-1}+\frac{4}{81}L_{-1}^3-\frac{56\nu}{9}L_{-2}L_{1}+\frac{32\nu}{27}L_{-1}^2L_1-
    	\frac{206\nu}{27}L_{-1}L_0\\&\,+
    	\frac{256\nu^2+57}{27} L_{-1}L_1^2-\frac{16\lb 515\nu^2+146\rb}{135}L_0L_1+\frac{8\nu\lb 256\nu^2+171\rb}{81}L_1^3\\&\,+\frac{14648\nu}{405} L_{-1}+\frac{4\lb 31636\nu^2+4401-1029c\rb}{405}L_1.
    \end{aligned}
    \eeq
    \underline{Order 4}:
    \beq\label{G4}
    \begin{aligned}
    	\mathbb G_4=&\,-\frac{5551}{810}L_{-4}+\frac{686}{405}L_{-3}L_{-1}+\frac{49}{72}L_{-2}^2-\frac{7}{27}L_{-2}L_{-1}^2+\frac{2}{243}L_{-1}^4+\frac{5488\nu}{405}L_{-3}L_1-\frac{112\nu}{27}L_{-2}L_{-1}L_1\\
    	&\,+\frac{64\nu}{243}L_{-1}^3L_1+\frac{721\nu}{54}L_{-2}L_0-\frac{7\lb 256\nu^2+57\rb}{108}L_{-2}L_1^2	-\frac{206\nu}{81}L_{-1}^2L_0+\frac{256\nu^2+57}{81}L_{-1}^2L_1^2\\
    	&\,-\frac{32\lb 515\nu^2+146\rb}{405}L_{-1}L_0L_1+\frac{16\nu\lb256\nu^2+171\rb}{243}L_{-1}L_1^3
    	+\frac{65536\nu^4+87552\nu^2+9747}{1944}L_1^4\\
    	&\,+\frac{10609\nu^2+4209}{162}L_0^2-\frac{\nu\lb 131840\nu^2+104107\rb}{810}L_0L_1^2
    	-\frac{108661\nu}{810}L_{-2}+\frac{29296\nu}{1215}L_{-1}^2
    	\\
    	&\,+\lb \frac{487456\nu^2}{1215}-\frac{2744c}{405}+\frac{48043}{540}\rb L_{-1}L_1
    	-\lb \frac{1464989\nu^2}{1215 } - \frac{10619c}{216} + \frac{50929}{216}\rb L_0\\
    	&\,+\frac{\nu\lb 4049408\nu^2-131712c+2252067\rb}{2430} L_1^2.
    \end{aligned}
    \eeq
    \begingroup
    \underline{Order 5}:
    \begin{align}
    	\nonumber \mathbb G_5=&\,\frac{30851}{1620}L_{-5}-\frac{5551}{1215}L_{-4}L_{-1}-\frac{2401}{810}L_{-3}L_{-2}+
    	\frac{686}{1215}L_{-3}L_{-1}^2+\frac{49}{108}L_{-2}^2L_{-1}-\frac{14}{243}L_{-2}L_{-1}^3
    	+\frac{4}{3645}L_{-1}^5\\ 
    	\nonumber &\,-\frac{44408\nu}{1215}L_{-4}L_1+\frac{10976\nu}{1215}L_{-3}L_{-1}L_1+\frac{98\nu}{27} L_{-2}^2 L_1
    	-\frac{112\nu}{81} L_{-2} L_{-1}^2 L_1 +\frac{32 \nu}{729} L_{-1}^4 L_1 \\
    	\nonumber &\,-\frac{35329\nu}{1215} L_{-3} L_0 +\frac{343  \left(256 \nu ^2+57\right)}{2430}L_{-3}
    	L_1^2+\frac{721\nu}{81} L_{-2}L_{-1}  L_0-\frac{7 \left(256 \nu ^2+57\right)}{162}
    	L_{-2} L_{-1} L_1^2  \\
    	\nonumber &\,-\frac{412\nu}{729}  L_{-1}^3L_0 +\frac{2\left(256 \nu ^2+57\right)}{729} L_{-1}^3L_1^2 +\frac{56\left(515 \nu
    		^2+146\right)}{405} L_{-2} L_0 L_1  -\frac{28\nu\left(256 \nu ^2+171 \right)}{243} L_{-2} L_1^3 \\
    		\nonumber &\,-\frac{32 \left(515
    			\nu ^2+146\right)}{1215}L_{-1}^2 L_0 L_1 +\frac{16\nu\left(256 \nu ^2+171 \right)}{729}L_{-1}^2  L_1^3 +\frac{\left(10609 \nu ^2+4209\right)}{243} L_{-1} L_0^2 \\
    			\nonumber &\,-\frac{\nu \left(131840 \nu ^2+104107 \right)}{1215}L_{-1}L_0 L_1^2 + \frac{ 65536 \nu ^4+87552 \nu ^2+9747}{2916}L_{-1}L_1^4  +\frac{8\nu  \left(53045 \nu ^2+51121 \right)}{1215}L_0^2 L_1\\
    	\nonumber &\,-\frac{8  \left(131840 \nu ^4+200193 \nu ^2+24966\right)}{3645}	L_0 L_1^3	+ \frac{2\nu \left(65536 \nu ^4+145920 \nu ^2+48735 \right)}{3645} L_1^5+\frac{551393 \nu }{1215} L_{-3} \\
    	\label{G5} &\,	-\frac{159929 \nu }{1215} L_{-2} L_{-1}+\frac{29296  \nu }{3645} L_{-1}^3-\frac{2 \left(656096 \nu ^2-7203 c+172695\right)}{1215}L_{-2} L_1 \\
    	\nonumber &\,+\frac{481216 \nu ^2 -5488 c+120657}{2430}L_{-1}^2L_1  -\lb \frac{1479574 \nu^2}{1215}-\frac{10619 c}{324}+\frac{37435}{108} \rb L_{-1}L_0 \\
    	\nonumber &\,+\frac{7 \nu \left(282112 \nu ^2-6272 c+182655  \right)}{1215} L_{-1} L_1^2-\frac{2 \nu \left(6078976 \nu ^2-229943 c+4309981 
    		\right)}{1215}L_0 L_1\\
    	\nonumber &\,+\frac{2 
    		\left(8098816 \nu ^4-263424 c \nu ^2+9684864 \nu ^2-58653c+1187289\right)}{3645}L_1^3\\
    	\nonumber &\,	+\frac{2
    		 \left(9362948 \nu ^2-465165 c+2269809\right)}{3645}L_{-1}+\frac{2 \nu\left(80655268 \nu ^2-7880031 c  +35907717 
    		 \right)}{3645}L_1. 
    	 \end{align}
      To appreciate the achieved simplifications, the reader may wish to compare  $\mathbb G_3$ (and $\mathbb G_2$ given earlier in \eqref{G2}) with their counterparts in \cite[Eqs (A.2)--(A.3)]{PP23}.\vspace{0.2cm}\\
     For $\nu=0$, the expression of  Whittaker state simplifies as the previous formulas reduce to
     \begin{subequations}
      \begin{align}
      &\mathbb{G}_1=\,\frac23 L_{-1},\qquad	\mathbb{G}_2=-\frac{7}{6}  L_{-2}+\frac{2}{9}L_{-1}^2+\frac{19}{6} L_1^2,\\
     &\mathbb{G}_3= \frac{343}{135} L_{-3}-\frac{7}{9} L_{-2}L_{-1}+\frac{4}{81} L_{-1}^3+\frac{19}{9} L_{-1}L_1^2-\frac{2336}{135} L_0L_1-\frac{4\lb 343 c-1467\rb}{135}  L_1,\\
     &\mathbb{G}_4= -\frac{5551}{810} L_{-4}+\frac{686}{405} L_{-3}L_{-1}+\frac{49}{72} L_{-2}^2-\frac{7}{27} L_{-2}L_{-1}^2+\frac{2}{243}  L_{-1}^4-\frac{133}{36} L_{-2}L_1^2-\frac{4672}{405} L_{-1}L_0L_1\\
     \nonumber &\qquad +\frac{19}{27} L_{-1}^2L_1^2-\frac{10976 c-144129}{1620}L_{-1}L_1+\frac{1403}{54} L_0^2+\frac{361}{72} L_1^4+\frac{10619c -50929 }{216}L_0,\\
     &\mathbb{G}_5=\frac{30851 }{1620}L_{-5}-\frac{5551 }{1215}L_{-4}L_{-1}
     -\frac{2401}{810} L_{-3}L_{-2}+\frac{686 }{1215}L_{-3}L_{-1}^2+\frac{49}{108} L_{-2}^2L_{-1}-\frac{14}{243} L_{-2}L_{-1}^3+\frac{4 }{3645}L_{-1}^5\\ 
     \nonumber&\qquad
     +\frac{6517}{810} L_{-3}L_1^2-\frac{133}{54} L_{-2}L_{-1}L_1^2+\frac{8176}{405} L_{-2}L_0L_1-\frac{4672 }{1215}L_{-1}^2L_0L_1+\frac{38}{243}L_{-1}^3L_1^2+\frac{1403}{81} L_{-1}L_0^2\\
     \nonumber&\qquad  +\frac{361}{108} L_{-1}L_1^4 -\frac{22192}{405} L_0L_1^3+\frac{2(2401 c-57565)}{405} L_{-2}L_1+\frac{ 10619 c-112305}{324} L_{-1}L_0\\
     \nonumber&\qquad -\frac{5488 c-120657 }{2430} L_{-1}^2L_1-\frac{2(6517 c-131921)}{405}  L_1^3 -\frac{2 (51685c-252201) }{405} L_{-1}.
      \end{align}
      In this case, the first non-vanishing coefficient in the conformal block, $\mathcal U_2\lb 0\rb=\frac{7 \left(14497 c^2-139322 c+207769\right)}{17280}$, is determined by the next term in the expansion of $|\Psi\rangle$,
      \begin{align}
      	\mathbb{G}_6=&\,-\frac{210338 }{3645}L_{-6}+\frac{30851 }{2430}L_{-5}L_{-1}+\frac{38857 }{4860}L_{-4}L_{-2}-\frac{5551 }{3645}L_{-4}L_{-1}^2+\frac{117649 }{36450}L_{-3}^2-\frac{2401 }{1215}L_{-3}L_{-2}L_{-1}\\
      	\nonumber&\, +\frac{1372 }{10935}L_{-3}L_{-1}^3-\frac{343 }{1296}L_{-2}^3+\frac{49}{324} L_{-2}^2L_{-1}^2-\frac{7}{729} L_{-2}L_{-1}^4
      	+\frac{4 }{32805}L_{-1}^6-\frac{105469 }{4860}L_{-4}L_1^2\\
      	\nonumber&\,+\frac{6517 }{1215}L_{-3}L_{-1}L_1^2-\frac{801248 }{18225}L_{-3}L_0L_1+\frac{931}{432} L_{-2}^2L_1^2-\frac{133}{162} L_{-2}L_{-1}^2L_1^2+\frac{16352 }{1215}L_{-2}L_{-1}L_0L_1\\
      	\nonumber&\,-\frac{9821}{324} L_{-2}L_0^2-\frac{2527}{432} L_{-2}L_1^4+\frac{19}{729} L_{-1}^4L_1^2-\frac{9344 }{10935}L_{-1}^3L_0L_1+\frac{1403}{243} L_{-1}^2L_0^2
      	+\frac{361}{324} L_{-1}^2L_1^4+\frac{6859 }{1296}L_1^6\\
      	\nonumber&\,-\frac{44384 }{1215}L_{-1}L_0L_1^3+\frac{16911617 }{72900}L_0^2L_1^2
      	-\frac{3764768 c-135441381}{145800}L_{-3}L_1+\frac{76832 c-2522375}{9720}L_{-2}L_{-1}L_1\\
      	\nonumber &\,-\frac{7 (2389275c-43148897) }{291600}L_{-2}L_0-\frac{7 (1568 c-48357) }{21870}L_{-1}^3L_1+\frac{10619 c-173681}{972} L_{-1}^2L_0\\
      	\nonumber&\,-\frac{208544 c-6067987}{9720}L_{-1}L_1^3+\frac{96676097c-2155846795}{291600}L_0L_1^2+
      	\frac{7 (28280777 c-143573431)}{145800}L_{-2}\\
      	\nonumber &\,-\frac{4961760 c-34601027}{29160}L_{-1}^2+
      	\frac{3764768 c^2-533275715c+2482839717}{72900}L_1^2.
      \end{align}
      \end{subequations}
     The expression of $\mathbb G_k\bigl|_{\nu=0}$ for arbitrary $k$ satisfies an extra constraint $m_0+m_1=\frac{k-\ell}{2}\;\operatorname{mod}\;2$ in addition to selection rules already present in \eqref{selectionrulesWH}.
     
     \subsection{Proof of Theorem~\ref{conjWh52}\label{AppProof}} 
     
In this Appendix, we prove Theorem~\ref{conjWh52} about embedding of  the rank $\frac52$ Whittaker vector $|\Psi\rangle$ 
into the rank 2 Whittaker module ${\mathcal V}^{[2]}$.
The proof uses several definitions and ideas from \cite{Nagoya15}. 

Slightly generalizing the initial setting, let $\left|\Psi\right\rangle$  denote a vector satisfying
\begin{equation}
    \lb L_n-\mathcal{L}_n^{(5/2)}\rb\left|\Psi\right\rangle=0, \qquad n\ge3, 
\end{equation}
where
\begin{equation}
    \mathcal{L}_n^{(5/2)}=\begin{cases}
        \Lambda_n, \qquad & n=3,4,5,\\
        0, \qquad & n>5,
    \end{cases}
\end{equation}
and it is assumed that $\Lambda_4\ne  0$. Likewise,  the rank 2 Whittaker module ${\mathcal V}^{[2]}$  will be generated from a vector $|J\rangle$ 
that satisfies
\begin{equation}
    \lb L_n-\mathcal{L}_n^{(2)}\rb|J\rangle=0, \qquad n\ge2, 
\end{equation}
\begin{equation}\label{Lambda234}
    \mathcal{L}_n^{(2)}=\begin{cases}
        \Lambda_n, \qquad &n=2,3,4,\\
        0, \qquad &n>4.
    \end{cases}
\end{equation}

We are going to prove that the Whittaker vector $|\Psi\rangle$ can be embedded into the module ${\mathcal V}^{[2]}$ in the form of a formal power series in $\Lambda_5$:
\begin{equation}\label{PsiV2ser}
    |\Psi\rangle=\sum_{k=0}^{\infty}\Lambda_5^k \left|\Psi_k\right\rangle, \qquad \left|\Psi_0\right\rangle=|J\rangle.
\end{equation}
To find $|\Psi_k\rangle\in {\mathcal V}^{[2]}$, we will use a PBW basis labeled by Young diagrams  \cite{Nagoya15}
\begin{equation}\label{PBW_app}
    v_\lambda=\mathbf{L}_{-\lambda} |J\rangle, \qquad \mathbf{L}_{-\lambda}=L_{-\lambda_1+2}\cdots L_{-\lambda_\ell+2}, \qquad \lambda=\lb \lambda_1,\ldots,\lambda_\ell\rb\in\mathbb Y, 
    \qquad \lambda_1\ge \ldots \ge \lambda_n \ge 1,
\end{equation}
\begin{equation}
    v_\emptyset=|J\rangle, \quad \mathbf{L}_{-\emptyset} = 1.
\end{equation}
The difference between $\mathbf{L}_{-\lambda}$ and the notation $\mathbb{L}_{-\lambda}$ previously used in \eqref{notation_young} is the shift of indices of all Virasoro generators by $2$.

Let us define a degree for the basis elements by
\begin{equation}\label{degL}
    \deg\,v_\lambda=\deg\,\mathbf L_{-\lambda} \left|J\right\rangle=|\lambda|.
\end{equation}
The degree of  an arbitrary vector $v \in {\mathcal V}^{[2]}$ is defined as the maximal degree of its non-zero components, 
\begin{equation}
    \mathrm{deg}\,v=\max\,\lb\mathrm{deg}\,  v_\lambda\rb, \qquad v=\sum_\lambda c_\lambda v_\lambda, \qquad c_\lambda \ne 0.
\end{equation}
With respect to this degree, we are going to consider an infinite flag  $U_0\subset U_1\subset U_2\subset \ldots$ of subspaces $U_m \subset {\mathcal V}^{[2]}$ given by
\begin{equation}\label{DefUm}
    U_m=\mathrm{span}\,\left(v_\lambda\, | \,\mathrm{deg}\,v_\lambda \leq m\right).
\end{equation}
It will also be convenient to introduce a linear functional $\xi$ on  ${\mathcal V}^{[2]}$  defined as the coefficient of $|J\rangle$ in the decomposition in the PBW basis \eqref{PBW_app},
\begin{equation}
    \xi(v)=c_{\emptyset}, \qquad v=\sum_\lambda c_\lambda v_\lambda.
\end{equation}

We will use the following modification  of Virasoro generators:
\begin{equation}
    \widetilde{L}_n=\begin{cases}
        L_n-\Lambda_n, \qquad & n=3,4,\\
        L_n, \qquad & n>4.
    \end{cases}
\end{equation}
The products of such shifted generators will be similarly labeled by a Young diagram $\mu\in \mathbb Y$, 
\begin{equation}
    \widetilde{\mathbf L}_\mu = \widetilde{L}_{\mu_1+2}\cdots \widetilde{L}_{\mu_n+2}, \qquad \mu=\lb \mu_1,\ldots, \mu_n\rb, \qquad 
    \mu_1\ge \mu_2 \ge \cdots \ge \mu_n\ge 1.
\end{equation}
Note that the products $\mathbf L_{-\lambda}$ and $\widetilde{\mathbf L}_\mu$ involve the generators $L_n$ with $n\leq 1$ and $n\geq 3$, respectively.
We further define a bilinear form $M$ on ${\mathcal V}^{[2]}$ by
\begin{equation}
    M\lb v_\mu, v_\lambda\rb=\xi\left(\widetilde{\mathbf L}_\mu \mathbf L_{-\lambda}\left|J\right\rangle\right).
\end{equation}

The main result we need  from \cite{Nagoya15} is Lemma~2.22 therein, which states that 
\begin{equation}\label{Mform}
    M\lb v_\mu, v_\lambda\rb=\begin{cases}
        0, \qquad &|\mu|>|\lambda|, \\
        0, \qquad &|\mu|=|\lambda|,\ \mu\neq \lambda, \\
        \neq 0, \qquad &\mu=\lambda.
    \end{cases}
\end{equation}
Let us introduce the matrix $M_k$ of pairings defined by the form $M$ restricted to the space $U_k$, 
\begin{equation}
    M_k = \left(M\lb v_\mu,v_\lambda\rb\right)_{0\leq |\mu|,|\lambda|\leq k}.
\end{equation}
Ordering the basis vectors in descending order with respect to the size of Young diagrams, 
it follows from \eqref{Mform} that $M_k$ becomes a triangular matrix with non-zero determinant. 
Note that the ordering of basis vectors corresponding to Young diagrams of the same size  does not affect the triangularity of $M_k$ thanks to 
\eqref{Mform}.

The problem of finding the rank $\frac52$ Whittaker vector $|\Psi\rangle$ in the form \eqref{PsiV2ser} 
can be rewritten as a problem of solving the following recurrence relations for descendants $\left|\Psi_k\right\rangle$ in terms of the action of shifted Virasoro generators $\widetilde{L}_n$:
\begin{equation}\label{recdesc}
    \widetilde{L}_n \left|\Psi_k\right\rangle = \begin{cases}
        \left|\Psi_{k-1}\right\rangle, \qquad &n=5,\\
        0, \qquad &n\ge 3, n\neq 5.
    \end{cases}
\end{equation}
In order to ensure uniqueness of solution,  we impose the
{\em orthogonal gauge} defined by the following conditions on the components of $|\Psi\rangle$:
\begin{equation}\label{orthgauge}
    \xi\left(\left|\Psi_k\right\rangle\right)=\delta_{k,0}.
\end{equation} 
From \eqref{recdesc} it follows that
\begin{equation}
    \xi\left(\cdots\widetilde{L}_5^{m_5}\widetilde{L}_4^{m_4}\widetilde{L}_3^{m_3} |\Psi_k\rangle\right)=\begin{cases}
        1, \qquad m_5=k, m_3=m_4=m_6=\cdots = 0,\\
        0, \qquad  \text{otherwise}, 
    \end{cases}
\end{equation}
or, equivalently,
\begin{equation}\label{xiLtPsi}
    \xi\lb \widetilde{\mathbf L}_\mu \left|\Psi_k\right\rangle\rb=\begin{cases}
        1, \qquad \mu=\lb 3^k\rb,\\
        0, \qquad \text{otherwise}.
    \end{cases}
\end{equation}

\begin{theo} \label{Wh52deg} 
The following holds:
\begin{enumerate}
\item The rank $\frac52$ Whittaker vector $\left|\Psi\right\rangle$ in the form \eqref{PsiV2ser} exists.
\item For any $k\in\mathbb N$, its component $\left|\Psi_k\right\rangle$ belongs to $U_{3k}$. 
\item The vector $|\Psi\rangle$ is unique in the orthogonal gauge \eqref{orthgauge}.
\end{enumerate}
\end{theo}
\pf 
We will prove the above statements by induction in $k$ for components $ \left|\Psi_k\right\rangle$ in \eqref{PsiV2ser}. 
Assume that $\left|\Psi_{k-1}\right\rangle\in U_{3(k-1)}$ satisfies the recursion \eqref{recdesc}. 
We are looking for $\left|\Psi_k\right\rangle$ in the form
\begin{equation}\label{psikgen}
    \left|\Psi_k\right\rangle=\sum_{1\leq |\lambda|\leq N}c_\lambda^{(k)}\mathbf L_{-\lambda}\left|J\right\rangle,
\end{equation}
where $N$ is an integer satisfying $N\ge 3k$. Let us show that such a vector exists and is unique. 
Moreover, it will be shown that $c_\lambda^{(k)}=0$ for $|\lambda| > 3k$ and therefore we may set $N=3k$. 

\noindent
{\it \underline{Step 1} (uniqueness).}  We will prove uniqueness
by showing that the coefficients $c_\lambda^{(k)}$ satisfy a system of linear equations having a unique solution. 
This system comes from computing $\xi\lb \widetilde{\mathbf L}_\mu \left|\Psi_k\right\rangle\rb$ for $1\le |\mu|\leq N$ in two different ways: by 
direct computation using \eqref{psikgen} and by recursion relations \eqref{recdesc} for the components $\left|\Psi_j\right\rangle$ of $|\Psi\rangle$ leading to 
\eqref{xiLtPsi}:
\begin{equation}\label{syst1step}
    \sum_{1\leq |\lambda|\leq N}M\lb v_\mu,v_\lambda\rb c_\lambda^{(k)}=\xi\lb \widetilde{\mathbf L}_\mu \left|\Psi_k\right\rangle\rb=\delta_{\mu,\lb 3^k\rb}.
\end{equation}
Due to \eqref{Mform},  the matrix of the system is non-degenerate, which implies uniqueness of solution for $c_\lambda^{(k)}$.
As a corollary,  with such $c_\lambda^{(k)}$, the vector $|\Psi_k\rangle$ ($k\in\mathbb N$) given by \eqref{psikgen} satisfies
\begin{equation}\label{relstep2}
    \xi\lb \widetilde{\mathbf L}_\mu \left|\Psi_k\right\rangle\rb=\delta_{\mu,\lb 3^k\rb}
\end{equation}
for any $\mu\in\mathbb Y$. Indeed, for $\mu=\emptyset$ this follows from  \eqref{orthgauge}, the case $0<|\mu|\leq N$ follows from 
\eqref{syst1step}, and the case $|\mu|>N$  follows from the first relation in \eqref{Mform}.\vspace{0.1cm}

\noindent
{\it \underline{Step 2} (restriction of support of $|\Psi_k\rangle$).} Let us show that 
the solution of \eqref{syst1step} satisfies $c_\lambda^{(k)}=0$ for $|\lambda| > 3k$.
The starting point here is that
the only non-zero coefficient in the right hand side of \eqref{syst1step}  corresponds to $\mu=\lb 3^k\rb$.
Therefore, if we order the basis vectors of $U_N$ in \textit{any} descending order, denoted by $\succ$ (not necessarily coming from the degree introduced above), such that the matrix of coefficients 
of \eqref{syst1step}  is lower triangular, then $c_\lambda^{(k)}=0$ for all $\lambda \succ \lb 3^k\rb$.

In particular, due to \eqref{Mform}, for any order of basis elements $v_\lambda$  consistent with the non-ascending degree \eqref{degL}, 
the corresponding matrix of coefficients becomes  lower triangular  and hence $c_\lambda^{(k)}=0$ for $\left|\lambda\right| > |\lb 3^k\rb|= 3k$.
Therefore, if $\left|\Psi_k\right\rangle$ exists, then it should verify $\left|\Psi_k\right\rangle\in U_{3k}$ and  we may impose the restriction $N=3k$. 

In what follows, a similar idea will be used to even further restrict the subspace of $\mathcal V^{[2]}$ to which $\left|\Psi_k\right\rangle$ belongs using other orderings of basis elements.\vspace{0.1cm}

\noindent
{\it \underline{Step 3} (existence).} 
The requirement  for $\left|\Psi_k\right\rangle$ in the form \eqref{psikgen} to satisfy  \eqref{recdesc}
leads to an overdetermined linear system for the coefficients $c_\lambda^{(k)}$ in \eqref{psikgen}.
Its subsystem \eqref{syst1step} has a unique solution and one needs to check that it solves the whole overdetermined system too.
Let us show that \eqref{psikgen}, with $c_\lambda^{(k)}$ defined by the solution of \eqref{syst1step}, indeed  satisfies \eqref{recdesc}.
Equivalently, we need to check that for $n\ge 3$ the vectors 
\begin{equation}\label{syst2step}
    w_n=\widetilde{L}_{n} |\Psi_k\rangle-\delta_{n,5}|\Psi_{k-1}\rangle \in U_{3k-n+2}
\end{equation}
are equal to zero. (The restriction on the degree comes from Lemma~2.19 of \cite{Nagoya15}.)
In turn, thanks to non-degeneracy of the bilinear form $M_m$ on $U_m$, it is sufficient to show that
\begin{equation}
    M\lb v_\mu,w_n\rb=0
\end{equation}
for all $n\ge 3$ and all $|\mu|\leq 3k-n+2$.

For $n\neq 5$, we need to show that
\begin{equation}
    \xi\left(\widetilde{\mathbf L}_\mu \widetilde{L}_{n} \left|\Psi_k\right\rangle\right)=0,
\end{equation}
which follows from \eqref{relstep2} by PBW ordering of $\widetilde{\mathbf L}_\mu \widetilde{L}_{n}$.
Similar arguments work for $n=5$ with $\mu\ne\lb 3^{k-1}\rb$.
In the case $\mu=\lb 3^{k-1}\rb$, we have
\begin{equation}
    \xi\left(\widetilde{L}_5^{k-1}\left( \widetilde{L}_{5} \left|\Psi_k\right\rangle-\left|\Psi_{k-1}\right\rangle\right)\right)=\xi\left(\widetilde{L}_5^{k} \left|\Psi_k\right\rangle\right)-\xi\left(\widetilde{L}_5^{k-1} \left|\Psi_{k-1}\right\rangle\right)=0.
\end{equation}
The 1st term is equal to $1$ by \eqref{relstep2} and the 2nd term is equal to $1$ by induction.  This completes the proof. \epf

\begin{rmk}    
The explicit form of $\left|\Psi_1\right\rangle$ reads
\begin{equation}\label{Psi1deg}
    \left|\Psi_1\right\rangle=\left(\frac{1}{6\Lambda_4}L_{-1}-\frac{5\Lambda_3}{24\Lambda_4^2}L_0+\frac{15\Lambda_3^2-16\Lambda_2\Lambda_4}{48\Lambda_4^3}L_1+0\cdot L_1^2+0\cdot L_0 L_1+0\cdot L_1^3\right)
    \left|J\right\rangle.
\end{equation}
Observe that some coefficients of the decomposition of $\left|\Psi_1\right\rangle\in U_3$ are equal to zero. The number of vanishing coefficients becomes even larger for higher orders. 
This means that the restriction $\left|\Psi_k\right\rangle\in U_{3k}$ is not optimal and
$\left|\Psi_k \right\rangle \in V_k \varsubsetneq U_{3k}$.
\end{rmk}
We are now going to improve the statement of Theorem~\ref{Wh52deg} by explaining why a part of the coefficients in the decompositions such as \eqref{Psi1deg} are equal to zero. 
To do this, let us introduce two more degrees for basis elements $v_\lambda$, denoted  by $\deg_{\,\delta}$ with $\delta=1,2$ and defined by
\begin{equation}
    {\mathrm{deg}_{\,\delta}}v_\lambda=n_1+\delta n_2+\sum_{k>2}n_k\lb k-2\rb, \qquad \lambda=\lb\cdots\, 3^{n_3}2^{n_2}1^{n_1}\rb,
\end{equation}
or, in other words,
\begin{equation}
    {\mathrm{deg}_{\,\delta}} L_{-k}=k, \quad k>0,\qquad\quad {\mathrm{deg}_{\,\delta}} L_{0}=\delta, \qquad {\mathrm{deg}_{\,\delta}} L_{1}=1.
\end{equation}
These degrees can be used to show partial triangularity of the bilinear form $M$ in other bases. 
\begin{lemma}\label{deg12-lemma}
For all $\lambda,\mu\in\mathbb Y$ and arbitrary $\Lambda_4\ne 0$, $\Lambda_3$, and $\Lambda_2$ in the definition of rank 2 Whittaker module ${\mathcal V}^{[2]}$ (see \eqref{Lambda234}), we have 
\begin{equation}\label{Mform1}
    M\lb v_\mu, v_\lambda\rb=0 \qquad \mathrm{if} \qquad {\mathrm{deg}_1}v_\mu>{\mathrm{deg}_1}v_\lambda.
\end{equation}
If $\Lambda_3=0$,  we have 
\begin{equation}\label{Mform2}
    M\lb v_\mu, v_\lambda\rb=0 \qquad \mathrm{if} \qquad {\mathrm{deg}_2}v_\mu>{\mathrm{deg}_2}v_\lambda.
\end{equation}
\end{lemma}
\begin{rmk}   
Note that the relations \eqref{Mform1} and \eqref{Mform2} do not follow from \eqref{Mform} because there are cases when $|\mu|<|\lambda|$ for which
${\mathrm{deg}_{\,\delta}}v_\mu>{\mathrm{deg}_{\,\delta}}v_\lambda$. For example, $\mu=\lb 1^2\rb$ and $\lambda=\lb 3\rb$.
\end{rmk}

\begin{rmk}   
The degrees ${\mathrm{deg}_{\,\delta}}$ alone do not imply a triangularity of the matrix associated to the form~$M$, since $M\lb v_\mu,v_\lambda\rb$ does not necessarily vanish if $\mathrm{deg}_\delta v_\mu = \mathrm{deg}_\delta v_\lambda$ with $\lambda\ne \mu$, cf. \eqref{Mform}. In other words, we can only ensure that the relevant matrix is \textit{block} lower-triangular, different blocks being labeled by ${\mathrm{deg}_{\,\delta}}$. The triangularity can however be restored  if we use in addition the previous degree \eqref{degL} inside each block.
\end{rmk}
\begin{rmk} Equations \eqref{Wh2} in the main text satisfy the restriction $\Lambda_3=0$. Accordingly, the degree $\deg_2$ appears explicitly in the selection rules \eqref{selectionrulesWH}.
\end{rmk}

For $\delta=1,2$, the space $U_{3k}$ defined by \eqref{DefUm} can be decomposed into a direct sum of two subspaces $U_{3k}=W_k^{(\delta)} \bigoplus V_k^{(\delta)}$, where
\begin{equation}
    W_k^{(\delta)}=\mathrm{span}\left(v_\lambda\in U_{3k}|\, {\mathrm{deg}_{\,\delta}}v_\lambda> k\right) , \qquad 
    V_k^{(\delta)}=\mathrm{span}\left(v_\lambda\in U_{3k}|\, {\mathrm{deg}_{\,\delta}}v_\lambda\leq k\right).
\end{equation}
Note that for $k>0$ we have a strict inclusion $V_k^{(2)}\varsubsetneq   V_k^{(1)}$.

The matrix $M_{3k}$ of bilinear form $M$ restricted to $U_{3k}=W_k^{(\delta)}\bigoplus V_k^{(\delta)}$ has the form of a $2\times 2$ block matrix
\begin{equation}\label{MMatrix_delta}
    M_{3k}^{(\delta)}=\begin{pmatrix}
        M_{WW} & M_{WV} \\
        M_{VW} & M_{VV}
    \end{pmatrix},
\end{equation}
where the indices $W$ and $V$ correspond to restrictions to $W_k^{(\delta)}$ and $V_k^{(\delta)}$, and additionally inside each block we use a descending order of basis elements with respect to $\operatorname{deg} v_\lambda =|\lambda|$. Then,  it follows from \eqref{Mform} that $M_{WW}$ and $M_{VV}$ are lower-triangular and non-degenerate matrices, while Lemma~\ref{deg12-lemma} ensures that $M_{WV}=0$. 
The matrix $M_{3k}^{(\delta)}$ is therefore lower-triangular. Using this fact, we can now formulate an enhanced version of Theorem~\ref{Wh52deg}.

\begin{theo} \label{Wh52deg12} 
The rank $\frac52$ Whittaker vector $\left|\Psi\right\rangle$ exists and is unique in the orthogonal gauge. For any $k\in\mathbb N$, its component $\left|\Psi_k\right\rangle$ belongs to the subspace $V_k^{(1)}$.
If in addition $\Lambda_3=0$, then $|\Psi_k\rangle\in V_k^{(2)}\subset  V_k^{(1)}$.
\end{theo}
\pf
Theorem~\ref{Wh52deg} already proves the existence and uniqueness of $|\Psi\rangle$ with components $|\Psi_k\rangle\in U_{3k}$. 
The only thing that needs to be adjusted is Step~2 in its proof. Namely, we will now use the lower triangular matrix $M_{3k}^{(\delta)}$ 
given by \eqref{MMatrix_delta} associated to the bilinear form $M$.
To complete the proof, it suffices to note that $v_{\lb 3^k\rb}=L_{-1}^k|J\rangle\in V_k^{(\delta)}$.
\epf 

Finally, let us note that Lemma~\ref{deg12-lemma} follows immediately from the part $(b)$ of the following lemma.
\begin{lemma} Let $\lambda,\mu\in\mathbb Y$ be arbitrary Young diagrams.
For the degree $\deg_1$, and, if $\Lambda_3=0$, also for the degree $\deg_2$, we have  
\begin{enumerate}
\item[(a)]  If $\tilde L_{2+j} v_\lambda \ne 0$ with $j\ge 1$, then
\begin{equation}
    \deg_{\,\delta} \tilde L_{2+j} v_\lambda \le \deg_{\,\delta}  v_\lambda - \deg_{\,\delta} L_{2-j}\,.
\end{equation}
\item[(b)]  If $\deg_{\,\delta} v_\mu > \deg_{\,\delta} v_\lambda$, then 
\begin{equation}
    \widetilde{\mathbf L}_\mu  v_\lambda =0.
\end{equation}
\end{enumerate}
\end{lemma}

\pf 
The first claim can be proved by straightforward case-by-case analysis of different values of $j$ for the case $\delta=1$ and (if $\Lambda_3=0$) also for $\delta=2$.
Statement $(b)$ is obtained by recursive application of  $(a)$. \epf
    
    \subsection{Rank $\frac52$ state in the rank $\frac32$ Whittaker module\label{AppA3}}
      We start by adapting the computations from Subsection \ref{subsec_Diffopapp} for the representation of Virasoro generators given by \eqref{52equationsV2}. The main steps are as follows. First, define conformal block as
    	\beq
    	\mathcal F^{(5/2)}\lb c_{_{1/2}},c_{_{3/2}},c_{_{5/2}}\rb= \bigl\langle 0 \bigr| I^{(5/2)}\bigr\rangle,
    	\eeq
    	where $L_n\bigl|I^{(5/2)}\bigr\rangle=\mathcal L_n^{(5/2)}\bigl|I^{(5/2)}\bigr\rangle$ for $n\ge 0$ and the differential operators $\mathcal L^{(5/2)}_n$ are defined by \eqref{52equationsV2}. Next, consider the Ward identities
    	\beq 
    	\mathcal L^{(5/2)}_{0}\mathcal F^{(5/2)}=0, \qquad  \mathcal L^{(5/2)}_{1}\mathcal F^{(5/2)}=0.
    	\eeq 
    	Their general solution reads
    	\beq
    	\mathcal F^{(5/2)}\lb c_{_{1/2}},c_{_{3/2}},c_{_{5/2}}\rb = \exp\left\{\frac{4c_{_{3/2}}^5}{405c_{_{5/2}}^3}-\frac{4c_{_{1/2}}c_{_{3/2}}^3}{27c_{_{5/2}}^2}
    	+\frac{2c_{_{1/2}}^2c_{_{3/2}}}{3c_{_{5/2}}}\right\} F\lb \xi\rb,
    	\eeq
    	where the irregular cross-ratio $\xi$ is given by
    	\beq \label{sc52}
    	\xi=-\frac{2i}{3} \frac{c_{_{3/2}}^{5/2}}{c_{_{5/2}}^{3/2}}\left(1-\frac{6c_{_{1/2}}c_{_{5/2}}}{c_{_{3/2}}^2}\right)^{5/4},
    	\eeq
    	and  $F\lb \xi\rb$ is an arbitrary function. It represents the part of conformal block which is not fixed by the Ward identities. 
    	
    	At the next step, we would like to embed the rank $\frac52$ irregular state into the rank $\frac32$ Whittaker module in a way similar to \eqref{irr_state_dec}. Write 
    	\beq\label{auxI52}
    	\bigl| I^{(5/2)}\bigr\rangle=f\lb c_{_{1/2}},c_{_{3/2}},c_{_{5/2}}\rb |\tilde\Psi\rangle,
    	\eeq
    	where the normalized state $|\tilde \Psi\rangle$ and the singular prefactor $f$ are required to satisfy
    	\beq\label{fcond5232}
    	|\tilde\Psi\rangle=\bigl|I^{(3/2)}\bigr\rangle+o\lb 1\rb,\qquad f^{-1}\mathcal L_n^{(5/2)}f=\mathcal L_n^{(3/2)}+o\lb 1\rb.
    	\eeq
    	as $c_{_{5/2}}\to0$. The nontrivial operators  $\mathcal L_n^{(3/2)}$ are given by 
    	\beq
    	\begin{gathered}
    	\mathcal L_3^{(3/2)}=-c_{_{3/2}}^2,\qquad \mathcal L_2^{(3/2)}=-2c_{_{1/2}}c_{_{3/2}},\\
    	\mathcal L_1^{(3/2)}=-c_{_{1/2}}^2+\frac{c_{_{3/2}}}{2}\frac{\partial}{\partial c_{_{1/2}}},\qquad
    	\mathcal L_0^{(3/2)}=
    	\frac{3 c_{_{3/2}}}{2}\frac{\partial}{\partial c_{_{3/2}}}+
    	\frac{c_{_{1/2}}}{2}\frac{\partial}{\partial c_{_{1/2}}},
    	\end{gathered}
    	\eeq
    	cf Proposition~\ref{PropLs}.
    	The singular prefactor can be chosen in a form that incorporates the as-yet undetermined conformal block:
    	\beq\label{orthogauge32}
    	f\lb c_{_{1/2}},c_{_{3/2}},c_{_{5/2}}\rb=\frac{\bigl\langle 0 \bigr| I^{(5/2)}\bigl\rangle}{\bigl\langle 0 \bigr| I^{(3/2)}\bigr\rangle}.
    	\eeq
        Indeed, projecting  \eqref{auxI52} onto the vacuum and using \eqref{orthogauge32}, one obtains $\langle 0 |\tilde \Psi\rangle=\bigl\langle 0 \bigr| I^{(3/2)}\bigr\rangle$. The above choice therefore amounts to working in the orthogonal gauge.
    	Also note that the rank $\frac32$ conformal block  $\mathcal F^{(3/2)}\lb c_{_{1/2}},c_{_{3/2}}\rb=\bigl\langle 0\bigr|I^{(3/2)}\bigr\rangle$ can be easily found from the Ward identities: 
    	\beq
    	\mathcal F^{(3/2)}\lb c_{_{1/2}},c_{_{3/2}}\rb=\exp \frac{2c_{_{1/2}}^3}{3c_{_{3/2}}}.
    	\eeq

    	The leading  singular behavior of the  function $F\lb \xi \rb$ as $c_{_{5/2}}\to 0$ is determined by the condition \eqref{fcond5232}
    	\beq
    	\ln F\lb \xi\rb=\frac{\xi^2}{45}+O\lb \xi\rb,\qquad \xi\to\infty.
    	\eeq
    	The results obtained in the case of rank 2 reduction suggest that $\ln F\lb \xi\rb$ should admit an expansion in integer powers of $\xi$. This, however, immediately implies (cf \eqref{sc52}) that  $|\tilde \Psi\rangle$ should be expanded in \textit{half-integer} powers of $c_{_{5/2}}$:
    	\beq
    	|\tilde \Psi\rangle=\bigl| I^{(3/2)}\bigr\rangle+\sum_{k=1}^{\infty} c_{_{5/2}}^{k/2}
    	\bigl|I^{(3/2)}_{k/2}\bigr\rangle, \qquad \bigl\langle 0\bigr|I^{(3/2)}_{k/2}\bigr\rangle=0.
    	\eeq
    	
    	We will now use the ideas of Subsection~\ref{SubsecAlgCon} to develop an algebraic construction of the normalized 
    	Whittaker state $|\tilde \Psi\rangle$.
    		Let $|\tilde J\rangle$ satisfy
    	\beq
    	L_{n>3}|\tilde J\rangle=0,\qquad L_{3}|\tilde J\rangle=|\tilde J\rangle,\qquad L_2|\tilde J\rangle=0,
    	\eeq 
    	and consider the corresponding Whittaker module
    	\beq
    	\mathcal V^{[3/2]}=\bigoplus_{m_{0,1}\in \mathbb Z_{\ge 0},\lambda\in\mathbb Y}\mathbb C\,\mathbb L_{-\lambda}L_0^{m_0}L_1^{m_1}|\tilde J\rangle.
    	\eeq
    	We would like to find $|\tilde\Psi\rangle\in \mathcal V^{[3/2]}$ such that
    	\beq\label{Wh52version32}
    	L_{n>5}|\tilde\Psi\rangle=0,\qquad L_5|\tilde\Psi\rangle=\varepsilon^2 |\tilde\Psi\rangle,
    	\qquad L_4|\tilde\Psi\rangle=-2\varepsilon  |\tilde\Psi\rangle,\qquad 
    	L_3|\tilde\Psi\rangle= |\tilde\Psi\rangle.
    	\eeq
    	The parameterization of eigenvalues is suggested by \eqref{52equationsV2} and corresponds to setting $c_{_{1/2}}=0$, $c_{_{3/2}}=-i$, ${c_{_{5/2}}=i\varepsilon}$ therein. The parameter $\varepsilon$  should not be confused with the one used in Subsection~\ref{SubsecAlgCon}. It is related to the parameter $\xi$ as
    	\beq
    	\xi=-\frac{2i}{3}\varepsilon^{-3/2}.
    	\eeq
    	As $\varepsilon\to0$, the relations \eqref{Wh52version32} are verified not only by $|\tilde{J} \rangle$ but also by $L_1^p|\tilde{J}\rangle$ for any $p\in\mathbb N$. This constitutes a major difference as compared to the rank $2$ embedding. We look for  $|\tilde\Psi\rangle$ in the form
    	\beq \label{psit32ser}
    	|\tilde\Psi\rangle=|\tilde J\rangle+\sum_{k=1}^{\infty} \varepsilon^{k/2}\mathbb{G}_{k/2}|\tilde J\rangle, \qquad \mathbb G_{k/2}=\sum_{m_0,m_1\in\mathbb Z_{\ge0},\lambda\in\mathbb Y}G_{\lambda;m_0,m_1}^{[k/2]}\mathbb L_{-\lambda}L_0^{m_0}L_1^{m_1}, 
    	\eeq
    	and impose the orthogonal gauge conditions
    	\beq\label{OGII}
    	G_{\emptyset;0,0}^{[k/2]}=0,\qquad k\in\mathbb N.
    	\eeq
    	
    	In contrast to the rank $2$ embedding, the algebraic relations \eqref{Wh52version32} are not sufficient for the construction of $|\tilde\Psi\rangle$. We will also need an analog of~$L_\varepsilon$. It can be found by constructing an antihomomorphism of a suitable subalgebra of $\mathsf{Vir}$ into the algebra of 1st order differential operators along the lines of Subsection~\ref{SubsectionConfBlocks52}. Consider a linear combination $ L_\varepsilon=AL_2+BL_1+CL_0$ and require that
    	\beq
    	\begin{gathered}
    		L_n\mapsto \mathcal L_n,\qquad n\in \{\varepsilon\}\cup \mathbb N_{\ge3},\\
    		\mathcal L_{n\ge 6}= 0,\qquad \mathcal L_5= \varepsilon^2,\qquad 
    		\mathcal L_4=-2\varepsilon,
    		\qquad \mathcal L_3= 1,\qquad 
    		\mathcal L_{\varepsilon}=f\lb \varepsilon\rb\frac{\partial}{\partial\varepsilon}.
    	\end{gathered}
    	\eeq
    	Computing the commutators
    	\begin{subequations}
    		\begin{align}
    			\left[L_3,L_{\varepsilon}\right]=&\,AL_5+2BL_4+3CL_3\mapsto A\varepsilon^2-4B\varepsilon+3C,\\
    			\left[L_4,L_{\varepsilon}\right]=&\,2AL_6+3BL_5+4CL_4\mapsto 3B\varepsilon^2-8C\varepsilon,\\
    			\left[L_5,L_{\varepsilon}\right]=&\,3AL_7+4BL_6+5CL_5\mapsto 5C\varepsilon^2,
    		\end{align}
    	\end{subequations}
    	and comparing them with
    	\beq
    	\left[\mathcal L_{\varepsilon},\mathcal L_3\right]=0,\qquad 
    	\left[\mathcal L_{\varepsilon},\mathcal L_4\right]=-2f\lb \varepsilon\rb,\qquad
    	\left[\mathcal L_{\varepsilon},\mathcal L_5\right]=2\varepsilon f\lb \varepsilon\rb,
    	\eeq
    	we can set 
    	\beq
    	L_\varepsilon=L_2+\varepsilon L_1+\varepsilon^2 L_0,\qquad \mathcal L_{\varepsilon}=\frac{5\varepsilon^3}{2}\frac{\partial}{\partial\varepsilon}.
    	\eeq
    	
    	Note that we can add to $\mathcal L_{\varepsilon}$ an arbitrary function of $\varepsilon$ without changing the commutation relations. This freedom corresponds to different embeddings of the rank $\frac52$ Whittaker vector into rank $\frac32$ Whittaker module. 
    	Let us define the action of $L_\varepsilon$ as
    	\beq \label{lepsil32}
    	L_\varepsilon|\tilde{\Psi}\rangle=\left(\frac{5\varepsilon^3}{2}\frac{\partial}{\partial\varepsilon}+\sum_{k=1}^{\infty}\alpha_k \varepsilon^{\frac{3k}{2}-1}\right)|\tilde{\Psi}\rangle,
    	\eeq
    	where the coefficients of the formal series in the second term can be chosen arbitrarily.
    	\begin{conj} The state $|\tilde{\Psi}\rangle$ defined by the relations \eqref{Wh52version32}, \eqref{psit32ser}--\eqref{OGII} and \eqref{lepsil32} exists
    		and is uniquely determined in terms of $\left\{\alpha_n\right\}_{n\in\mathbb N}$ and the central charge $c$. Moreover,
    		\beq\label{selectionrulesWHv2}
    		\mathbb G_{k/2}=\sum_{\substack{1\leq\ell\leq k \\ \ell=k\;\mathrm{mod}\;3}}\sum_{\substack{m_0,m_1\in\mathbb Z_{\ge0},\lambda\in\mathbb Y \\ m_1+3m_0+2|\lambda|=\ell}}G_{\lambda;m_0,m_1}^{[k/2]}\mathbb L_{-\lambda}L_0^{m_0}L_1^{m_1}.
    		\eeq
    		where the coefficients $G_{\lambda;m_0,m_1}^{[k/2]}$ depend only on $\alpha_1, \ldots, \alpha_{\left[(k+2)/3\right]}$.
    	\end{conj}

    	The $\mathrm{mod}\;3$ selection rule stems from the $\mathbb Z_3$-symmetry of the relations \eqref{Wh52version32} defining the irregular vector~$|\tilde{\Psi}\rangle$. Namely, their form is preserved by the transformation
    	\beq
    	L_n\mapsto \omega^{-n}L_n, \qquad \varepsilon\mapsto \omega^2 \varepsilon, \qquad \omega^3=1.
    	\eeq
    	The restriction $m_1+3m_0+2|\lambda|=\ell$ was observed experimentally. It can be reformulated using the following definition of degree of the Virasoro generators:
    	\beq
    	\mathrm{deg}\, L_1=\frac{1}{2}, \qquad \mathrm{deg}\, L_0=\frac{3}{2}, \qquad \mathrm{deg}\, L_{-k}=k, \qquad k\in \mathbb{N}.
    	\eeq
    	Under such assignments, the corresponding selection rule can be rewritten as $\operatorname{deg}\, \mathbb G_{k/2}\leq \frac{k}{2}$.
    	
    		The irregular state is thus fixed by a single function (formal series) that has to be determined from other considerations. A similar phenomenon was observed in a related but different context in \cite[Conjecture~3.2]{Nagoya18}.
    	The first  descendant contributions to $|\tilde{\Psi}\rangle$ are explicitly given by 
    	\begin{subequations}
    	\begin{align}
    		&\, \mathbb G_{1/2}=-2\nu L_1, \qquad \mathbb G_{1}=-\frac{2}{5}L_{-1}+\frac{20 \nu ^2+1}{10}L_1^2,\\
    	&\, \mathbb G_{3/2}=-\frac{16\nu}{15} L_0+\frac{4\nu}{5}L_{-1}L_1
    	-\frac{\nu\left(20 \nu ^2+3  \right)}{15} L_1^3,\\
    	&\,\begin{aligned}
    		\mathbb G_{2}=&\,-\frac{1}{5}L_{-2}+\frac{2}{25}L_{-1}^2+\frac{160 \nu ^2+9}{75} L_0 L_1-\frac{20 \nu ^2+1}{25} L_{-1}L_1^2
    		+\frac{400 \nu ^4+120 \nu ^2+3}{600} L_1^4\\&\,-\frac{640 \nu ^2-1125 \alpha -15 c+36}{150} L_1,
    	\end{aligned}\\
    	&\,\begin{aligned}
    		\mathbb G_{5/2}=&\,\frac{2\nu}{5} L_{-2}L_1+\frac{32\nu}{75} L_{-1}L_0-\frac{28\nu}{25} L_{-1}-\frac{2\nu\left(80 \nu ^2+13  \right)}{75}  L_0 L_1^2 +\frac{2\nu\left(20 \nu ^2+3 \right)}{75}  L_{-1}L_1^3\\
    		&\,+\frac{4\nu}{25} L_{-1}L_1^2-\frac{\nu\left(80 \nu ^4+40 \nu ^2+3 \right)}{300}  L_1^5
    		+\frac{\nu\left(2560 \nu ^2-4500 \alpha   -60c  +377  \right)}{300}  L_1^2,
    	\end{aligned}
    	\end{align}
    \end{subequations}	
    where we parameterize $\alpha_1=-2\nu$, $\alpha_2=\frac{15\alpha}{2}$.

\section{Completion of the proof of the conifold gap property}
\label{appendix:constant-term}

Here, we present the final part of the proof of Theorem \ref{thm:coni-gap},  
specifically addressing the proof of equation \eqref{eq:limit-Fg-appendix}. 
Before proceeding, let us briefly recall the proof strategy outlined in Section 4.4.3. Starting from the original spectral curve 
${\mathcal C}$ defined in equation \eqref{eq:spcurve}, 
we performed a symplectic transformation \eqref{eq:symplectic-tr-to-Weber}, 
resulting in a new spectral curve 
$\widetilde{\mathcal C}$ as defined in equation \eqref{eq:spcurve-Weber-expression}. 
Importantly, both spectral curves ${\mathcal C}$ and $\widetilde{\mathcal C}$
share the same free energy $F_g$. 
While the defining equation for 
${\mathcal C}$ has diverging coefficients as $\Lambda \to 0$, 
the symplectic transformation absorbs this divergence, 
and the Weber curve ${\mathcal C}_{\rm Web}$, which is given in \eqref{eq:Weber-curve},  
appears as the limit of $\widetilde{\mathcal C}$ as $\Lambda \to 0$.
The claim to be shown is that the limit of $F_g$ of ${\mathcal C}$ 
as $\Lambda \to 0$ exists and coincides with that of 
${\mathcal C}_{\rm Web}$. 
In this Appendix, we provide the proof of this statement.

\subsection{Review of topological recursion for the Weber curve}

First, we recall several results regarding the topological recursion for 
the Weber curve ${\mathcal C}_{\rm Web}$, which plays an essential role in the proof.
The Weber curve is known as the spectral curve for the Gaussian matrix model, 
whose correlator and free energy were studied from the viewpoint of TR by many authors; 
see \cite{Nor09, IKoT1} for example.

We can apply TR to the Weber curve \eqref{eq:spcurve-Weber-expression} 
through the Zhukovsky parametrization
\begin{equation}  \label{eq:Weber-paranetrization}
(\widetilde{X},\widetilde{Y}) = 
\left( \sqrt{\nu}\,(w + w^{-1}) ,  \frac{\sqrt{\nu}}{2}\,(w - w^{-1}) \right).
\end{equation}
The ramification points are $w = \pm 1$, 
and $\sigma(w) = w^{-1}$ gives the covering involution map. 
We denote by $W_{g,n}^{\rm Web}$ the correlators of the Weber curve. 
The first few correlators are 
\begin{equation} \label{eq:W01-02-Web}
W_{0,1}^{\rm Web}(w_1) 
= \frac{\nu (w_1^2-1)^2 }{2w_1^3} \, dw_1, \qquad
W_{0,2}^{\rm Web}(w_1, w_2) = \frac{dw_1 dw_2}{(w_1 - w_2)^2}, 
\end{equation}
\begin{equation}
W_{0,3}^{\rm Web}(w_1, w_2, w_3) = 
\frac{dw_1 dw_2 dw_3}{2\nu} \biggl( \frac{1}{(w_1-1)^2(w_2-1)^2(w_3-1)^2} 
+ \frac{1}{(w_1+1)^2(w_2+1)^2(w_3+1)^2} \biggr), 
\end{equation}
\begin{equation}
W_{1,1}^{\rm Web}(w_1) = - \frac{w_1^3 \, dw_1}{\nu (w_1^2-1)^4}.
\end{equation}
The correlators with $2g-2+n \ge 1$ are holomorphic except for the ramification points. 
Particularly, the fact that they are holomorphic at $w_i = 0, \infty$, 
which corresponds to the $\widetilde{X} = \infty$, 
will be crucial in the following discussion. 
We will also use the following alternative expression of TR 
which can be obtained by the residue theorem, similarly to \eqref{eq:alternative-TR-scaling}:
\begin{equation} \label{eq:alt-Weber}
W_{g,n}^{\rm Web}(w_1, \dots, w_n) = 
-2
\sum_{i=1}^{n}  \Res_{\widetilde{X}=\widetilde{X}_i} 
\frac{\int^{\widetilde{X}'=\widetilde{X}}_{\widetilde{X}'=\infty} {W}^{\rm Web}_{0,2}(\widetilde{X}_1, \widetilde{X}')}
{2W_{0,1}^{\rm Web}(\widetilde{X})} \, {R}^{\rm Web}_{g,n}(\widetilde{X}, \widetilde{X}_2, \dots, \widetilde{X}_n),
\end{equation}
where 
$R_{g,n}^{\rm Web}$ is 
\eqref{eq:R-gn} for the Weber curve, and the relation between the coordinates 
$w_i$ and $\widetilde{X}_i$ are given in \eqref{eq:Weber-paranetrization}; that is,
\begin{equation} \label{eq:coord-wi}
\widetilde{X}_i = \nu (w_i + w_{i}^{-1}),  \qquad i = 1,\dots, n.
\end{equation} 

As we have seen in \eqref{eq:Weber-free-energy-explicit}, 
an explicit formula for the free energy $F^{\rm Web}_g$ in terms of Bernoulli numbers is known, 
and it plays an essential role in the derivation of \eqref{eq:limit-Fg-appendix}, which will be discussed below.

\subsection{Another series expansion of $W_{g,n}$}

Here we show that the correlator $W_{g,n}$ of the spectral curve ${\mathcal C}$
has a series expansion when $\Lambda \to 0$ 
when we regard $\widetilde{X}$,  given in \eqref{eq:symplectic-tr-to-Weber}, as $t$-independent. 
We note that, since the coordinate $X$, introduced in \eqref{eq:scaling-data}, 
is related to $\widetilde{X}$ through $t$-depending functions as 
\begin{equation} \label{eq:rel-coord}
X = \gamma^{-2} x =  \frac{\tilde{e}_2 + \tilde{e}_3}{2} + \gamma^{-2} \alpha \widetilde{X}, \qquad \gamma^{-2}=(3a)^{-1}(-t)^{-1/2}.
\end{equation}
the expansion we will obtain below will be different from \eqref{eq:expansion-Wgn}. 

Let us look at the expansion of $W_{0,1}$. 
Using the series expansion \eqref{tilde-e1}--\eqref{tilde-e3} of $e_i$
and the relation \eqref{eq:rel-coord} among the coordinates, we have
\begin{align} 
\widetilde{Y} & = (1 + 16 \alpha^5  \widetilde{X} )^{1/2} 
\left( \frac{\widetilde{X}^2}{4} - \frac{(e_2-e_3)^2}{16 \alpha^2}\right)^{1/2} \notag \\
& 
=  \frac{\sqrt{\widetilde{X}^2- 4\nu}}{2} + \frac{\widetilde{X} \sqrt{\widetilde{X}^2- 4\nu}}{8} \, \Lambda^{1/2} 
- \frac{\widetilde{X}^4 - 4\nu \widetilde{X}^2 + 2 \nu^2}{64 \sqrt{\widetilde{X}^2- 4\nu} } \Lambda 
+ \frac{\widetilde{X}^5 + 56 \nu \widetilde{X}^3 - 242 \nu^2 \widetilde{X}}{256 \sqrt{\widetilde{X}^2- 4\nu}} \Lambda^{3/2}
+ O\lb \Lambda^{2} \rb.
\label{eq:alt-expansion-tildeY} 
\end{align}
Here and in what follows, the branch of $\sqrt{\widetilde{X}^2 - 4\nu}$ 
is chosen so that
\begin{equation} \label{eq:branch-weber-limit}
\frac{\sqrt{\widetilde{X}^2- 4\nu}}{2} \, d\widetilde{X} 
= \frac{\nu (w^2-1)^2 }{2w^3} \, dw
\end{equation}
holds under the change of coordinate \eqref{eq:Weber-paranetrization}.
The equality \eqref{eq:alt-expansion-tildeY} shows that 
$W_{0,1} = \widetilde{Y} d\widetilde{X}$ has the series expansion 
of the form 
\begin{equation}
W_{0,1}(\widetilde{X}) = 
\sum_{\ell \ge 0} \widetilde{W}_{0,1}^{[\ell]}(\widetilde{X}) \, \Lambda^{\ell/2} 
\end{equation}
with certain meromorphic differentials $\widetilde{W}_{0,1}^{[\ell]}$
on the Weber curve ${\mathcal C}_{\rm Web}$. 
In particular, \eqref{eq:branch-weber-limit} implies that 
the leading term $\widetilde{W}_{0,1}^{[0]}$ agrees with $W_{0,1}^{\rm Web}$ 
under the change of coordinate \eqref{eq:Weber-paranetrization}.
Even when taking into account the coordinate transformation \eqref{eq:rel-coord}, 
it should be noted that $\widetilde{W}_{0,1}^{[\ell]}(\widetilde{X})$ 
defined here and $W_{0,1}^{[\ell]}(X)$ given in \eqref{eq:expansion-Wgn}
are entirely different differentials.

It also follows from the expression \eqref{eq:Bergman-via-XY} 
of Bergman bidifferential and the above expansion \eqref{eq:alt-expansion-tildeY} 
that $W_{0,2}$ also has the series expansion of the form 
\begin{equation}
W_{0,2}(\widetilde{X}_1, \widetilde{X}_2) = 
\sum_{\ell \ge 0} 
\widetilde{W}_{0,2}^{[\ell]}(\widetilde{X}_1, \widetilde{X}_2) \, \Lambda^{\ell/2} 
\end{equation}
with certain meromorphic bidifferentials $\widetilde{W}_{0,2}^{[\ell]}$ 
on ${\mathcal C}_{\rm Web}$. 
The first few terms are 
\begin{subequations}
\begin{align}
\widetilde{W}_{0,2}^{[0]}(\widetilde{X}_1, \widetilde{X}_2) & = 
\frac{\widetilde{X}_1 \widetilde{X}_2 + \sqrt{\widetilde{X}_1^2- 4\nu} \sqrt{\widetilde{X}_2^2- 4\nu} - 4 \nu}
{2(\widetilde{X}_1 - \widetilde{X}_2)^2 \sqrt{\widetilde{X}_1^2- 4\nu} \sqrt{\widetilde{X}_2^2- 4\nu}} \, 
d\widetilde{X}_1 d\widetilde{X}_2, \\
\widetilde{W}_{0,2}^{[1]}(\widetilde{X}_1, \widetilde{X}_2) & = 0, \\
\widetilde{W}_{0,2}^{[2]}(\widetilde{X}_1, \widetilde{X}_2) & = 
\frac{\widetilde{X}_1^3 \widetilde{X}_2^3 - 4 \nu \widetilde{X}_1 \widetilde{X}_2(\widetilde{X}_1^2+\widetilde{X}_2^2) 
+ 18 \nu^2 \widetilde{X}_1 \widetilde{X}_2 + 8 \nu^3}
{64(\widetilde{X}_1^2- 4\nu)^{3/2} (\widetilde{X}_2^2- 4\nu)^{3/2}} \, 
d\widetilde{X}_1 d\widetilde{X}_2, \\
\widetilde{W}_{0,2}^{[3]}(\widetilde{X}_1, \widetilde{X}_2) & =
-\frac{(\widetilde{X}_1+\widetilde{X}_2)(\widetilde{X}_1\widetilde{X}_2-2\nu)}{128\sqrt{\widetilde{X}_1^2- 4\nu} \sqrt{\widetilde{X}_2^2- 4\nu}} \, d\widetilde{X}_1 d\widetilde{X}_2.
\end{align}
\end{subequations}
We can also confirm that $\widetilde{W}_{0,2}^{[0]} = W_{0,2}^{\rm Web}$ 
holds under the coordinate transform \eqref{eq:coord-wi}. 

More generally, we have

\begin{prop}
For each  $g \ge 0$, $n \ge 1$, 
the correlator $W_{g,n}$ of the spectral curve $\widetilde{\mathcal C}$ 
(and ${\mathcal C}$)
has the following convergent series expansion at $\Lambda = 0$:
\begin{equation}\label{eq:alt-expansion-Wgn}
W_{g,n}(\widetilde{X}_1, \dots, \widetilde{X}_n) = 
\sum_{\ell \ge 0} \widetilde{W}_{g,n}^{[\ell]}(\widetilde{X}_1, \dots, \widetilde{X}_n)
\, \Lambda^{\ell/2}.
\end{equation}
The coefficients $\widetilde{W}_{g,n}^{[\ell]}$ are meromorphic multi-differentials 
on the Weber curve ${\mathcal C}_{\rm Web}$, and independent of $\Lambda$. 
In particular, the leading term coincides with the correlator of the Weber curve 
\begin{equation} \label{eq:leading-Weber}
W_{g,n}^{\rm Web}(w_1, \dots, w_n) = 
\widetilde{W}_{g,n}^{[0]}(\widetilde{X}_1, \dots, \widetilde{X}_n)
\end{equation} 
under the change of the coordinate \eqref{eq:coord-wi}.
\end{prop}

\begin{proof} 

Since the case for $(g,n) = (0,1)$ and $(0,2)$ have already been shown above, 
we will now provide the proof for general $(g,n)$ using the induction.

Similarly to the proof of Proposition \ref{prop:expansion-Wgn}, 
it is necessary to rewrite the original TR formula \eqref{eq:top-rec} 
in order to avoid difficulties arising from the 
collision of ramification points (i.e., $e_2, e_3 \to 1/3$) in $\Lambda \to 0$.
As well as \eqref{eq:alternative-TR-scaling}, the residue theorem provides 
the following alternative expression of TR formula (c.f., \cite[Lemma B.1]{Iwaki19}):
\begin{align} 
W_{g,n}(\widetilde{X}_1, \dots, \widetilde{X}_n) 
& =  
\frac{d\widetilde{X}_1}{{\omega}_A \, \widetilde{Y}(\widetilde{X}_1)} \,  
\oint_{\widetilde{X} \in A} \frac{{R}_{g,n}(\widetilde{X}, \widetilde{X}_2, \dots, \widetilde{X}_n)}
{2W_{0,1}(\widetilde{X})}  
- 2 
\sum_{i=1}^{n}  \Res_{\widetilde{X}=\widetilde{X}_i} 
\frac{\int^{\widetilde{X}'=\widetilde{X}}_{\widetilde{X}'=\infty} {W}_{0,2}(\widetilde{X}_1, \widetilde{X}')}
{2W_{0,1}(\widetilde{X})}  
\, {R}_{g,n}(\widetilde{X}, \widetilde{X}_2, \dots, \widetilde{X}_n). 
\label{eq:alt-TR}
\end{align}
Under the induction hypothesis, ${R}_{g,n}$ has a series expansion of the form 
\begin{equation} \label{eq:ser-Rgn-Web}
{R}_{g,n}(\widetilde{X}, \widetilde{X}_2, \dots, \widetilde{X}_n)
= \sum_{\ell \ge 0}
\widetilde{R}^{[\ell]}_{g,n}(\widetilde{X}, \widetilde{X}_2, \dots, \widetilde{X}_n) \, \Lambda^{\ell/2}.
\end{equation}
The coefficients $\widetilde{R}^{[\ell]}_{g,n}$ are meromorphic quadratic differentials 
on ${\mathcal C}_{\rm Web}$ in the first variable $\widetilde{X}$. 

We note that, under the change \eqref{eq:Weber-paranetrization} of the coordinate, 
the $A$-cycle is now represented by a lift of closed curve on $\widetilde{X}$-plane 
which encircles the two ramification points $\widetilde{X} = \pm (e_2 - e_3)/2 \alpha$ 
of \eqref{eq:spcurve-Weber-expression}.
Then, we can use the same argument which we have used 
in the proof of Proposition \ref{prop:expansion-Wgn};
that is, the $A$-cycle integral can be reduced to term-wise residue calculus
at $\widetilde{X} = \infty$ 
(i.e., $w = \infty$ in the parametrization \eqref{eq:Weber-paranetrization}). 
This allows us to obtain a series expansion of $W_{g,n}$ of the form \eqref{eq:alt-expansion-Wgn}, 
and each coefficient can be explicitly computed term-by-term at any order in principle.

Next, let us show \eqref{eq:leading-Weber}. 
Looking at the leading terms of the both sides of \eqref{eq:alt-TR}, we have
\begin{align} 
\widetilde{W}^{[0]}_{g,n}(\widetilde{X}_1, \dots, \widetilde{X}_n) 
& = - \frac{d\widetilde{X}_1}{\omega^{[0]}_A \, \widetilde{Y}^{[0]}(\widetilde{X}_1)} \,  
\oint_{\widetilde{X} \in A} \frac{\widetilde{R}^{[0]}_{g,n}(\widetilde{X}, \widetilde{X}_2, \dots, \widetilde{X}_n)}
{2\widetilde{W}_{0,1}^{[0]}(\widetilde{X})}  
- 2  \, 
\sum_{i=1}^{n}  \Res_{\widetilde{X}=\widetilde{X}_i} \frac{\int^{\widetilde{X}'=\widetilde{X}}_{\widetilde{X}'=\infty} 
\widetilde{W}^{[0]}_{0,2}(\widetilde{X}_1, \widetilde{X}')}
{2\widetilde{W}_{0,1}^{[0]}(\widetilde{X})} \, \widetilde{R}^{[0]}_{g,n}(\widetilde{X}, \widetilde{X}_2, \dots, \widetilde{X}_n),
\label{eq:alternative-TR-Web}
\end{align} 
where 
$\widetilde{Y}^{[0]}(\widetilde{X}) = \frac{1}{2}\sqrt{\widetilde{X}^2 - 4\nu}$  
is the leading term of $\widetilde{Y}$ (c.f., \eqref{eq:alt-expansion-tildeY}).  
Under the induction hypothesis, 
$\widetilde{R}^{[0]}_{g,n}$ coincides with $R_{g,n}^{\rm Web}$ 
under the coordinate change \eqref{eq:coord-wi}.
Since $\widetilde{X}=\infty$ is not the ramification point of the Weber curve ${\mathcal C}_{\rm Web}$,
$W_{g,n}^{\rm Web}$ are holomorphic there. 
This implies that 
$\widetilde{R}^{[0]}_{g,n}/\widetilde{W}_{0,1}^{[0]} = R_{g,n}^{\rm Web}/W_{0,1}^{\rm Web}$ 
is holomorphic at $\widetilde{X} = \infty$, and hence, 
the first term of \eqref{eq:alternative-TR-Web} becomes trivial.
Therefore, we can conclude that the leading term $\widetilde{W}^{[0]}_{g,n}$ satisfies 
the same equation \eqref{eq:alt-Weber} as $W_{g,n}^{\rm Web}$, which implies 
the desired equality \eqref{eq:leading-Weber}. This completes the proof.
\end{proof}

\subsection{Proof of ~ $\lim_{\Lambda \to 0} F_g  = F_g^{\rm Web}$.}

Here we provide the final step in the proof of \eqref{eq:limit-Fg-appendix}. 

First, we show that the residue at $r_1$ in the definition \eqref{Fgdef} of $F_g$ 
does not contribute to the constant term in \eqref{eq:expansion-Fg}, as follows. 
Since $e_1 = x(r_1)$ is away from the collision of the ramification points $e_2, e_3$,  
no singular points of the integrant of \eqref{Fgdef} hits the residue cycle around $r_1$.
Consequently, we have 
\begin{equation}
\Res_{z=r_1} \,\Phi\lb z\rb W_{g,1}\lb z\rb 
= \Lambda^{2g-2} \Res_{X = e_1} \left( \int^{X} W^{\rm res}_{0,1}(X) \right) W^{\rm Res}_{g,1}(X) 
~\rightarrow~ 0, \qquad \Lambda \to 0,
\end{equation}
for $g \ge 2$. Here we have used \eqref{eq:scaling-degree-Wgn}. 

On the other hand, by taking the coordinate $\widetilde{X}$ 
given in \eqref{eq:Weber-paranetrization}, 
the aforementioned collision of the ramification points is resolved,
and it can be seen from \eqref{tilde-e1}--\eqref{tilde-e3} that 
the points corresponding to $e_2, e_3$ converge to the ramification points 
$\pm 2\sqrt{\nu}$ of the Weber curve ${\mathcal C}_{\rm Web}$ when $\Lambda \to 0$. 
Therefore, we have 
\begin{equation} 
\lim_{\Lambda \to 0} F_g 
= 
\frac{1}{2-2g} \left(  \Res_{\widetilde{X} = 2 \sqrt{\nu}} + \Res_{\widetilde{X} = - 2 \sqrt{\nu}} \right) 
\left(\int^{\widetilde{X}} \widetilde{W}^{[0]}_{0,1}(\widetilde{X}) \right) \widetilde{W}_{g,1}^{[0]}(\widetilde{X}).
\end{equation}
Then, the result \eqref{eq:leading-Weber} shows that 
the right hand side of the above equality is identical to $F_{g}^{\rm Web}$. 
Thus we have obtained the desired result \eqref{eq:limit-Fg-appendix}, 
and this also completes the proof of Theorem \ref{thm:coni-gap}. 

\section{Verification of a resurgence formula}
\label{appendix:conjecture-resurgence}

In recent years, the study of resurgence properties of the free energy and partition function 
in topological string theory has advanced significantly. 
Recently, Gu--Mari\~no discovered a conjectural formula describing 
the alien derivatives of the free energy in \cite{gm22} through 
the non-perturbative analysis of the holomorphic anomaly equation, based on earlier works \cite{cesv1, cesv2}. 
Since the TR free energy also satisfies the holomorphic anomaly equation, 
it is expected to fulfill the above conjectural formula.  
Indeed, in \cite{IM24}, it was derived based on heuristic observations of 
the monodromy preserving deformations of the quantum curve 
that the TR free energy considered in this paper satisfies the conjectural formula of \cite{gm22}. 

In the following, we review this conjectural formula and compare it with known studies on 
the Stokes phenomenon for Painlev\'e transcendents. 
In conclusion, it can be confirmed that the formula conjectured in \cite{gm22, IM24} 
for the free energy $F^{\rm deg}(t; \hbar) = \sum_{g \ge 0} \hbar^{2g-2} F^{\rm deg}_g(t)$ 
defined by a degenerate elliptic curve ${\mathcal C}_{\rm deg}$ 
is consistent with the connection formula for the tri-tronqu\'ee solution of the P$_{\rm I}$ 
obtained by Kapaev, Kitaev, and others. The purpose of this section is to introduce this fact.

Unfortunately, as of now, there is no mathematically rigorous proof that 
the free energy is Borel summable and resurgent. 
Below, we will proceed under the assumption that TR partition function and free energy are 
Borel summable and resurgent, and compare it with the Stokes phenomenon of the Painlev\'e equations.
See \cite{Sauzin14} for the foundation of resurgence theory.

\subsection{Conjectural resurgence formula} 

Let 
\begin{equation}
\widehat{F}(t, \nu;\hbar) = \sum_{g \ge 2} \hbar^{2g-2} F_g(t, \nu) 
\quad \Bigl( = {F}(t, \nu;\hbar) - \hbar^{-2} F_0(t, \nu) - F_1(t,\nu) \Bigr)
\end{equation}
be the ``stable part'' of the TR free energy of the spectral curve ${\mathcal C}$. 
The conjecture of Drukker--Mari\~no--Putrov \cite{dmp-np} claims that 
the singularities of the Borel transform 
\begin{equation}
{\mathcal B}\widehat{F}(t, \nu;\zeta) = \sum_{g \ge 2} \frac{F_g(t,\nu)}{(2g-3)!} \, \zeta^{2g-3}
\end{equation}
are contained in the lattice 
\begin{equation}
{\mathcal P} = \left\{ m \oint_A ydx + n \oint_B ydx ~;~ m,n \in \mathbb Z \right\} ~ \subset ~ {\mathbb C}_\zeta
\end{equation}
of periods of $W_{0,1} = y dx$ on the complex $\zeta$-plane, which is called the Borel plane. 
These singularities are called the Borel singularities. 
We also assume that ${\mathcal B}\widehat{F}(t, \nu;\zeta)$ has analytic continuation along 
any path in ${\mathbb C}_\zeta \setminus {\mathcal P}$, 
and the Borel sum 
\begin{equation} \label{eq:Borel-sum}
{\mathcal S}F(t,\nu; \hbar) = \hbar^{-2} F_0(t, \nu) + F_1(t,\nu) + 
\int^{+\infty}_{0} e^{-\zeta/\hbar} {\mathcal B} \widehat{F}(t, \nu;\zeta) \, d\zeta
\end{equation}
of the free energy is well-defined if the integration contour never hits the points in ${\mathcal P}$. 

Now, let us focus on the contribution of integer multiples of the $B$-period 
\begin{equation} \label{eq:Borel-B-period}
\phi(t,\nu) = \oint_{B} y \, dx
\end{equation}
to the Borel sum.
We assume that our parameters $t, \nu$ are chosen so that the above $B$-period is real and positive.
As in \cite[\S 2.3]{IM24}, it is possible to realize the situation by choosing 
$t \in {\mathbb R}_{<0}$ with $|t|$ being sufficiently large, 
and an appropriate $A$-cycle with gives $\nu \in {\mathbb R}_{>0}$ by \eqref{eq:nu}. 
Actually, the choice of cycles in \cite[\S 2.3]{IM24} differs from that in Section \ref{subsec:rescale} of this paper. 
However, through similar calculations, it can be seen that the $B$-period exhibits the following asymptotic behavior
\begin{equation}
\phi(t,\nu) =\frac{1}{\Lambda} \left(  \frac{16}{15} 
+ \nu \Lambda \left(-1 + \log \Bigl(\frac{\Lambda \nu}{64}\Bigr) \right)
+ \frac{141 \nu^2 \Lambda^2}{64} + \frac{7717 \nu^3 \Lambda^3}{2048} 
+ \frac{2663129 \nu^4 \Lambda^4}{262144} 
+ O(\Lambda^5)
\right)
\end{equation}
as $t \to - \infty$.
This asymptotic behavior is obtained when choosing $a = 1/\sqrt{6}$ 
and $\sqrt{3a} = (3/2)^{1/4}$ in \eqref{eq:scaling-data}, where we have $\Lambda = 4/(3s)$ 
with $s = 24^{1/4}(-t)^{5/4}$. 
Let us consider this choice in what follows. 

In this situation, the Borel sum is expected to be ill-defined 
since the Borel singularities $n \phi(t,\nu)$ lies 
on the integration contour in \eqref{eq:Borel-sum}.  
In such a case, one can consider the lateral Borel summation ${\mathcal S}_{\pm}$ 
which assigns the same Laplace integral \eqref{eq:Borel-sum} but the path of integration 
is rotated by angle $\pm \epsilon$ with a sufficiently small $\epsilon > 0$. 
These lateral Borel sums ${\mathcal S}_{+}F$ and ${\mathcal S}_{-}F$ differ 
by exponentially small terms due to the Borel singularities on the positive real axis, 
and this is understood as the Stokes phenomenon.

Now we can formulate the conjecture of \cite{IM24}. 
It reads the following explicit formula for the action of the Stokes automorphism 
${\mathfrak S}$ $(= ``{\mathcal S}_{-}^{-1} \circ {\mathcal S}_{+}")$ 
on the partition function associated with the Borel singularity \eqref{eq:Borel-B-period}: 
\begin{equation} \label{eq:Stokes-rel}
{\mathfrak S}F(t,\nu;\hbar) = \sum_{n=0}^{\infty} F^{(n)}(t,\nu;\hbar), 
\end{equation}
where the right hand side is 
obtained through the action of the following differential-difference operator with respect to $\nu$: 
\begin{equation} \label{eq:trans-series}
\exp\left( \sum_{n=0}^{\infty} F^{(n)}(t,\nu;\hbar) \right) 
= \exp\left( \frac{1}{2 \pi i} {\rm Li}_2 (e^{- \hbar \partial_\nu})  
- \frac{\hbar \partial_\nu}{2 \pi i} \log(1 - e^{- \hbar \partial_\nu})  \right) Z(t,\nu; \hbar). 
\end{equation}
The formal series $F^{(n)}$ is called the $n$-instanton amplitude.  
$F^{(0)}$ is identical to the original free energy $F$, 
and the 1-instanton amplitude is explicitly given as follows
(which was conjectured earlier in \cite{gm22}):  
\begin{equation} \label{eq:1-inst}
F^{(1)}(t, \nu;\hbar) 
= \frac{1}{2 \pi i} \left( 1 + \hbar
\frac{\partial F}{\partial \nu}(t, \nu - \hbar; \hbar) \right) 
\exp\Bigl(  F(t,\nu - \hbar; \hbar) - F(t,\nu; \hbar) \Bigr).
\end{equation}
The variation formula $\partial_\nu F_0 = \phi$
implies that the right hand side of \eqref{eq:1-inst} is a formal power series with the exponential factor $e^{- \phi(t,\nu)/\hbar}$. 
This makes the right hand side of \eqref{eq:trans-series} a well-defined trans-series.

\subsection{Non-linear Stokes phenomenon for tri-tronque\'e solution of P$_{\rm I}$} 

Finally, let us compare the formula \eqref{eq:1-inst} with the known results of Stokes phenomenon 
for the $0$-parameter solution of P$_{\rm I}$. 
The behavior of the solutions of P$_{\rm I}$ that exhibit Stokes phenomena on 
the negative real axis of the $t$-plane is discussed, for instance in \cite{Kapaev2004}, 
where the exponentially small terms arising from the Stokes phenomena are explicitly described. 

As we have seen in \S \ref{subsec_warmup}, the $0$-parameter solution $q(t;\hbar)$ 
is obtained from TR free energy $F^{\rm deg}(t;\hbar)$ of the degenerate elliptic curve ${\mathcal C}_{\rm deg}$
by taking $t$-derivative. 
On the other hand, \eqref{eq:Fgdeg-as-subleading} implies that the Borel sum of the $0$-parameter solution 
can be obtained as
\begin{equation} \label{eq:stokes-auto-PIsol}
{\mathcal S}_{\pm}q(t;\hbar) 
= {\mathcal S}_{\pm} \left[ - \hbar^2 \frac{d^2}{dt^2} F^{\rm deg}(t;\hbar) \right]
= \lim_{\nu \to +0} {\mathcal S}_{\pm} \left[ - \hbar^2 \frac{d^2}{dt^2} F(t,\nu; \hbar) \right]
\end{equation}
Therefore, the difference of ${\mathcal S}_{\pm}q(t;\hbar)$, 
or the action of Stokes automorphism on the 0-parameter solution, 
should be obtained from \eqref{eq:Stokes-rel} by taking $t$-derivative and setting $\nu \to +0$.

Here, it is important to note that due to the conifold gap property \eqref{conigap} 
(or \eqref{eq:expansion-Fg}), 
the evaluation $\nu = 0$ cannot be made at the level of the formal series. 
The constant term $\kappa_{g}/ \nu^{2g-2}$ of $F_g$ has singularities at $\nu=0$, 
but these singularities are resolved by taking the Borel sum. 
In fact, the following equation holds (cf. \eqref{BarnesGas}):
\begin{equation}
{\mathcal S} \left( \sum_{g \ge 2} \hbar^{2g-2} \frac{B_{2g}}{4g(g-1) \nu^{2g-2}} \right) 
=  \log G\left( 1 + \frac{\nu}{\hbar} \right) 
- \Bigl(\frac{\nu^2}{2 \hbar^2} - \frac{1}{12} \Bigr) \log \frac{\nu}{\hbar}
+ \frac{3\nu^2}{4\hbar^2}
- \frac{\nu}{2 \hbar} \log 2\pi
- \frac{1}{12} + \log A.
\end{equation}
where $A$ is the Glaisher's constant.
This is valid at when $|\arg \nu| < \pi$, and hence, 
the formula is valid for our current situation $\nu \in {\mathbb R}_{>0}$.
Using this expression, we can verify that the aforementioned singularity at $\nu=0$ cancels with the other terms, 
and we can evaluate the behavior of \eqref{eq:trans-series} at $\nu \to +0$. 
The resulting formula takes the form
\begin{equation}
{\mathfrak S}q(t;\hbar) = \sum_{k \ge 0} q^{(k)}(t; \hbar), 
\end{equation}
where the series $q^{(0)}(t;\hbar)$ is identical to the original 0-parameter solution $q(t;\hbar)$, 
and \eqref{eq:1-inst}--\eqref{eq:stokes-auto-PIsol} implies that the 1-instanton amplitude $q^{(1)}$ reads
\begin{equation} \label{eq:quasi-linear-Stokes-PI}
q^{(1)}(t;\hbar) = \exp\left( - \frac{2^{{11}/{4}} \, 3^{{1}/{4}} }{5 \hbar}(-t)^{{5}/{4}} \right)
\frac{i \hbar^{1/2}}{\pi^{1/2}} \,2^{-11/8} \, 3^{-1/8}\, (-t)^{- {1}/{8}} 
\Bigl( 1 + O(\Lambda) \Bigr).
\end{equation}
Here we have used the series expansion \eqref{eq:expansion-Fg} (and explicit values of some of $F_{g}^{[k]}$)
to compute the asymptotic behavior.
A remarkable observation here is that, the leading term of 
the right hand side of \eqref{eq:quasi-linear-Stokes-PI} 
precisely agrees with the connection formula of the tritronqu\'ee solutions of \PIeq 
(e.g., \cite[eq. (2.67)]{Kapaev2004}), which was found in 
\cite{Kapaev88, JK1992, Takei, CCH15} etc. by various methods.
We can also verify that the next few terms also agree with 
series expansion of trans-series solution of \PIeq which can be found in \cite{GIKM12} for example. 
This fact supports the validity of conjectures of \cite{gm22, IM24}.

The Stokes phenomenon for the general solution of P$_{\rm I}$, 
which has a representation different from the Fourier series type discussed in this paper, 
is addressed in \cite{BSSV23, Takei2000}. 
The connection formula is derived by using the fact that P$_{\rm I}$ describes isomonodromic deformations. 
Since the conjectured formula \eqref{eq:trans-series} above is also a consequence of monodromy invariance, 
we expect that performing an analysis similar to the one conducted here for general values of 
$\nu$ and $\rho$ could lead to the derivation of results such as \cite[(7.73)]{BSSV23}.


\end{document}